\documentclass[12pt,onecolumn,twoside]{IEEEtran}

\usepackage{calc}

\usepackage{multirow}
\usepackage{cite}
\usepackage{graphicx,subfigure}
\usepackage{psfrag}
\usepackage{amsmath,amssymb,amsthm}
\usepackage{color}
\usepackage{tikz} 
\usepackage{bm}
\usepackage{bbm}
\usepackage{verbatim}

\usepackage{cases}

\DeclareMathOperator*{\defeq}{\triangleq}

\newtheorem{theorem}{Theorem}
\newtheorem{corollary}{Corollary}

\newtheorem{lemma}{Lemma}

\newcommand{\bit}{\begin{itemize}}
\newcommand{\eit}{\end{itemize}}

\newcommand{\bc}{\begin{center}}
\newcommand{\ec}{\end{center}}

\newcommand{\ba}{\begin{array}}
\newcommand{\ea}{\end{array}}

\newcommand{\beq}{\begin{equation}}
\newcommand{\eeq}{\end{equation}}

\newcommand{\beqn}{\begin{equation*}}
\newcommand{\eeqn}{\end{equation*}}

\newcommand{\bean}{\begin{eqnarray*}}
\newcommand{\eean}{\end{eqnarray*}}
\newcommand{\bea}{\begin{eqnarray}}
\newcommand{\eea}{\end{eqnarray}}

\def\E{\mathbb{E}}

\newcommand{\Lc}{{\mathcal L}}

\newcommand{\Nc}{{\mathcal N}}

\newcommand{\Rc}{{\mathcal R}}
\newcommand{\Sc}{{\mathcal S}}

\newcommand{\Zc}{{\mathcal Z}}
\newcommand{\dsum}{d^{*}_{\text{sum}}}
\newcommand{\dso}{d^{*}}
\newcommand{\dco}{d_c^{*}}

\newcommand{\non}{\nonumber}

\newcommand{\Hen}{\mathbb{H}}
\newcommand{\hen}{\mathrm{h}}
\newcommand{\Imu}{\mathbb{I}}

\newcommand{\bln}{n}

\newcommand{\Ho}{\mathcal{H}_{\text{out}}}

\DeclareMathOperator*{\argmax}{arg\,max}

\IEEEoverridecommandlockouts

\begin{document}
\sloppy

%\title{Optimal Secure GDoF of Symmetric Gaussian  Wiretap Channel with a Helper}
%\title{Adding Common Randomness Removes the Secrecy Constraints in Wireless Networks}
\title{How to Remove the Secrecy Constraints in Wireless Networks?}
\title{How to Break the Limits of Secrecy Constraints in Communication Networks?}
\title{Adding Common Randomness at the Transmitters Can Remove the Secrecy Constraints in Communication Networks}
\title{Adding Common Randomness  Can Remove  the Secrecy Constraints in Communication Networks}
\author{Fan Li and Jinyuan Chen  
\thanks{Fan Li and Jinyuan Chen is with Louisiana Tech University, Department of Electrical Engineering, Ruston, USA (emails: fli005@latech.edu, jinyuan@latech.edu).  This work will be presented at ISIT2019.}
%Fan Li and Jinyuan Chen is with Louisiana Tech University, Department of Electrical Engineering, Ruston, USA (emails: fli005@latech.edu, jinyuan@latech.edu).  This work will be presented in part at the 2019 IEEE International Symposium on Information Theory. 
%This work was partly supported by Louisiana Board of Regents Support Fund (BoRSF) Research Competitiveness Subprogram (RCS) under grant 32-4121-40336.} 
}

\maketitle
\pagestyle{headings}

%%%%%%%%%%%%%%%%%%%%%%%%%%%%%%%%%%%%%%%%%%%%%
\begin{abstract}
  
In communication networks secrecy constraints \emph{usually} incur an extra limit in capacity or generalized degrees-of-freedom (GDoF), in the sense that  a penalty in capacity or GDoF is  incurred due to the secrecy constraints. 
Over the past decades a significant amount of effort  has been made by the researchers to understand the  limits of secrecy constraints in communication networks.
In this work, we focus on how to remove the secrecy constraints in communication networks, i.e., how to remove the GDoF penalty due to  secrecy constraints. 
We begin with three basic settings:  a two-user symmetric Gaussian interference channel with confidential messages,  a symmetric Gaussian wiretap channel with a helper,   and a two-user  symmetric Gaussian multiple access wiretap channel. 
Interestingly,  in this work we show that adding common randomness at the transmitters can \emph{totally} remove the penalty in  GDoF or GDoF region of the three settings considered here. The results reveal that adding common randomness at the transmitters is a powerful way to remove the secrecy constraints in communication networks in terms of GDoF performance. Common randomness can be generated offline.  
The role of the common randomness is to jam the  information signal at the eavesdroppers, without causing too much interference at the legitimate receivers.  
To accomplish this role, a new method of Markov chain-based interference neutralization is proposed in the achievability schemes utilizing common randomness. 
From the practical point of view,  we hope to use less  common randomness to remove  secrecy constraints in terms of GDoF performance. 
With this motivation, for most of the cases we characterize  the minimal  GDoF of  common randomness to remove  secrecy constraints, based on our derived converses and achievability.

\end{abstract}

\section{Introduction}

For the secure communications  with secrecy constraints, the confidential messages need to be transmitted reliably to the legitimate receiver(s), without leaking the confidential information to the  eavesdroppers (cf.~\cite{Shannon:49, Wyner:75}). 
In  communication networks  secrecy constraints \emph{usually} impose an extra limit in capacity or generalized degrees-of-freedom (GDoF), in the sense that  a penalty in capacity or GDoF is  incurred due to secrecy constraints (cf.~\cite{Wyner:75, LMSY:08, LBPSV:09,  KGLP:11, HKY:13, XU:14, XU:15, GTJ:15, MM:16, MU:16, ChenIC:18}).  
Since Shannon's work  of \cite{Shannon:49} in 1949,  a significant amount of effort  has been made by the researchers to understand the  limits of secrecy constraints in communication networks  (cf.~\cite{ Wyner:75, YTL:08, LMSY:08, LBPSV:09,  LLPS:10, KGLP:11,  LLP:11, HKY:13, XU:14, XU:15, GTJ:15, GJ:15, MM:16, MU:16, ChenIC:18, TekinYener:08d, KG:15, CG:18arxiv, BSP15,    LYT:08, LP:08, FW:16,   MXU:17} and the references therein). 
In this work, we focus on how to remove the secrecy constraints in communication networks, i.e., how to remove the GDoF penalty due to  secrecy constraints.

In this work we consider three basic settings:  a two-user symmetric Gaussian interference channel with secrecy constraints,  a symmetric Gaussian wiretap channel with a helper,   and a two-user  symmetric Gaussian multiple access wiretap channel. 
Interestingly, we show that adding common randomness at the transmitters can remove the secrecy constraints in these three settings, i.e., it  can \emph{totally} remove the penalty in  GDoF or GDoF region of the three settings. 
Let us take a two-user symmetric Gaussian interference channel as an example. 
For this interference channel \emph{without} secrecy constraints, the  GDoF is a  ``W'' curve (see Fig.~\ref{fig:ICrGDoF} and \cite{ETW:08}).  If  secrecy constraints are imposed on this channel, 
then the secure GDoF is  significantly reduced,  compared to the original ``W'' curve (see~Fig.~\ref{fig:ICrGDoF} and \cite{ChenIC:18}). It implies that  a GDoF penalty is incurred due to  secrecy constraints.  Interestingly we show in this work  that  adding common randomness at the transmitters can totally remove the GDoF penalty due to  secrecy constraints (see Fig.~\ref{fig:ICrGDoF}).  
The results reveal that adding common randomness at the transmitters is a constructive way to remove the secrecy constraints in terms of GDoF performance in communication networks.

\begin{figure}[t!]
\centering
\includegraphics[width=7cm]{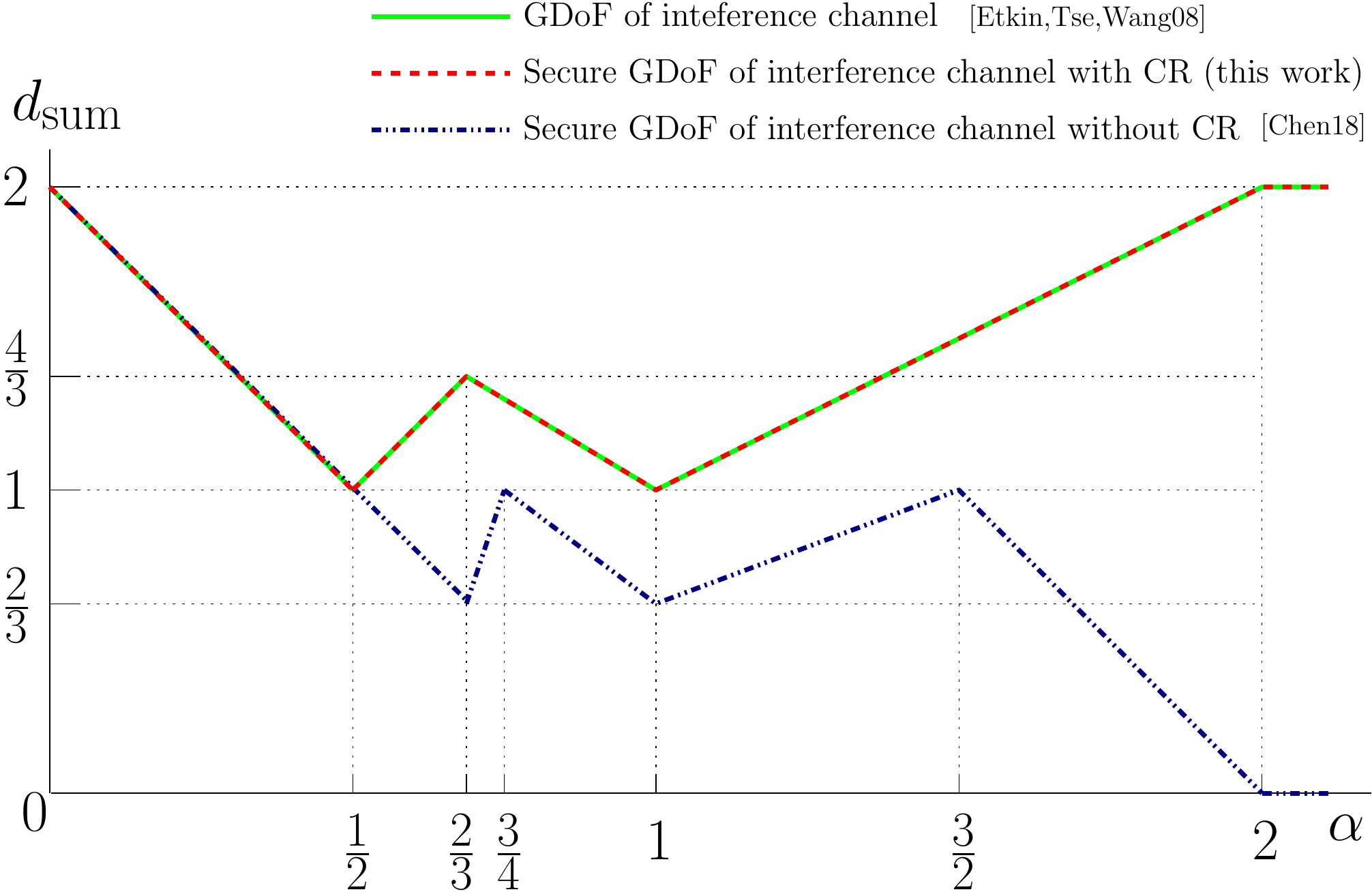}
 \vspace{-.05 in}
\caption{The optimal secure sum GDoF vs. $\alpha$, for  two-user symmetric Gaussian interference channels without and with common randomness (CR), where $\alpha$ is a channel parameter indicating the interference-to-signal ratio.} 
\label{fig:ICrGDoF}
\vspace{-15pt}
\end{figure}

The role of the common randomness is to jam the  information signal at the eavesdroppers, without causing too much interference at the legitimate receivers.   By jamming the  information signal at the eavesdroppers with common randomness, we seek to remove the penalty in GDoF.  However, the jamming signal generated from the common randomness needs to be designed carefully so that it must not create too much interference at the legitimate receivers.  Otherwise, the interference will incur a new penalty in GDoF.
To accomplish the role of the common randomness, a new method of Markov chain-based interference neutralization is proposed in the achievability schemes. 
The idea of the Markov chain-based interference neutralization method is given as follows: the common randomness  is used to generate a certain number of  signals with specific  directions and powers;   one signal is used to jam the information signal at an eavesdropper but it will create an interference at a legitimate receiver; this interference will be neutralized by another signal generated from the same common randomness; the added signal also creates another interference but will be neutralized by the next generated signal; this process repeats until the residual interference is under the noise level.  Since  one signal is used to neutralize the previous signal and will be  neutralized by the next signal, it forms a  Markov chain for this  interference neutralization process.

Common randomness can be generated offline.  
From the practical point of view,  we expect to use less  common randomness to remove  secrecy constraints,  in terms of GDoF performance. 
With this motivation,   we also characterize  the minimal  GDoF of the common randomness to remove the  secrecy constraints for most of the cases, based on our derived converses and achievability.

In terms of the organization of this work, section~\ref{sec:system} describes the system models and section~\ref{sec:mainresult} provides the  main results. The converse  is described  in Section~\ref{sec:converse}.
The achievability  is  provided  in  Sections~\ref{sec:CJGauIC}-\ref{sec:CJGauwiretap} and some of the appendices, while  a scheme example is described in  Section~\ref{sec:schemeexample}.
The work is concluded in Section~\ref{sec:concl}.
Regarding the notations,  $\Imu(\bullet)$,  $\Hen(\bullet)$  and $\hen(\bullet)$   denote the mutual information,  entropy, and differential entropy,  respectively.  
 The notations of $\Zc^+$, $\Rc$  and $\Nc$  denote the sets of positive integers, real numbers, and nonnegative integers, respectively.   
We define that $(\bullet)^+= \max\{\bullet, 0\}$.  We consider all the logarithms with base~$2$.  The notation of  $f(a)=o(g(a))$ implies that $\lim_{a \to \infty} f(a)/g(a) =0$.

\section{The three system models  \label{sec:system} }

For this work we focus on three settings:  a two-user interference channel with secrecy constraints,  a wiretap channel with a helper,   and a two-user  multiple access wiretap channel. 
These three settings share a common channel input-output  relationship,  given as 
\begin{align}
&y_{1} (t) =  \sqrt{P^{\alpha_{11}}} h_{11} x_{1} (t) +   \sqrt{P^{\alpha_{12}}} h_{12} x_{2}(t)  +z_{1} (t), \label{eq:ICchannelGen1}  \\
&y_{2}(t)  = \sqrt{P^{\alpha_{21}}} h_{21} x_{1}(t)  +   \sqrt{P^{\alpha_{22}}} h_{22} x_{2}(t)  +z_{2}(t),  \quad t \in \{1,2, 3, \cdots, n\}    \label{eq:ICchannelGen2} 
\end{align}
 where $x_{\ell}(t) $ represents the transmitted signal of transmitter~$\ell$ at time $t$, with a normalized power constraint $\E |x_{\ell}(t)|^2 \leq 1$;  $y_{k}(t) $ is the  signal received at receiver~$k$; and $z_{k}(t) \sim \mathcal{N}(0, 1)$ is the additive white Gaussian noise, for $k, \ell  \in \{1,2\}$.  The term $\sqrt{P^{\alpha_{k\ell}}} h_{k\ell}$ captures the channel gain between receiver~$k$ and transmitter~$\ell$, where $h_{k\ell} \in (1, 2]$ denotes  the  channel coefficient.   The exponent $\alpha_{k\ell}$ represents the \emph{link strength}  for the channel between receiver~$k$ and transmitter~$\ell$.   The parameter  $P \geq 1$ reflects the base of link strength of all the links.  
Note that  $ \sqrt{P^{\alpha_{k\ell}}} h_{k\ell}$ can represent any real channel gain  bigger or equal to $1$. Thus,   the  above model in \eqref{eq:ICchannelGen1} and \eqref{eq:ICchannelGen2} is able to describe  the general channels,  in the sense of  secure capacity approximation.
The channel parameters $\{h_{k\ell}, \alpha_{k\ell} \}_{k\ell}$ are assumed to be available at all the nodes.
 In this work we focus on the \emph{symmetric} case such that \[\alpha_{11} = \alpha_{22} =1, \  \alpha_{12}= \alpha_{21}=\alpha, \  \alpha >0.\]  
The three settings considered here are different, mainly on the number of confidential messages, the intended receivers of the messages, and the secrecy constraints. 
In what follows, we will present the details of  three settings.

\subsection{Interference channel with secrecy constraints (IC-SC)} \label{sec:sysICtwouser}

In the setting of interference channel, transmitter~$\ell$ intends to send the confidential message $w_{\ell}$ to receiver~$\ell$ using $n$ channel uses, where the message  $w_{\ell}$ is independently and uniformly chosen from a set $\mathcal{W}_{\ell} \defeq  \{1, 2, 3,  \cdots, 2^{nR_{\ell}}\}$, for $\ell\in\{1, 2\}$. To transmit $w_{\ell}$,  a   function 
  \[f_{\ell}: \mathcal{W}_{\ell} \times  \mathcal{W}_{c}   \to   \mathcal{R} ^{\bln}\] 
 is used to map  $w_{\ell} \in \mathcal{W}_{\ell}$  to the signal  $ x_{\ell}^{\bln}  = f_{\ell}( w_{\ell}, w_{c})   \in  \mathcal{R} ^{\bln}$, where $w_{c} \in  \mathcal{W}_{c} $ denotes the \emph{common randomness} that is available at both transmitters but not at the receivers.  
We assume that $w_{c}$ is uniformly and independently chosen from a set $\mathcal{W}_{c} \defeq  \{1, 2,  \cdots, 2^{nR_c}\}$.  In our setting,  $w_1, w_2$ and $w_{c}$ are assumed to be mutually independent. 
The  rate tuple $(R_1(P, \alpha), R_2 (P, \alpha), R_c (P, \alpha))$ is said to be achievable if there exists a sequence of  $n$-length codes such that each receiver can decode its desired message reliably, that is, \[ \text{Pr}[ \hat{w_k} \neq w_{k}  ]  \leq \epsilon,  \quad  \forall k\in \{1,2\} \]  for  any $\epsilon>0$,  and the transmission of the messages is secure, that is, \[ \Imu(w_1; y_{2}^{\bln})  \leq  \bln \epsilon \quad  \text{and} \quad  \Imu(w_2; y_{1}^{\bln})   \leq  \bln \epsilon\] (known as weak secrecy constraints), as $n$ goes large.  
The secure capacity region $\bar{C} (P, \alpha)$ represents the collection of  all the achievable  rate tuples $(R_1 (P, \alpha), R_2 (P, \alpha), R_c (P, \alpha))$.
The  secure GDoF  region $\bar{\mathcal{D}} ( \alpha)$ is defined as 
 \begin{align}
\bar{\mathcal{D}} ( \alpha)   \defeq & \Big\{ (d_1, d_2, d_c) :         \exists  \bigl( R_1 (P, \alpha), R_2 (P, \alpha), R_c (P, \alpha) \bigr) \in \bar{C}(P, \alpha) \non\\
 & \quad s. t. \quad  d_c = \lim_{P \to \infty}   \frac{  R_c (P, \alpha)}{ \frac{1}{2} \log P},   \  d_k = \lim_{P \to \infty}   \frac{  R_k (P, \alpha)}{ \frac{1}{2} \log P}, \  \forall k \in \{1,2\}    \Big\} .  \non
  \end{align}
The secure GDoF region  $\mathcal{D} ( d_c, \alpha)$ is defined as 
 \[\mathcal{D} ( d_c, \alpha)\defeq  \{  (d_1, d_2 ) :         \exists (d_1, d_2, d_c) \in   \bar{\mathcal{D}} ( \alpha) \}\]
 which is a function of $d_c$ and $\alpha$.
 The  secure sum GDoF is then defined  as \[ d_{\text{sum}}   (d_c, \alpha)\defeq     \max_{d_1, d_2: (d_1, d_2) \in \mathcal{D} ( d_c, \alpha) }    d_1+ d_2.\]
For this setting we are interested in  the maximal (optimal) secure sum GDoF defined as \[\dsum (\alpha) \defeq      \max_{d_c: d_c \geq 0}   d_{\text{sum}}   (d_c, \alpha).\]  
 We are also interested in the minimal (optimal) GDoF of the common randomness to achieve the maximal secure sum GDoF, defined as
\[\dco (\alpha) \defeq        \min_{d_c: \  d_{\text{sum}}   (d_c, \alpha) =  \dsum (\alpha) }    d_c .\] 
Note that degrees-of-freedom (DoF) can be treated as a specific point of GDoF by considering  $\alpha_{12}= \alpha_{21}= \alpha_{22}= \alpha_{11}=1$.
 
\subsection{The wiretap channel with a helper (WTH)}  \label{sec:syswchr}

In the setting of wiretap channel with a helper, transmitter~1 wishes to send the confidential message $w_1$  to receiver~1.  This setting is slightly different from the  previous interference channel setting, as transmitter~2 will just act as a helper without sending any message in this setting ($w_2$ can be set as empty).  
For transmitter~1, the mapping function $f_{1}$  is similar as that in the interference channel  described in Section~\ref{sec:sysICtwouser}. 
For transmitter~2  (helper), a  function  $f_{2}:   \mathcal{W}_{c}   \to   \mathcal{R} ^{\bln}$  maps  $ w_{c} \in  \mathcal{W}_{c} $  to the signal  $ x_{2}^{\bln}  = f_{2}(w_{c})   \in  \mathcal{R} ^{\bln}$,  where $ w_{c} \in  \mathcal{W}_{c} $ denotes the common randomness that is available at both  transmitters but not at the receivers. 
As before, we assume that $w_{c}$ is uniformly and independently chosen from a set $\mathcal{W}_{c} = \{1, 2, 3, \cdots, 2^{nR_c}\}$ and $w_{1}$ and $w_{c}$ are mutually independent.
 A  rate pair $(R_1(P, \alpha), R_c (P, \alpha))$ is said to be achievable if there exists a sequence of $n$-length codes such that   receiver~1 can reliably decode its desired message $w_1$ and the transmission of the message is secure such that   $\Imu(w_1; y_{2}^{\bln})   \leq  \bln \epsilon$, for  any $\epsilon >0$ as $n$ goes large.
 The  secure capacity region $\bar{C} (P, \alpha)$ denotes the collection of all achievable secure rate pairs $(R_1 (P, \alpha), R_c (P, \alpha))$. A   secure GDoF  region  is defined as 
\[\bar{\mathcal{D}} ( \alpha)  \defeq \Big\{ (d, d_c) :         \exists  \bigl( R_1 (P, \alpha), R_c (P, \alpha) \bigr)  \in  \bar{C}(P, \alpha ), \ s. t. \  d_c =  \lim_{P \to \infty}   \frac{  R_c (P, \alpha)}{ \frac{1}{2} \log P},   d = \lim_{P \to \infty}   \frac{  R_1 (P, \alpha)}{ \frac{1}{2} \log P}     \Big\} 
. \]
We are interested in the  maximal (optimal) secure GDoF defined as  
\[ \dso (\alpha) \defeq     \max_{d, d_c:  (d, d_c) \in  \bar{\mathcal{D}} ( \alpha) }   d.\] 
We are also interested in the minimum (optimal) GDoF of the common randomness to achieve the maximal secure  GDoF, defined as
\[ \dco (\alpha) \defeq      \min_{d_c :  (\dso(\alpha), d_c) \in  \bar{\mathcal{D}} ( \alpha)  }    d_c .\]

\subsection{Multiple access wiretap channel (MAC-WT)}\label{sec:sysmawc}

Let us now consider the two-user Gaussian multiple access wiretap channel.
The system model of this channel is similar as that of the interference channel defined in Section~\ref{sec:sysICtwouser}. One difference is that both messages $w_1$ and $w_2$ are intended to receiver~1 in this setting.  Another difference is that  receiver~2 now is the eavesdropper. Both messages need to be secure from receiver~2 and the secrecy constraint becomes  $\Imu(w_1, w_2; y_{2}^{\bln})  \leq  \bln \epsilon$. 
The definitions of the rate tuple $(R_1(P, \alpha), R_2 (P, \alpha), R_c (P, \alpha))$,     secure capacity region $\bar{C} (P, \alpha)$, and secure GDoF  regions $\bar{\mathcal{D}} ( \alpha)$ and $\mathcal{D} ( d_c, \alpha)$  follow from that in Section~\ref{sec:sysICtwouser}. 
In this setting,  the secure GDoF  region  $\mathcal{D} ( d_c, \alpha)$ might not be  symmetric due to the asymmetric links arriving at receiver~1. 
Therefore, instead of the  maximal secure sum GDoF, we will focus on the maximal (optimal) secure  GDoF  region defined as 
\[ \mathcal{D}^* ( \alpha)  \defeq \{ (d_1, d_2):   \exists (d_1, d_2)  \in   \cup_{d_c} \mathcal{D} (d_c, \alpha)  \}.\]    
We are also interested in the minimal (optimal) GDoF of the common randomness to achieve any given GDoF pair $(d_1, d_2) \in \mathcal{D}^* ( \alpha)$,  defined as
\[  \dco (\alpha, d_1, d_2) \defeq       \min_{d_c:   (d_1, d_2)   \in \mathcal{D} (d_c, \alpha)   }    d_c      .\]

\section{The main results  \label{sec:mainresult}}

We will  provide here the main results of the channels defined in Section~\ref{sec:system}.
The detailed proofs are provided in Sections~\ref{sec:CJGauIC}-\ref{sec:converse}, as well as the appendices.

\subsection{Removing the secrecy constraints \label{sec:break}}

\begin{theorem} [IC-SC]  \label{thm:ICrGDoF}
For  almost all the channel coefficients  $\{h_{k\ell}\} \in (1, 2]^{2\times 2}$ of the symmetric  Gaussian IC-SC channel with common randomness  (see Section~\ref{sec:sysICtwouser}),  the optimal characterization of the secure sum GDoF  is 
\begin{subnumcases} 
{ \dsum (\alpha)  =} 
     2(1- \alpha)    &    for   \ $ 0 \leq \alpha \leq  \frac{1}{2}$            			\label{thm:GDoFICh1} \\
     2\alpha  &  for \ $\frac{1}{2}  \leq \alpha \leq  \frac{2}{3}$           		\label{thm:GDoFICh2} \\ 
        2(1 -  \alpha / 2)  &  for  \   $\frac{2}{3}  \leq \alpha \leq  1$    			\label{thm:GDoFICh3} \\ 
            \alpha  &  for  \  $1  \leq  \alpha \leq   2$    							\label{thm:GDoFICh4} \\ 
                                               2  &  for  \   $  \alpha  \geq 2$.				\label{thm:capacitydet5}
\end{subnumcases}
This optimal secure sum GDoF  is the same as the optimal sum GDoF of the setting without any secrecy constraint.
\end{theorem}
\begin{proof}
See  Section~\ref{sec:CJGauIC} for the achievability proof. The optimal  sum GDoF of the interference channel without  secrecy constraint, which is characterized in  \cite{ETW:08}, is serving as the upper bound of the secure sum GDoF of this IC-SC channel with common randomness. 
\end{proof}

Note that, without secrecy constraints, the optimal sum GDoF  of the interference channel is a ``W'' curve (see \cite{ETW:08} and Fig.~\ref{fig:ICrGDoF}).  With  secrecy constraints, the secure sum GDoF of the interference channel is then reduced to a modified ``W'' curve (cf.~\cite{ChenIC:18}). It implies that  there is a  penalty in GDoF incurred by the secrecy constraints.   
Interestingly,  Theorem~\ref{thm:ICrGDoF} reveals that we can remove this penalty by adding common randomness, in terms of sum GDoF.

\begin{theorem}  [WTH]\label{thm:GDoFwthr}
Given  the symmetric  Gaussian  WTH channel  with common randomness (see Section~\ref{sec:syswchr}),  the optimal secure  GDoF  is  expressed by
\begin{align}
\dso (\alpha) = 1,      \quad \forall  \alpha \in [0,  \infty),  \non 
\end{align}
which is the same as the maximal  GDoF of the setting  without secrecy constraint.   
\end{theorem}
\begin{proof}
See  Section~\ref{sec:CJGauwiretap}  for the achievability proof. Without secrecy constraint, the WTH channel can be enhanced to a point-to-point channel, and the maximal GDoF  of the  point-to-point channel is $1$.
\end{proof}

For the symmetric Gaussian WTH channel without common randomness, the secure GDoF  is another modified ``W'' curve (cf.~\cite{CG:18arxiv}).
Without secrecy constraint, the maximal GDoF  of the setting is $1$.
Thus, there is  a penalty  in GDoF due to secrecy constraint.   
 Theorem~\ref{thm:GDoFwthr} reveals that we can remove this GDoF penalty by adding common randomness.

\begin{theorem}  [MAC-WT]\label{thm:GDoFmawc}
Given  the symmetric  Gaussian  MAC-WT channel  with common randomness (see Section~\ref{sec:sysmawc}),  the optimal secure   GDoF region $\mathcal{D}^* ( \alpha) $ is the set of all  pairs $(d_1, d_2)$ satisfying 
\begin{align}  
d_1 + d_2 &\leq \max\{1, \alpha\}  \label{eq:MACWT11} \\
0 \leq d_1 &\leq   1   \label{eq:MACWT22} \\
0 \leq d_2 &\leq \alpha,  \label{eq:MACWT33}
\end{align}
which is the same as the optimal GDoF region of the  symmetric Gaussian multiple access channel  without eavesdropper, i.e.,  without secrecy constraint.   
\end{theorem}
\begin{proof}
 The achievability proof is  provided in  Section~\ref{sec:CJGauMAC}.  
 The optimal    GDoF region of the  multiple access channel without  secrecy constraint is serving as the outer bound of the optimal secure   GDoF region  of the  MAC-WT channel with common randomness. The optimal    GDoF region of the symmetric Gaussian multiple access channel is characterized  as in \eqref{eq:MACWT11}-\eqref{eq:MACWT33}, which can be easily derived from the capacity region of the setting (cf.~\cite{CT:06}). 
\end{proof}

For the multiple access channel, there is  a penalty in GDoF region due to secrecy constraint. For example, considering the case with $\alpha=1$,  the optimal sum GDoF of the multiple access channel without secrecy constraint is $1$.  With  secrecy constraint, i.e., with an eavesdropper, the optimal secure sum GDoF of multiple access wiretap channel is reduced to $2/3$ (cf.~\cite{XU:14}). Therefore, secrecy constraint incurs an extra limit on the   GDoF region.
Theorem~\ref{thm:GDoFmawc} reveals that by adding  common randomness we can achieve a secure GDoF region that is the same as the one without secrecy constraint. 
In other words, with common randomness, secrecy constraint will not incur any penalty in   GDoF region of the symmetric multiple access wiretap channel.

\subsection{How much  common randomness is required? \label{sec:minimumcr}}

The  results in Theorems~\ref{thm:ICrGDoF}-\ref{thm:GDoFmawc}  reveal that we can remove the secrecy constraints, i.e., remove the penalty in  GDoF, by adding common randomness for each channel considered here.  
From the practical point of view,  we hope to use less  common randomness to remove the  secrecy constraints.
Therefore, it would be interesting to characterize the minimal  GDoF of the common randomness to achieve this goal.  
The results on this perspective are given in the following theorems. 

\begin{theorem} [IC-SC]  \label{thm:ICrGDoFcr}
For the two-user symmetric  Gaussian IC-SC channel,  the minimal GDoF of the common randomness to achieve the maximal secure sum GDoF $\dsum (\alpha)$   is  
 \begin{align}
\dco (\alpha) =   \dsum(\alpha)/2 -   (1- \alpha)^+      \quad  \quad  \alpha \in [0,  \infty). 
\end{align}
\end{theorem}
\begin{proof}
See  Section~\ref{sec:CJGauIC} for the achievability proof and Section~\ref{sec:converseIC} for the converse proof.
\end{proof}

\begin{theorem}  [WTH]\label{thm:GDoFwthrcr}
For  the symmetric  Gaussian  WTH channel,  the minimal GDoF of the common randomness to achieve the maximal secure  GDoF $\dso (\alpha)$  is   
 \begin{align}
\dco (\alpha) =  1-   (1- \alpha)^+      \quad \quad  \alpha \in [0,  \infty).
\end{align}
\end{theorem}
\begin{proof}
See  Section~\ref{sec:CJGauwiretap} and Section~\ref{sec:conversewiretap} for the achievability and converse proofs, respectively.
\end{proof}

For the MAC-WT channel, we were able  to characterize  the minimal GDoF of the common randomness to achieve any given GDoF pair $(d_1, d_2)$ in the  maximal  secure  GDoF  region  $\mathcal{D}^* ( 1) $ expressed in  Theorem~\ref{thm:GDoFmawc},  for the case of  $\alpha=1$.

\begin{theorem}  [MAC-WT] \label{thm:GDoFmawccr}
Given  the symmetric  Gaussian  MAC-WT channel, and for $\alpha=1$, the minimal GDoF of the common randomness to achieve   any given GDoF pair $(d_1, d_2)$ in the  maximal  secure  GDoF  region  $\mathcal{D}^* ( 1) $   is  
 \begin{align}
\dco (1, d_1, d_2) =  \max\{d_1,  d_2  \}     \quad  \text{for} \quad   (d_1, d_2) \in  \mathcal{D}^* ( 1) ,  \  \alpha =1 .     \non 
\end{align}
\end{theorem}
\begin{proof}
 The achievability and converse proofs are provided in  Section~\ref{sec:CJGauMAC} and Section~\ref{sec:conversewiretapmac}, respectively.
\end{proof}

When $ \alpha =1$,  Theorem~\ref{thm:GDoFmawccr} reveals that the minimal GDoF of the common randomness to achieve the secure GDoF pair $(d_1 =1/2,  d_2 =1/2) \in \mathcal{D}^* ( 1)$ is 1/2. It implies that  1/2 GDoF of common randomness achieves the  maximal secure  sum GDoF  1.  Without common randomness, the secure  sum GDoF cannot be more than $2/3$ for the case with $ \alpha =1$.  
Note that it is challenging to characterize  $ \dco (\alpha, d_1, d_2)$   for the general case of  $\alpha$. For  the general case,  the optimal secure  GDoF  region is non-symmetric in $(d_1, d_2)$ as shown in   Theorem~\ref{thm:GDoFmawc}. To achieve the GDoF pairs in the   asymmetric  secure  GDoF  region, it might require several converse bounds on the minimal GDoF of the common randomness, which will be studied in our future work.

\section{Scheme example  \label{sec:schemeexample} }

\begin{figure}[t!]
\centering
\includegraphics[width=8cm]{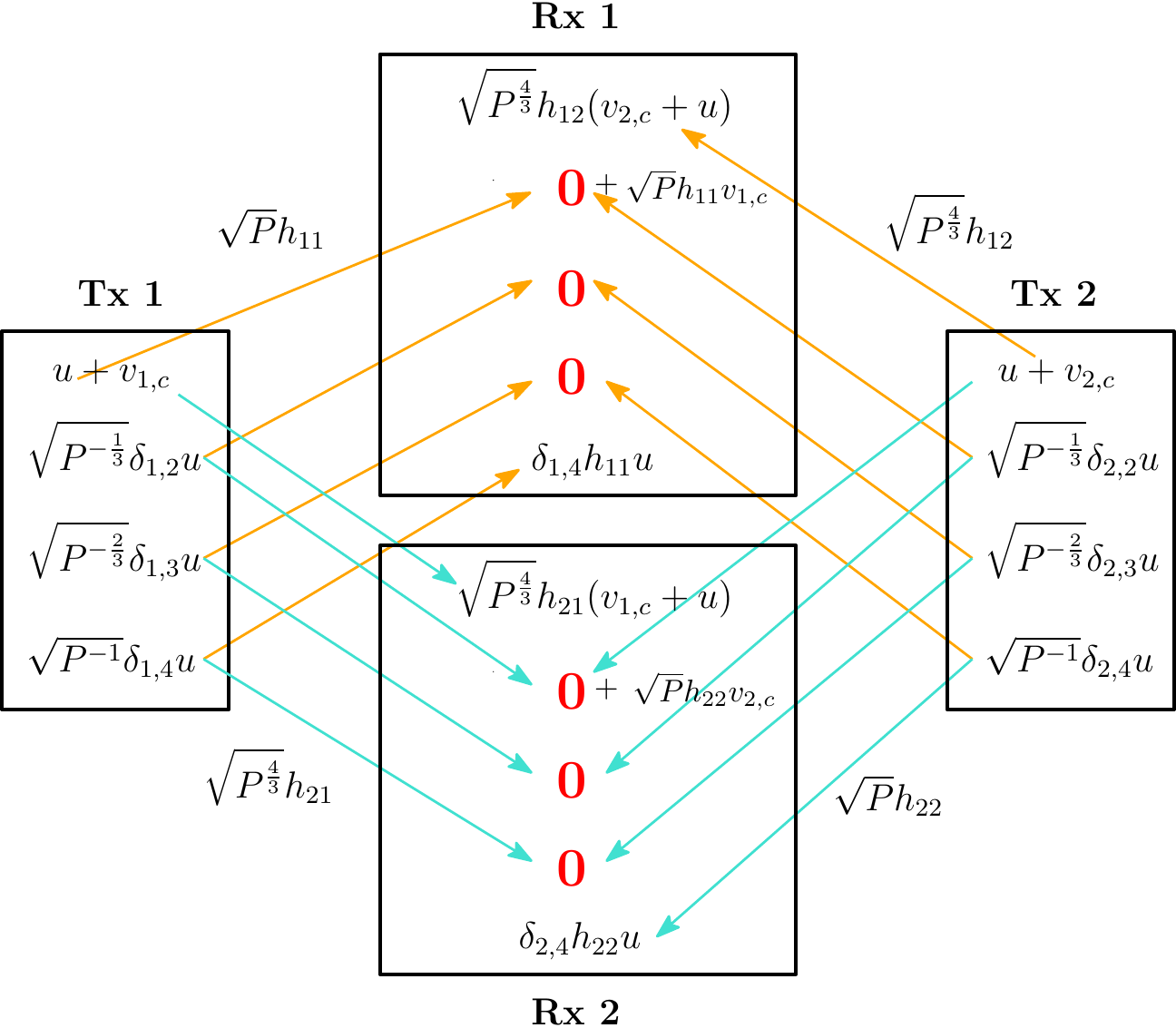}
 \vspace{-.05 in}
\caption{Markov chain-based interference neutralization at the receivers, for a two-user interference channel  with $\alpha = 4/3$.}
\label{fig:ICalpha43}
\vspace{-12pt}
\end{figure}

We will here provide  a scheme example,  focusing on the IC-SC channel with  $\alpha = 4/3$ (see Section~\ref{sec:sysICtwouser}). Note that for the case of $\alpha = 4/3$, without the consideration of secrecy constraints the sum  GDoF  is $4/3$ (cf.~\cite{ETW:08}).  With the consideration of secrecy constraints, the secure sum  GDoF is reduced to $8/9$ (cf.~\cite{ChenIC:18}).  In this example, we will show that by adding common randomness  the secure sum GDoF  can be improved to $4/3$, which matches the sum GDoF for the case \emph{without} secrecy constraints. In our scheme,   Markov chain-based  interference neutralization will be used in the signal design.
In this scheme, the transmitted signals are given as (without time index): 
\[x_k  =  v_{k,c}  + \sum_{\ell=1}^{4}  \delta_{k, \ell}  \sqrt{P^{ -  \beta_{u_{\ell}}}}   \cdot   u\] for $k\in \{1,2\}$,   
where $\beta_{u_{\ell}} =\frac{1}{3}(\ell-1)$, for $\ell \in  \{1, 2, 3, 4\}$;  and 
\begin{align}
\delta_{j,\ell} = \left\{ \begin{array}{ll}
- \frac{h_{ii}}{h_{ij}} \cdot   \big(\frac{h_{11} h_{22}}{h_{12} h_{21}}\big)^{\frac{\ell}{2}-1}  & \quad \textrm{$\ell \in\{ 2, 4\}$}\\
\big(\frac{h_{11} h_{22}}{h_{12} h_{21}}\big)^{\frac{\ell-1}{2}} & \quad  \textrm{$\ell \in \{ 1, 3\}$}
\end{array} \right.
\end{align}
for $i, j \in \{1,2\}, i \not=j.$
$u$ is the common randomness. $v_{1,c}$ and $v_{2,c}$ carry the messages of transmitters~1 and 2, respectively. The random variables  $v_{1,c}$, $v_{2,c}$  and $u$  are \emph{independently} and \emph{uniformly} drawn from a pulse  amplitude modulation (PAM) set 
 \[v_{1,c}, v_{2,c}, u    \in    \Omega ( \xi =  \frac{ \gamma}{Q} ,  \quad  Q =  P^{ \frac{  2/3  - \epsilon }{2}} )\] 
where   $\gamma  \in \bigl(0, 1/64 \bigr]$ is a constant, 
$\Omega (\xi,  Q)  \defeq   \{ \xi  a :   \    a \in  [-Q,   Q]    \cap  \Zc \}$,  and $\epsilon>0$ is a parameter that can be made arbitrarily small. 
 With this signal design,   $v_{k,c}$ carries  $ 2/3 $ GDoF, i.e., $ \Hen(v_{k,c}) =  \frac{2/3 - \epsilon}{2}\log  P + o(\log P)$, with  $ k \in \{1,2\}$. 
 One can check that the average power constraints $\E |x_1|^2 \leq 1$ and $\E |x_2|^2 \leq 1$ are satisfied.
Then, the received signals are given as  (without time index)
\begin{align}
y_{k} & \!  =   \! \sqrt{P}  h_{kk} v_{k,c}  \!+\! \underbrace{ \sqrt{P^{ 4/3 }}  h_{kj}   v_{j,c}  \!+\!    \sqrt{P^{4/3}} \delta_{j,1}      h_{kj}  u}_{\rm{aligned}}  \!+\!   \sum_{\ell=1}^{3} \! \underbrace{ (\sqrt{P^{(4 \!- \!\ell)/3}}   \delta_{k,\ell}   h_{kk}  \!+\!  \sqrt{P^{(4  \!- \! \ell)/3}}  \delta_{j,\ell\!+\!1}      h_{kj}   ) u}_{\rm{interference~neutralization}} \!+\!  \delta_{k,4}   h_{kk}  u\!+\!  z_{k}  \non\\
 & =     \sqrt{P}  h_{kk} v_{k,c} +    \sqrt{P^{ 4/3 }} h_{kj}  (   v_{j,c}  +  u)  \!+\!       \delta_{k,4}   h_{kk}  u \!+\!  z_{k}   \non
\end{align}
for $k \neq j$, $k,j\in \{1,2\}$.
The idea of the Markov chain-based interference neutralization method is given as follows.  As shown in Fig.~\ref{fig:ICalpha43},  the common randomness $u$ is used to generate a certain number of  signals with specific  directions and powers, i.e., $\{  \delta_{1, \ell}  \sqrt{P^{ -  \beta_{u_{\ell}}}}      u\}_{\ell=1}^4$ at transmitter~1 and $\{  \delta_{2, \ell}  \sqrt{P^{ -  \beta_{u_{\ell}}}}  u\}_{\ell=1}^4$ at transmitter~2;   the  signal $\delta_{2, 1}  \sqrt{P^{ -  \beta_{u_{1}}}}  u$ from transmitter~2 is used to jam the information signal $v_{2,c}$ at receiver~1 but it will create an interference at receiver~2; this interference will be neutralized by the signal $\delta_{1, 2}  \sqrt{P^{ -  \beta_{u_{2}}}}      u$  from transmitter~1; the added signal $\delta_{1, 2}  \sqrt{P^{ -  \beta_{u_{2}}}}      u$ also creates another interference at receiver~1 but will be neutralized by the next generated signal $\delta_{2, 3}  \sqrt{P^{ -  \beta_{u_{3}}}}  u$; this process repeats until the residual interference is under the noise level.  Since  one signal is used to neutralize the previous signal and will be  neutralized by the next signal, it forms a  Markov chain for this  interference neutralization process.

From our signal design, it can be proved that the secure rates  $R_k  =    \Imu(v_{k,c}; y_k) - \Imu(v_{k,c}; y_j | v_{j,c} )  \geq \frac{ 2/3 - \epsilon}{2} \log P + o(\log P) $,  $k \neq j$, $k,j\in \{1,2\}$,  and  the secure sum GDoF $d_{\text{sum}} = 4/3$,  are achievable for almost all the channel coefficients  $\{h_{k\ell}\} \in (1, 2]^{2\times 2}$, by using $d_c= 2/3$ GDoF of common randomness.
More details on the proposed scheme can be found in Section~\ref{sec:CJGauIC}.

\section{Achievability  for interference channel  \label{sec:CJGauIC} }

We will here provide the achievability  scheme for the \emph{symmetric} Gaussian IC-SC channel (see Section~\ref{sec:sysICtwouser}).  
Our scheme uses PAM modulation,   Markov chain-based  interference neutralization and  alignment technique in the signal design. 
For the case with $0 \leq  \alpha \leq 1/2$, the optimal secure sum GDoF is achievable without adding common randomness (cf.~\cite{ETW:08, ChenIC:18}). Thus,  here we will  just focus on the case with  $\alpha > 1/2$.   The  scheme details are given in the following subsections.

\subsubsection{Codebook generation}
Transmitter~$k$, $k=1,2$, at first generates a codebook as
  \begin{align}
     \mathcal{B}_{k} \defeq \Bigl\{  v^{\bln}_{k} (w_k,  w'_k):  \  w_{k} \in \{1,2,\cdots, 2^{\bln R_{k}}\},   
      w'_k \in \{1,2,\cdots, 2^{\bln R'_k}\}   \Bigr\}     \label{eq:code2341J}
     \end{align}
where $v^{\bln}_{k} $ denotes the corresponding codewords. The elements of  the codewords are generated  independently and identically based on a particular distribution.   $w'_k$ is an independent randomness that is used to protect the confidential message, and is uniformly distributed over $\{1,2,\cdots, 2^{\bln R'_k}\}$.  
 $R_{k}$ and  $R'_k$ are the rates  of   $w_{k}$  and $w'_k$, respectively.
To transmit the confidential message $w_{k}$, transmitter~$k$  \emph{randomly} chooses a codeword $v^{\bln}_{k}$ from  a sub-codebook  $\mathcal{B}_{k}( w_{k}) $   defined  by 
\begin{align}
  \mathcal{B}_{k} (w_{k})  \defeq \bigl\{ v^{\bln}_{k} (w_{k},  w'_k): \  w'_k \in \{1,2,\cdots, 2^{\bln R'_k}\}   \bigr\},  \quad k=1,2     \label{eq: codebook89284}
\end{align}
 according to a uniform distribution. Then, the selected codeword $v^{\bln}_{k}$ is mapped to the channel input based on the following signal design 
 \begin{align}
  x_{k}(t) =     \varepsilon  v_{k} (t)    + \varepsilon \sum_{\ell=1}^{\tau}  \delta_{k,\ell}   \sqrt{P^{ -  \beta_{u_{\ell}}}} \cdot   u (t)      \label{eq:xvkkk}
   \end{align}
  for $ k=1,2$, where $v_{k} (t)$ denotes the $t$th element of  $v^{\bln}_{k}$;   $\{\delta_{j,\ell}\}_{j, \ell}$ are  parameters that will be  designed specifically later on for different cases of $\alpha$,  based on the Markov chain-based  interference neutralization and  alignment technique. 
  $\varepsilon$ is a parameter designed as 
   \begin{align}\label{eq:varadef11}
 \varepsilon \defeq \begin{cases}  
1      & \quad  \rm{ if~\alpha \neq 1}\\
\frac{h_{11} h_{22}- h_{12} h_{21}}{8}     & \quad  \rm{ if~\alpha = 1}
\end{cases}
   \end{align}
   which is used to regularize the power of the transmitted signal.
$\tau$ is a parameter designed as 
   \begin{align} \label{eq:tau99}
 \tau \defeq \begin{cases}  
\lceil \frac{\alpha}{1-\alpha} \rceil      & \quad  \rm{ if~\alpha < 1}\\
\lceil \frac{\alpha}{\alpha -1} \rceil      & \quad  \rm{ if~\alpha > 1}\\
1   & \quad  \rm{ if~\alpha = 1} .
\end{cases}
   \end{align}
   $u$ is a random variable \emph{independently}  and  \emph{uniformly}  drawn from a  PAM constellation set,  which will be specified  later on.  For the proposed scheme, the common randomness $w_c$ is mapped into three random variables, i.e.,  $w'_1, w'_2$ and $u$, such that $\Hen(w_c) = \Hen(w'_1)+\Hen(w'_2) + \Hen(u)$ and  $w'_1, w'_2$ and $u$ are mutually independent.  Based on our definition, $w'_1, w'_2$ and $u$ are  available at the transmitters but not at the receivers.

\subsubsection{Signal design}

For transmitter $k$, $k=1,2$,  each element of the codeword is designed to have the  following form
 \begin{align}
   v_{k}  =   v_{k, c} +   \sqrt{P^{ - \beta_{k,p}}}     v_{k,p} .  \label{eq:xvk}  
 \end{align}
With this,  the   input signal in~\eqref{eq:xvkkk} can be expressed as 
 \begin{align}
  x_{k}  =   \varepsilon    v_{k,c} +   \varepsilon   \sqrt{P^{ - \beta_{k,p}}}    \cdot  v_{k,p}  + \varepsilon \sum_{\ell=1}^{\tau}  \delta_{k,\ell}   \sqrt{P^{ -  \beta_{u_{\ell}}}} \cdot   u  , \quad k =1,2   \label{eq:xvkkk1}  
 \end{align}
(without time index for simplicity), where  random variables 
$\{v_{k,c}, v_{k,p}, u\}$ are  \emph{independently}  and  \emph{uniformly}  drawn from the following PAM constellation sets
 \begin{align}
   v_{k,c},  u    &  \in    \Omega ( \xi =  \frac{   \gamma}{Q} ,   \   Q =  P^{ \frac{ \lambda_{k,c} }{2}} )  \label{eq:constellationGsym1}   \\ 
   v_{k,p}      &  \in    \Omega ( \xi =  \frac{\gamma}{2Q} ,   \   Q = P^{ \frac{  \lambda_{k,p} }{2}} )    \label{eq:constellationGsym2}    
 \end{align}
 where $\gamma$ is a parameter satisfying the  constraint \[\gamma  \in \bigl(0, \  \frac{1}{ \tau \cdot 2^\tau}]. \]
  In the proposed scheme,   the designed parameters $\{\beta_{k,p}, \beta_{u_{\ell}}, \lambda_{k,c},  \lambda_{k,p}, \lambda_{u}\}_{k,\ell}$  are given in Table~\ref{tab:ICpara}   for different regimes\footnote{Without loss of generality we will take the assumption that $P^{ \frac{ \lambda_{k,c} }{2}}$ and $P^{ \frac{  \lambda_{k,p} }{2}}$  are integers, for $k=1,2$. 
For example,   when $P^{ \frac{ \lambda_{2,c} }{2}}$ isn't an integer,   the parameter $\epsilon$  in Table~\ref{tab:ICpara} can be  slightly modified such that $P^{ \frac{ \lambda_{2,c} }{2}}$ is an integer, for the regime with large $P$. Similar assumption will also be used in the next channel models later.}.  
Based on the signal design in \eqref{eq:constellationGsym1} and \eqref{eq:constellationGsym2}, we have
\begin{align}
 \E |v_{k,c}|^2  &=  \frac{2 \times (\frac{ \gamma}{Q})^2 }{ 2Q +1}  \sum_{i=1}^{Q} i^2  =  \frac{  (\frac{  \gamma}{Q})^2 \cdot Q(Q+1)}{3}  \leq   \frac{  2 \gamma^2 }{3}    \label{eq:power35781}\\   
 \E |u|^2 &\leq   \frac{ 2  \gamma^2 }{3}     \label{eq:power35782}\\
   \E |v_{k,p}|^2   & \leq   \frac{    \gamma^2 }{6}.  \label{eq:power3578}
\end{align} 
From \eqref{eq:varadef11}, \eqref{eq:xvkkk1} and \eqref{eq:power35781}-\eqref{eq:power3578}, we can  verify that the signal $x_k$ satisfies the  power constraint, that is  
\begin{align}
 \E |x_k|^2   =  & \varepsilon^2  \E |v_{k,c}|^2   + \varepsilon^2   P^{ - \beta_{k,p}}   \E |v_{k,p}|^2    +   \varepsilon^2  \big(\sum_{\ell=1}^{\tau}  \delta_{k,\ell}  P^{ -  \beta_{u_{\ell}}} \big)^2 \E |u|^2       \non\\
  \leq & \varepsilon^2   \times \frac{   2\gamma^2 }{3}   +  \varepsilon^2   \times \frac{  \gamma^2 }{6} +  \varepsilon^2    \times  \tau^2  4^{\tau} \times  \frac{ 2  \gamma^2 }{3}   \non\\     \leq &1     \non
\end{align} 
for $k=1,2$, where  $\gamma  \in \bigl(0, \frac{1}{ \tau \cdot 2^\tau}\bigr]$, $\varepsilon^2 \leq 1$ and $\delta_{k,\ell}$  is  designed specifically for different cases of $\alpha$ satisfying the inequality $\varepsilon^2  \delta^2_{k,\ell} \leq  4^{\tau}, \forall k, \ell $, which will be shown later on.

\subsubsection{Secure rate analysis} We define  the rates  $R_k$ and $R_k'$ as 
\begin{align}
R_k &\defeq   \Imu(v_k; y_k) -  \Imu ( v_k; y_{\ell} | v_{\ell} ) - \epsilon   \label{eq:Rk623J} \\  
R_k'  &\defeq  \Imu ( v_k; y_{\ell} | v_{\ell}) - \epsilon  \label{eq:Rk623bJ}  
\end{align}
for some $\epsilon >0$, and $ \ell, k \in \{1,2\}, \ell \neq k$. With our  codebook and signal design, the result of \cite[Theorem~2]{XU:15}   (or \cite[Theorem~2]{LMSY:08}) suggests that the rate pair $(R_1, R_2)$ defined above is achievable  and  the transmission of the messages is secure, i.e.,   $\Imu(w_1; y_{2}^{\bln})  \leq  \bln \epsilon$ and $\Imu(w_2; y_{1}^{\bln})  \leq  \bln \epsilon$. Remind that, based on our codebook design, $v_{1}$ and $v_{2}$ are independent, since $w_1, w_2, w'_1, w'_2$ are mutually independent (cf.~\eqref{eq:code2341J}). 

In what follows we will show how to remove the secrecy constraints in terms of GDoF performance by adding common randomness, focusing on the regime of  $ \alpha > 1/2$. 
 Specifically, we will consider the following five cases: $\frac{2}{3} \leq  \alpha < 1$,   $1 < \alpha \leq 2$,  $\alpha=1$, $\frac{1}{2} < \alpha \leq \frac{2}{3}$, and  $2 \leq \alpha $. 
 In the achievability scheme, a Markov chain-based  interference neutralization method is proposed to accomplish the role of common randomness.

\begin{table}
\caption{Parameter design for the  IC-SC channel.}
\begin{center}
{\renewcommand{\arraystretch}{1.7}
\begin{tabular}{|c|c|c|c|c|c|}
  \hline
                     &   $\frac{1}{2} < \alpha \leq \frac{2}{3}$  &  $\frac{2}{3} \leq  \alpha < 1$  &$\alpha= 1$& $1 < \alpha \leq 2$   & $2 \leq \alpha $  \\
    \hline
  $\beta_{u_{\ell}}, \ell   \in  \{ \!1,\! 2, \cdots\!, \!\tau\!\}$     		     		& $(1\!-\!\alpha)\ell  $      	&    $(1\!-\!\alpha)\ell $   & $0$ &$ (\alpha\!-\!1)(\ell\!-\!1) $   &   $ (\alpha\!-\!1)(\ell\!-\!1) $ \\
    \hline
   $\beta_{1,p}, \ \beta_{2,p}$ 			&   $\alpha$    		&    $\alpha$    &  $\infty$ &  $\infty$    &    $\infty$  \\
    \hline
   $\lambda_{1,c}, \ \lambda_{2,c}$ 		&   $2\alpha -1 - \epsilon$ 	&  $\alpha/2 - \epsilon$     & $1/2 - \epsilon$ &   $\alpha/2 - \epsilon$     &  $1 - \epsilon$   \\
    \hline
   $\lambda_{u}$ 					&    $2\alpha -1 - \epsilon$ &   $\alpha/2 - \epsilon$    & $1/2 - \epsilon$ &   $\alpha/2 - \epsilon$     &   $1 - \epsilon$  \\
  \hline
   $\lambda_{1,p}, \ \lambda_{2,p}$  &   $1 - \alpha - \epsilon$  	&  $1- \alpha - \epsilon$     &   $0$ &   $0 $  &    0  \\
    \hline
    \end{tabular}
}
\end{center}
\label{tab:ICpara}
\end{table}

\subsection{ $2/3 \leq  \alpha < 1$   \label{sec:CJschemeIC231}}

In this case with $2/3 \leq  \alpha < 1$,  based on the  parameters designed in Table~\ref{tab:ICpara},   the transmitted signals take the following forms
\begin{align}
 x_1  = &    v_{1,c}  +    \sqrt{P^{ -  \alpha}}  \cdot   v_{1,p} +  \sum_{\ell=1}^{\tau}  \delta_{1,\ell}   \sqrt{P^{ -  \beta_{u_{\ell}}}} \cdot   u  \label{eq:IC321x1}  \\
 x_2  = &     v_{2,c} +    \sqrt{P^{ -  \alpha}}  \cdot   v_{2,p} +    \sum_{\ell=1}^{\tau}  \delta_{2,\ell}   \sqrt{P^{ -  \beta_{u_{\ell}}}}   \cdot  u  \label{eq:IC321x2}
 \end{align}
where the parameters $\{\delta_{j,\ell}\}_{j, \ell}$ are designed as 
\begin{align} \label{eq: IC321delta}
\delta_{j,\ell} = \left\{ \begin{array}{ll}
- \big(\frac{h_{12} h_{21}}{h_{11} h_{22}}\big)^{\frac{\ell}{2}} & \quad  \textrm{$\ell \in \{2k: 2k\leq \tau, k\in \Zc^+ \}$}\\
&\\
 \frac{h_{ji}}{h_{jj}} \cdot   \big(\frac{h_{12} h_{21}}{h_{11} h_{22}}\big)^{\frac{\ell-1}{2}}  & \quad \textrm{$\ell \in \{2k-1: 2k-1 \leq \tau, k\in \Zc^+\}$}
\end{array} \right.
\end{align}
for $i, j \in \{1,2\}, i \not=j.$ Note that  the common randomness $u$ is used to generate a certain number of  signals with specific  directions and powers, i.e., $\{  \delta_{1, \ell}  \sqrt{P^{ -  \beta_{u_{\ell}}}} u\}_{\ell=1}^\tau$ at transmitter~1 and $\{  \delta_{2, \ell}  \sqrt{P^{ -  \beta_{u_{\ell}}}} u\}_{\ell=1}^\tau$ at transmitter~2.
Then, the received signals are expressed as
\begin{align}
y_{1}&  =     \sqrt{P}  h_{11} v_{1,c}  +     \sqrt{P^{ 1 - \alpha}}  h_{11} v_{1,p}    + \underbrace{   \sqrt{P^{ \alpha }}  h_{12}     v_{2,c}  +    \sqrt{P^{1-  \beta_{u_1}}} \delta_{1, 1}     h_{11}  u}_{\rm{aligned}}\non\\ & \quad  +    \sum_{\ell=1}^{\tau-1}   \underbrace{ (\sqrt{P^{1-  \beta_{u_{\ell+1}}}}   \delta_{1,\ell+1}   h_{11}  +  \sqrt{P^{ \alpha-  \beta_{u_{\ell}}}} \delta_{2,\ell}     h_{12}   ) u}_{\rm{interference~neutralization}} +  \underbrace{  \sqrt{P^{ (\tau+1)\alpha-\tau}}  \delta_{2, \tau}   h_{12}  u +    h_{12} v_{2,p}}_{\text{treated as noise}} +  z_{1}  \non\\
 & =    \sqrt{P}  h_{11} v_{1,c}  +     \sqrt{P^{ 1 - \alpha}}  h_{11} v_{1,p}   +     \sqrt{P^{ \alpha }} h_{12}  (   v_{2,c}  +  u)  +   \sqrt{P^{ (\tau+1)\alpha-\tau}}  \delta_{2, \tau}   h_{12}  u +    h_{12} v_{2,p} +  z_{1}    \label{eq:IC321y1}  \\
y_{2} & =     \sqrt{P} h_{22} v_{2,c}  +     \sqrt{P^{ 1 - \alpha}}  h_{22} v_{2,p}  + \underbrace{   \sqrt{P^{ \alpha }}  h_{21}     v_{1,c}  +    \sqrt{P^{1-  \beta_{u_1}}} \delta_{2,1}      h_{22}  u}_{\rm{aligned}}\non\\ & \quad +    \sum_{\ell=1}^{\tau-1} \underbrace{ (\sqrt{P^{1-  \beta_{u_{\ell+1}}}}    \delta_{2,\ell+1}    h_{22}  +  \sqrt{P^{ \alpha-  \beta_{u_{\ell}}}} \delta_{1,\ell}     h_{21}   ) u}_{\rm{interference~neutralization}} + \underbrace{   \sqrt{P^{ (\tau+1)\alpha-\tau}}   \delta_{1, \tau}    h_{21} u +   h_{21} v_{1,p}}_{\text{treated as noise}}  +  z_{2}  \non\\
& =     \sqrt{P} h_{22} v_{2,c}  +     \sqrt{P^{ 1 - \alpha}}  h_{22} v_{2,p}   +     \sqrt{P^{ \alpha }}  h_{21}  (   v_{1,c}  +  u) +    \sqrt{P^{ (\tau+1)\alpha-\tau}}   \delta_{1, \tau}    h_{21} u +   h_{21} v_{1,p}  +  z_{2}.   \label{eq:IC321y2} 
\end{align}

At the receivers, a Markov chain-based  interference neutralization method is used to remove the interference.  In the above expressions of $y_1$ and $y_2$, we can see that  the  signal $\delta_{1, 1}  \sqrt{P^{ -  \beta_{u_{1}}}}  u$ from transmitter ~1 is used to jam the information signal $v_{2,c}$ at receiver~1 but it will create an interference at receiver~2; this interference will be neutralized by the signal $\delta_{2, 2}  \sqrt{P^{ -  \beta_{u_{2}}}}      u$  from transmitter~2; the added signal $\delta_{2, 2}  \sqrt{P^{ -  \beta_{u_{2}}}}      u$ also creates another interference at receiver~1 but will be neutralized by the next generated signal $\delta_{1, 3}  \sqrt{P^{ -  \beta_{u_{3}}}}  u$; this process repeats until the residual interference can be treated as noise, that is,  both interference  $\sqrt{P^{ (\tau+1)\alpha-\tau}}  \delta_{2, \tau}   h_{12}  u$  and $ \sqrt{P^{ (\tau+1)\alpha-\tau}}   \delta_{1, \tau}    h_{21} u$  can be treated as noise terms  at receiver~1 and receiver~2, respectively.

Based on our signal design, we will prove that the secure rates satisfy $R_k  =     \Imu(v_k; y_k) -  \Imu ( v_k; y_{\ell} | v_{\ell} )  \geq \frac{ 1- \alpha/2 - 2\epsilon}{2} \log P + o(\log P) $,  for  $k, \ell \in \{1,2\}$, $k \neq \ell$,  and  the secure sum GDoF $d_{\text{sum}} = 2(1-\alpha/2)$ is achievable, for almost all the  realizations of the channel coefficients  $\{h_{k\ell}\} \in (1, 2]^{2\times 2}$.  
For the  secure rates described in  \eqref{eq:Rk623J},  letting $\epsilon \to 0$ gives 
\begin{align}
R_1  & =    \Imu(v_1; y_1) - \Imu(v_1; y_2 | v_2 )      \label{eq:lboundit1} \\
R_2  & =    \Imu(v_2; y_2) - \Imu(v_2; y_1 | v_1 ) .    \label{eq:lboundit2}
\end{align}
Due to the symmetry we will focus on bounding the secure rate   $R_1 $  (see \eqref{eq:lboundit1}).  
We will use $ \hat{v}_{1,c}$ and  $\hat{v}_{1,p}$ to denote the  estimates for  $v_{1,c}$ and $v_{1,p}$ respectively from $y_1$,  and use $ \text{Pr} [  \{ v_{1,c} \neq \hat{v}_{1,c} \} \cup  \{ v_{1,p} \neq \hat{v}_{1,p} \}  ] $ to represent  the  error probability of this estimation.
Then the term $\Imu(v_1; y_1)$ can be lower bounded by
 \begin{align}
  &\Imu(v_1; y_1)     \non\\
  \geq &  \Imu(v_1; \hat{v}_{1,c}, \hat{v}_{1,p})  \label{eq:rate471073}     \\
  =  & \Hen(v_1) -   \Hen(v_1  |  \hat{v}_{1,c}, \hat{v}_{1,p})    \non    \\
    \geq &   \Hen(v_1) -       \bigl( 1+    \text{Pr} [  \{ v_{1,c} \neq \hat{v}_{1,c} \} \cup  \{ v_{1,p} \neq \hat{v}_{1,p} \}  ] \cdot \Hen(v_{1}) \bigr)  \label{eq:rate4179439}     \\
         =    &  \bigl( 1 -   \text{Pr} [  \{ v_{1,c} \neq \hat{v}_{1,c} \} \cup  \{ v_{1,p} \neq \hat{v}_{1,p} \}  ] \bigr)   \cdot \Hen(v_{1})  - 1  \label{eq:rate2092434}   
 \end{align}
where \eqref{eq:rate471073} results from the Markov chain $v_1 \to y_1 \to  \{\hat{v}_{1,c}, \hat{v}_{1,p} \}$;
\eqref{eq:rate4179439} uses  Fano's inequality. The rates of $v_{1,c}$, $v_{1,p}$ and $v_{1} =  v_{1,c} +   \sqrt{P^{ - \alpha}}    \cdot  v_{1,p} $ are  computed as
  \begin{align}
  \Hen(v_{1,c}) &=  \log (2 \cdot P^{ \frac{ \alpha/2  - \epsilon}{2}} +1)   \label{eq:rateAWGNIC111341}     \\
  \Hen (v_{1,p}) &=  \log (2 \cdot P^{ \frac{ 1 - \alpha - \epsilon}{2}} +1)   \label{eq:rateAWGNIC222341}   \\
  \Hen(v_{1})
& =    \frac{1- \alpha/2  - 2\epsilon}{2} \log P + o(\log P)  \label{eq:rateAWGNIC333341}                                 
 \end{align}
 where  $v_{1,p}  \in    \Omega (\xi   = \frac{ \gamma}{2Q},   \   Q = P^{ \frac{ 1 - \alpha - \epsilon}{2}} ) $ and $v_{1,c} \in    \Omega (\xi   =   \frac{ \gamma }{Q},   \   Q =  P^{ \frac{ \alpha/2  - \epsilon}{2}} ) $. 
Based on our signal design,  with  $v_{k}$  we can reconstruct $\{v_{k,c}, v_{k,p}\}$, and vice versa, for $k=1, 2$. 
To  derive the lower bound of $ \Imu(v_1; y_1)$, we provide a result below. 
\begin{lemma}  \label{lm:ICrateerror537}
 Consider  the signal design in \eqref{eq:xvkkk1}-\eqref{eq:constellationGsym2} and \eqref{eq:IC321x1}-\eqref{eq: IC321delta} for the case with  $2/3 \leq \alpha < 1$. For almost all the channel realizations,  the error probability  of decoding $\{v_{k,c}, v_{k,p} \}$ from $y_k$ is vanishing when $P$ goes large, that is
 \begin{align}
 \text{Pr} [  \{ v_{k,c} \neq \hat{v}_{k,c} \} \cup  \{ v_{k,p} \neq \hat{v}_{k,p} \}  ]  \to 0         \quad \text {as}\quad  P\to \infty, \quad k=1,2 .    \label{eq:errorcase1111}
 \end{align}
 \end{lemma}
 \begin{proof}
See Appendix~\ref{sec:ICrateerror537}.
  \end{proof}
By incorporating the results of  \eqref{eq:rateAWGNIC333341} and Lemma~\ref{lm:ICrateerror537} into \eqref{eq:rate2092434},  the term  $ \Imu(v_1; y_1)$ in  \eqref{eq:lboundit1} can be lower bounded by  
  \begin{align}
  \Imu(v_1; y_1)   \geq     \frac{1- \alpha/2  - 2\epsilon}{2} \log P + o(\log P)     \label{eq:rateAWGNIC17762341}  
 \end{align}
 for  almost all the channel coefficients  $\{h_{k\ell}\} \in (1, 2]^{2\times 2}$. 
For  the term  $\Imu(v_1; y_2 | v_2 )$  in  \eqref{eq:lboundit1}, we can treat it  as a rate penalty. This penalty can be  bounded by 
  \begin{align}
  &\Imu(v_1; y_2 | v_2 )     \non\\
\leq &  \Imu(v_1; y_2,  v_{1,c} + u  | v_2 )     \label{eq:rate042845523}   \\
=   & \Imu(v_1;    v_{1,c} + u  )   +  \Imu(v_1;      \sqrt{P^{ (\tau+1)\alpha-\tau}}   \delta_{1, \tau}    h_{21} u +   h_{21} v_{1,p}  +  z_{2}      | v_2 , v_{1,c} + u )     \label{eq:rrateAWGNIC18735}   \\
\leq   & \underbrace{ \Hen( v_{1,c} + u  )}_{\leq \log (4  \cdot P^{ \frac{ \alpha/2 - \epsilon}{2}} +1) }  - \underbrace{\Hen(u  )}_{=  \log (2  \cdot P^{ \frac{ \alpha/2 - \epsilon}{2}} +1)}     +    \underbrace{ \hen(   \sqrt{P^{ (\tau+1)\alpha-\tau}}   \delta_{1, \tau}    h_{21} u +   h_{21} v_{1,p}  +  z_{2}  )  }_{ \leq \frac{1}{2} \log ( 2 \pi e (      4^{\tau +1}   + 5 )) } \non\\& -   \underbrace{\hen(    \sqrt{P^{ (\tau+1)\alpha-\tau}}   \delta_{1, \tau}    h_{21} u +    h_{21} v_{1,p} +  z_{2}  |  v_2 , v_{1,c} + u, v_1,  u) }_{=\hen( z_{2})   }   \label{eq:rrateAWGNIC77364}    \\
\leq &  \underbrace{ \log (4   P^{ \frac{ \alpha/2 - \epsilon}{2}} +1)    -   \log (2   P^{ \frac{ \alpha/2 - \epsilon}{2}} +1) }_{\leq 1}   +  \frac{1}{2} \log ( 2 \pi e (      4^{\tau +1}   + 5 ))    -  \frac{1}{2}\log (2 \pi e)    \label{eq:rrateAWGNIC1955}   \\
\leq &  1   +   \frac{1}{2} \log (      4^{\tau +1}   + 5 )   \non  \\
= & \log (2\sqrt{4^{\tau +1}   + 5})     \label{eq:rate49030433}  
 \end{align}
where  
 \eqref{eq:rrateAWGNIC18735} follows from the fact that $v_1, v_2, u $ are mutually independent;
\eqref{eq:rrateAWGNIC77364}   stems from the fact that $\{ v_{k,p}, v_{k,c}\}$ can be reconstructed from $v_{k}$ for $k=1,2$, and the identity that adding a condition will not increase the differential entropy;
\eqref{eq:rrateAWGNIC1955} results from  the derivations that  $\Hen( v_{1,c} + u  ) \leq  \log (4   P^{ \frac{ \alpha/2 - \epsilon}{2}} +1)$, and that  $\hen(   \sqrt{P^{ (\tau+1)\alpha-\tau}}   \delta_{1, \tau}    h_{21} u +   h_{21} v_{1,p}  +  z_{2}  )  \leq  \frac{1}{2} \log ( 2 \pi e (     P^{ (\tau+1)\alpha-\tau}   \delta^2_{1, \tau}   \cdot | h_{21}|^2\cdot \E |u|^2  +   |h_{21}|^2  \cdot \E |v_{1,p}|^2    +  1 )) \leq  \frac{1}{2} \log ( 2 \pi e (      4^{\tau} \times 4  + 4   +  1 )) $.
 In this case with $2/3 \leq  \alpha < 1$,    $P^{ (\tau+1)\alpha-\tau} =P^{ \lceil \frac{\alpha}{1-\alpha} \rceil (\alpha-1) + \alpha}  \leq P^{ \frac{\alpha}{1-\alpha} (\alpha-1) + \alpha} =1$,  and  $\delta^2_{k,\ell} \leq  4^{\tau}, \forall k, \ell $. 
With \eqref{eq:rateAWGNIC17762341} and \eqref{eq:rate49030433}, we  have
\begin{align}
R_1   =    \Imu(v_1; y_1) - \Imu(v_1; y_2 | v_2 )       \geq    \frac{1- \alpha/2 - 2\epsilon}{2} \log P + o(\log P) \non
\end{align}
and  also $R_2 \geq   \frac{1- \alpha/2 - 2\epsilon}{2} \log P + o(\log P)$ resulting from symmetry, for  almost all the channel realizations. It suggests that  the proposed scheme achieves  $d_{\text{sum}}  = 2(1-\alpha/2)$  for almost all the channel realizations by using $d_c= \alpha /2$ GDoF of common randomness.  Note that in our scheme the common randomness is mapped into three random variables, i.e., $w'_1$, $w'_2$ and $u$. In this case,  the rate of   $w'_1$ is $R'_1 =  \Imu ( v_1; y_{2} | v_{2}) - \epsilon  \leq  o(\log P)  - \epsilon$   (see \eqref{eq:Rk623bJ} and \eqref{eq:rate49030433}); the rate of $w'_2$ is $R'_2 =  \Imu ( v_2; y_{1} | v_{1}) - \epsilon \leq   o(\log P)  - \epsilon$; and  the rate of  $u$ is $\Hen (u) =   \log (2  \cdot P^{ \frac{ \lambda_{u} }{2}} +1)  =  \frac{\lambda_{u}}{2} \log P + o(\log P)$, which gives $d_c= \lambda_{u} = \alpha /2$ when $\epsilon \to 0$.

\subsection{$1 < \alpha \leq 2$   \label{sec:CJschemeIC12}}

In this case with $1<  \alpha \leq  2$,  based on the  parameters designed in Table~\ref{tab:ICpara}, the transmitted signals take the forms as 
\begin{align}
 x_1  = &    v_{1,c}  +      \sum_{\ell=1}^{\tau}  \delta_{1,\ell}  \sqrt{P^{ -  \beta_{u_{\ell}}}}   \cdot   u   \label{eq:IC12x1}  \\
 x_2  = &     v_{2,c} +      \sum_{\ell=1}^{\tau}  \delta_{2,\ell}   \sqrt{P^{ -  \beta_{u_{\ell}}}}    \cdot   u   \label{eq:IC12x2}
 \end{align}
where in this case the parameters $\{\delta_{j,\ell}\}_{j, \ell}$ are designed as 
\begin{align} \label{eq: IC12delta}
\delta_{j,\ell} = \left\{ \begin{array}{ll}
- \frac{h_{ii}}{h_{ij}} \cdot   \big(\frac{h_{11} h_{22}}{h_{12} h_{21}}\big)^{\frac{\ell}{2}-1}  & \quad \textrm{$\ell \in \{2k: 2k\leq \tau, k\in \Zc^+\}$}\\
&\\
\big(\frac{h_{11} h_{22}}{h_{12} h_{21}}\big)^{\frac{\ell-1}{2}} & \quad  \textrm{$\ell \in \{ 2k-1: 2k-1 \leq \tau, k\in \Zc^+\}$}
\end{array} \right.
\end{align}
for $i, j \in \{1,2\}, i \not=j$. Similarly to the previous case,  the common randomness $u$ is used to generate a certain number of  signals with specific  directions and powers. 
Then, the received signals are given as
\begin{align}
y_{1} & =     \sqrt{P}  h_{11} v_{1,c}  + \underbrace{   \sqrt{P^{ \alpha }}  h_{12}     v_{2,c}  +    \sqrt{P^{\alpha-  \beta_{u_1}}} \delta_{2,1}      h_{12}  u}_{\rm{aligned}}\non\\ & \quad +   \sum_{\ell=1}^{\tau-1}  \underbrace{ (\sqrt{P^{1-  \beta_{u_{\ell}}}}   \delta_{1,\ell}   h_{11}  +  \sqrt{P^{ \alpha-  \beta_{u_{\ell+1}}}}  \delta_{2,\ell+1}      h_{12}   ) u}_{\rm{interference~neutralization}}+    \underbrace{ \sqrt{P^{\tau- (\tau -1)\alpha}}   \delta_{1, \tau}  h_{11}  u}_{\text{treated as noise}}+  z_{1}  \non\\
&  =      \sqrt{P}  h_{11} v_{1,c} +     \sqrt{P^{ \alpha }} h_{12}  (   v_{2,c}  +  u)  +    \sqrt{P^{\tau- (\tau -1)\alpha}}   \delta_{1, \tau}  h_{11}  u+  z_{1}  \label{eq:IC12y1}  \\
y_{2} & =     \sqrt{P}  h_{22} v_{2,c}    +   + \underbrace{   \sqrt{P^{ \alpha }}  h_{21}     v_{1,c}  +    \sqrt{P^{\alpha-  \beta_{u_1}}} \delta_{1, 1}     h_{21}  u}_{\rm{aligned}}\non\\ & \quad +  \sum_{\ell=1}^{\tau-1}  \underbrace{  (\sqrt{P^{1-  \beta_{u_{\ell}}}}   \delta_{2,\ell}   h_{22}  +  \sqrt{P^{ \alpha-  \beta_{u_{\ell+1}}}} \delta_{1,\ell+1}     h_{21}   ) u}_{\rm{interference~neutralization}}+    \underbrace{   \sqrt{P^{\tau- (\tau -1)\alpha}}   \delta_{2, \tau}     h_{22} u}_{\text{treated as noise}}+  z_{2}  \non\\
& =     \sqrt{P}  h_{22} v_{2,c}    +      \sqrt{P^{ \alpha }}  h_{21}  (   v_{1,c}  +  u)   +     \sqrt{P^{\tau- (\tau -1)\alpha}}   \delta_{2, \tau}     h_{22}u+  z_{2} . \label{eq:IC12y2} 
\end{align}
As can be seen from  the  above expressions of $y_1$ and $y_2$, the interference can be removed by using the  Markov chain-based interference neutralization method. 
In the end,  the interference   $\sqrt{P^{\tau- (\tau -1)\alpha}}   \delta_{1, \tau}  h_{11}  u$  and  the interference $ \sqrt{P^{\tau- (\tau -1)\alpha}}   \delta_{2, \tau}     h_{22}u $  can be treated as the noise terms at receiver~1 and receiver~2, respectively.
Based on our signal design, we will prove that the secure rates satisfy  $R_k  =     \Imu(v_k; y_k) -  \Imu ( v_k; y_{\ell} | v_{\ell} )   \geq \frac{ \alpha/2  - \epsilon}{2} \log P + o(\log P) $,  $k \neq \ell$, $k, \ell \in \{1,2\}$,  and  the secure sum GDoF $d_{\text{sum}}  = \alpha$  is achievable,  for almost all the channel coefficients  $\{h_{k\ell}\} \in (1, 2]^{2\times 2}$.  

Due to the symmetry we will focus on bounding the secure rate  $R_1 $ (see \eqref{eq:lboundit1}).  
Let $ \hat{v}_{1,c}$ be the  estimate for  $v_{1,c}$ from $y_1$,  and let $ \text{Pr} [  \{ v_{1,c} \neq \hat{v}_{1,c} \}] $ denote the corresponding error probability for this estimation.  By following the proof  steps of  Lemma~\ref{lm:ICrateerror537},  in this case with $1 < \alpha \leq 2$  one can prove that  
 \begin{align}
 \text{Pr} [  v_{k,c} \neq \hat{v}_{k,c}  ]  \to 0         \quad \text {as}\quad  P\to \infty  \label{eq:ICrateerror7534}
\end{align}
for almost all the  realizations of the channel coefficients. 
With  $v_{1,c} \in \Omega ( \xi = \frac{ \gamma}{Q} ,   \   Q =  P^{ \frac{\alpha/2 - \epsilon }{2}} )$ and  $v_1 =   v_{1,c}$,  the rate of $v_{1} $ is computed as   
  \begin{align}
 \Hen(v_{1}) =  \Hen(v_{1,c}) &=  \log (2 \cdot P^{ \frac{\alpha/2 - \epsilon }{2}} +1)  \label{eq:rate1329044}     
 \end{align}
and then, $\Imu(v_1; y_1)$ can be  bounded by
  \begin{align}
  \Imu(v_1; y_1)    &\geq  \bigl( 1 -   \text{Pr} [   v_{1,c} \neq \hat{v}_{1,c}  ] \bigr)  \Hen(v_{1})  - 1  \label{eq:ratedw39048}       \\
        & =  \frac{ \alpha/2 - \epsilon}{2} \log P + o(\log P)    \label{eq:rate1084376}  
 \end{align}
 for almost all the  realizations of the channel coefficients, where  
 \eqref{eq:ratedw39048} is derived by following the steps in \eqref{eq:rate471073}-\eqref{eq:rate2092434};
  and \eqref{eq:rate1084376}  uses the results of  \eqref{eq:ICrateerror7534} and  \eqref{eq:rate1329044}.
For the term $\Imu(v_1; y_2 | v_2 )$ in \eqref{eq:lboundit1},  by following the steps in \eqref{eq:rate042845523}-\eqref{eq:rate49030433} it can be  bounded as
\begin{align}
  \Imu(v_1; y_2 | v_2 )    
\leq  o(\log P) .  \label{eq:ratehd32837}  
 \end{align}
From  \eqref{eq:rate1084376} and \eqref{eq:ratehd32837}, we have 
\begin{align}
R_1   =    \Imu(v_1; y_1) - \Imu(v_1; y_2 | v_2 )   \geq  \frac{ \alpha/2  - \epsilon}{2} \log P + o(\log P)  \label{eq:lbounditfinal1132}
\end{align}
and   also  $R_2 \geq  \frac{ \alpha/2  - \epsilon}{2} \log P + o(\log P) $ resulting from symmetry,  for almost all the channel realizations.  By letting $\epsilon \to 0$,    the proposed scheme achieves $d_{\text{sum}}  =  \alpha$   for almost all  channel realizations by using $d_c= \alpha /2$ GDoF of common randomness.

 \subsection{$ \alpha =1$   \label{sec:CJschemeICeq}}

In this case with $ \alpha =1$,  based on the  parameters designed in Table~\ref{tab:ICpara},  and  by setting 
\begin{align}
\delta_{1, 1}= \frac{h_{22}h_{12}- h_{12}h_{21}}{h_{11} h_{22}- h_{12} h_{21}},  \quad \delta_{2, 1}= \frac{h_{11}h_{21}- h_{21}h_{12} }{ h_{22} h_{11}- h_{21}  h_{12}},
\end{align}
  the transmitted signals take the following forms
\begin{align}
 x_1  = &  \varepsilon  v_{1, c}  +  \varepsilon  \frac{h_{22}h_{12}- h_{12}h_{21}}{h_{11} h_{22}- h_{12} h_{21}} \cdot u    \label{eq:ICeqx1}  \\
 x_2  = & \varepsilon v_{2, c}   + \varepsilon    \frac{h_{11}h_{21}- h_{21}h_{12} }{ h_{22} h_{11}- h_{21}  h_{12}} \cdot u.     \label{eq:ICeqx2}
 \end{align}
Note that in this case, $\tau=1$ and $\varepsilon = \frac{h_{11} h_{22}- h_{12} h_{21}}{8}   $. 
Then, the received signals are simplified as
\begin{align}
y_{1} &=  \varepsilon   \sqrt{P}  h_{11} v_{1, c} + \varepsilon   \sqrt{P}  h_{12}  \underbrace{  (   v_{2, c}  +  u) }_{\text{aligned}} +  z_{1}  \label{eq:ICeqy1}  \\
y_{2} &=  \varepsilon  \sqrt{P}  h_{22} v_{2, c} + \varepsilon   \sqrt{P}  h_{21} \underbrace{ (   v_{1, c}  +  u)}_{\text{aligned}}   +  z_{2} .   \label{eq:ICeqy2} 
\end{align}
From the previous steps in \eqref{eq:rate471073}-\eqref{eq:rateAWGNIC17762341}, for almost all the realizations of the channel coefficients, the term $ \Imu(v_1; y_1)$   in \eqref{eq:lboundit1} can be  lower bounded  by 
  \begin{align}
  \Imu(v_1; y_1)   \geq    \frac{ 1/2 - \epsilon}{2} \log P + o(\log P) .   \label{eq:rateAWGNIC1748397} 
 \end{align}
From the derivations  in \eqref{eq:rate042845523}-\eqref{eq:rate49030433}, the term  $\Imu(v_1; y_2 | v_2 )$  in \eqref{eq:lboundit1} can be upper bounded by 
  \begin{align}
\Imu(v_1; y_2 | v_2 )  \leq  o(\log P).       \label{eq:rateAWGNIC973287}  
 \end{align}
   With \eqref{eq:rateAWGNIC1748397} and \eqref{eq:rateAWGNIC973287}, we have the following bounds on the secure rates 
\begin{align}
R_1   =    \Imu(v_1; y_1) - \Imu(v_1; y_2 | v_2 )   \geq  \frac{ 1/2  - \epsilon}{2} \log P + o(\log P)  \label{eq:lbounditfinal1132a}
\end{align}
and  $R_2 \geq  \frac{ 1/2  - \epsilon}{2} \log P + o(\log P) $ (due to the symmetry),  for almost all  channel realizations.  By letting $\epsilon \to 0$, the proposed scheme achieves  $d_{\text{sum}}    =  1$   for almost all  channel realizations by using $d_c= 1 /2$ GDoF of common randomness.

\subsection{$1/2 <  \alpha \leq  2/3$   \label{sec:CJschemeIC1223}}

 When $1/2 <  \alpha \leq  2/3$, the  signals of the transmitters have the same forms as in \eqref{eq:IC321x1} and \eqref{eq:IC321x2}, and the parameters $\{\delta_{j,\ell}\}_{j, \ell}$ are designed as in  \eqref{eq: IC321delta}. Since $\tau=2$ for this case,  the transmitted signals can be simplified as
 \begin{align}
   x_1  = &     v_{1,c} +    \sqrt{P^{ -  \alpha}}  \cdot   v_{1,p} +    \big ( \sqrt{P^{ \alpha - 1 }}  \cdot  \frac{h_{12}}{h_{11}} - \sqrt{P^{ 2\alpha - 2 }}   \cdot  \frac{h_{12}h_{21}}{h_{11}h_{22}} \big) u \label{eq:IC1223x1}  \\
 x_2  = &      v_{2,c} +     \sqrt{P^{ -  \alpha}}  \cdot   v_{2,p} +    \big ( \sqrt{P^{ \alpha - 1 }}  \cdot  \frac{h_{21}}{h_{22}} - \sqrt{P^{ 2\alpha - 2 }}   \cdot  \frac{h_{21}h_{12}}{h_{22}h_{11}} \big) u. \label{eq:IC1223x2}  
 \end{align}
The received signals then take the following forms
\begin{align}
y_{1} &=     \sqrt{P} h_{11} v_{1,c}  +     \sqrt{P^{ 1 - \alpha}} h_{11} v_{1,p}   +     \sqrt{P^{ \alpha }} h_{12}    \underbrace{ (   v_{2,c}  +  u)}_{\text{aligned}}  +   \underbrace{   h_{12} v_{2,p}  -     \sqrt{P^{ 3\alpha-2}}  \cdot   \frac{h^2_{12}h_{21}}{h_{11}h_{22}} u}_{\text{treated as noise}} +  z_{1}    \label{eq:IC1223y1}  \\
y_{2} &=       \sqrt{P} h_{22} v_{2,c}  +      \sqrt{P^{ 1 - \alpha}} h_{22} v_{2,p}   +     \sqrt{P^{ \alpha }} h_{21}  \underbrace{  (   v_{1,c}  +  u)}_{\text{aligned}} +  \underbrace{    h_{21} v_{1,p}  -     \sqrt{P^{ 3\alpha-2}}  \cdot   \frac{h^2_{21}h_{12}}{h_{22}h_{11}} u }_{\text{treated as noise}} +  z_{2} .   \label{eq:IC1223y2} 
\end{align}
 In the above expressions of $y_1$ and $y_2$,  the interference is removed by using the  Markov chain-based interference neutralization method.

We will focus on bounding the secure rate  $R_1 $.  
By following the derivations in \eqref{eq:rate471073}-\eqref{eq:rate2092434},  the term $\Imu(v_1; y_1)$ can be lower bounded by
 \begin{align}
  \Imu(v_1; y_1)    \geq  \bigl( 1 -   \text{Pr} [  \{ v_{1,c} \neq \hat{v}_{1,c} \} \cup  \{ v_{1,p} \neq \hat{v}_{1,p} \}  ] \bigr)   \cdot \Hen(v_{1})  - 1  \label{eq:rate9401529}     
 \end{align}
where the rate of  $v_{1} $ in \eqref{eq:rate9401529}  is
\begin{align}
\Hen(v_{1}) = \Hen(v_{1,c}) +\Hen(v_{1,p}) =  \frac{ \alpha  - 2\epsilon}{2} \log P + o(\log P) \label{eq: rateofv1}
\end{align} 
and the  error probability in \eqref{eq:rate9401529} is vanishing, stated below. 
 \begin{lemma}  \label{lm:rateerror48912}
With \eqref{eq:xvkkk1}-\eqref{eq:constellationGsym2}  and \eqref{eq:IC1223x1}-\eqref{eq:IC1223x2} and for $1/2 < \alpha \leq 2/3$, the  error probability of decoding  $\{v_{k,c}, v_{k,p} \}$ from $y_k$   is vanishing when $P$ goes large, that is,
 \begin{align}
 \text{Pr} [  \{ v_{k,c} \neq \hat{v}_{k,c} \} \cup  \{ v_{k,p} \neq \hat{v}_{k,p} \}  ]  \to 0         \quad \text {as}\quad  P\to \infty, \quad k=1,2 .   \label{eq:rateerror48912}
 \end{align}
 \end{lemma}

\begin{proof}
See Appendix~\ref{sec:rateerror48912}.
\end{proof}
From \eqref{eq:rate9401529}, \eqref{eq: rateofv1} and Lemma~\ref{lm:rateerror48912}, the term $ \Imu(v_1; y_1)$  can be  lower bounded  by 
  \begin{align}
  \Imu(v_1; y_1)   \geq   \frac{ \alpha  - 2\epsilon}{2} \log P + o(\log P). \label{eq:rateAWGNIC17762}  
 \end{align}
 From the derivations  in \eqref{eq:rate042845523}-\eqref{eq:rate49030433},  the term $\Imu(v_1; y_2 | v_2 )$  in \eqref{eq:lboundit1} is  bounded by 
  \begin{align}
\Imu(v_1; y_2 | v_2 ) \leq  o(\log P).    \label{eq:rateAWGNIC9562}  
 \end{align}
With  \eqref{eq:rateAWGNIC17762} and \eqref{eq:rateAWGNIC9562}, we conclude that 
\begin{align}
R_1   =    \Imu(v_1; y_1) - \Imu(v_1; y_2 | v_2 )   \geq \frac{\alpha  - 2\epsilon}{2} \log P + o(\log P)  \label{eq:lbounditfinal1132b}
\end{align}
and  also $R_2 \geq   \frac{\alpha  - 2\epsilon}{2} \log P + o(\log P)$, which imply that the proposed scheme achieves  $d_{\text{sum}}  = 2\alpha$,  by using $d_c= 2\alpha  -1$ GDoF of common randomness.

 \subsection{$\alpha \geq  2$   \label{sec:CJschemeIC2}}

When $\alpha \geq  2$, the signals of the transmitters have the same forms as in \eqref{eq:IC12x1} and \eqref{eq:IC12x2},   and the parameters $\{\delta_{j,\ell}\}_{j, \ell}$ are designed as in  \eqref{eq: IC12delta}. Since $\tau=2$ for this case, the transmitted signals can be simplified as
  \begin{align}
   x_1  = &     v_{1,c} +     \big ( 1- \sqrt{P^{ 1- \alpha }}  \cdot  \frac{h_{22}}{h_{21}} \big) u \label{eq:IC2x1}  \\
   x_2  = &     v_{2,c} +     \big ( 1- \sqrt{P^{ 1- \alpha }}  \cdot  \frac{h_{11}}{h_{12}} \big) u \label{eq:IC2x1b} 
 \end{align}
and the received signals can be simplified as 
\begin{align}
y_{1} &=     \sqrt{P} h_{11} v_{1,c}  +     \sqrt{P^{ \alpha }} h_{12}  \underbrace{(   v_{2,c}  +  u)}_{\text{aligned}}  -  \underbrace{   \sqrt{P^{ 2- \alpha }}  \cdot  \frac{h_{22}h_{11}}{h_{21}}  u}_{\text{treated as noise}} +  z_{1}    \label{eq:IC2y1}  \\
y_{2} &=     \sqrt{P} h_{22} v_{2,c}   +    \sqrt{P^{ \alpha }} h_{21}  \underbrace{ (   v_{1,c}  +  u)}_{\text{aligned}}  -  \underbrace{   \sqrt{P^{ 2- \alpha }}  \cdot  \frac{h_{11}h_{22}}{h_{12}}  u }_{\text{treated as noise}} +  z_{2} .   \label{eq:IC2y2} 
\end{align}
By following the proof  steps of  Lemma~\ref{lm:rateerror48912},  the  error probability for the estimation of $v_{k,c}$ from $y_k$, $k=1,2$, is proved to be vanishing, that is, $\text{Pr} [  v_{k,c} \neq \hat{v}_{k,c}  ]  \to 0$ as  $P\to \infty$.
Given that $v_1 =  v_{1,c}$ and  $v_{1,c} \in \Omega ( \xi =  \frac{ \gamma }{Q} ,   \   Q =  P^{ \frac{1 - \epsilon }{2}} )$,  the rate of $v_{1} $ is $\Hen(v_{1}) =  \Hen(v_{1,c}) =  \log (2 \cdot P^{ \frac{1  - \epsilon }{2}} +1)$. 
At this point, the term $ \Imu(v_1; y_1)$  in \eqref{eq:lboundit1}  is bounded as 
  \begin{align}
  \Imu(v_1; y_1)    \geq  &     \bigl( 1 -   \text{Pr} [   v_{1,c} \neq \hat{v}_{1,c}  ] \bigr)   \cdot \Hen(v_{1})  - 1 \non\\
= &\frac{ 1  - \epsilon}{2} \log P + o(\log P).  \label{eq:rateAWGNIC17762322}  
 \end{align}
Furthermore, from  \eqref{eq:rate042845523}-\eqref{eq:rate49030433} we can derive the upper bound of $\Imu(v_1; y_2 | v_2 )$ in \eqref{eq:lboundit1} as
  \begin{align}
\Imu(v_1; y_2 | v_2 ) \leq o(\log P).      \label{eq:rateAWGNIC9562322}   
 \end{align}
With \eqref{eq:rateAWGNIC17762322} and \eqref{eq:rateAWGNIC9562322}, the following bounds hold true:
\begin{align}
R_1   =    \Imu(v_1; y_1) - \Imu(v_1; y_2 | v_2 )   \geq  \frac{ 1  - \epsilon}{2} \log P + o(\log P) \non 
\end{align}
and   $R_2 \geq    \frac{ 1  - \epsilon}{2} \log P + o(\log P)$, which suggest that the proposed scheme achieves $d_{\text{sum}}= 2$   by using $d_c= 1$ GDoF of common randomness.

\section{Achievability  for wiretap channel  with a helper  \label{sec:CJGauwiretap} }

In this section, we provide the achievability  scheme for a Gaussian WTH channel (see~Section~\ref{sec:syswchr}).  
The proposed scheme uses PAM modulation,   Markov chain-based  interference neutralization and  alignment technique  in the signal design. 
For the case with  $0 \leq  \alpha \leq 1/2$, the optimal secure sum GDoF  $\dso (\alpha) = 1$ is achievable without adding common randomness (cf.~\cite{CG:18arxiv}).  Therefore,  here we  will just focus on the case with  $\alpha > 1/2$ and prove that $d (\alpha) =  1$ is achievable.  The scheme details are given as follows.
\subsubsection{Codebook generation}
The codebook generation is similar to the previous case for the interference channel with confidential messages, with one difference being that only transmitter~$1$ is required to generate the codebook in this channel. Note that in this channel transmitter~2 will act as a helper without sending message.   
For transmitter~$1$,  it generates a codebook given by
  \begin{align}
     \mathcal{B} \defeq \Bigl\{  v^{\bln} (w_1,  w'_1):  \  w_{1} \in \{1,2, 3, \cdots, 2^{\bln R_{1}}\},   
      w'_1 \in \{1,2, 3, \cdots, 2^{\bln R'_1}\}   \Bigr\}     \label{eq:code23418937Jwiretap}
     \end{align}
   which is  similar to that in \eqref{eq:code2341J}. 
To transmit the confidential message $w_{1}$, transmitter~$1$  chooses a codeword $v^{\bln}$   \emph{randomly} from  a sub-codebook  $\mathcal{B}( w_{1}) $   defined  by 
\[   \mathcal{B} (w_{1})  \defeq \bigl\{ v^{\bln} (w_{1},  w'_1): \  w'_1 \in \{1,2,\cdots, 2^{\bln R'_1}\}   \bigr\}  \]
 according to a uniform distribution (similar to \eqref{eq: codebook89284}).
Then the selected codeword $v^{\bln}$ is mapped to the channel input under the following signal design
 \begin{align}
  x_{1}(t) =     \varepsilon  v (t)    + \varepsilon \sum_{\ell=0}^{\tau}  \delta_{1,\ell}   \sqrt{P^{ -  \beta_{u_{\ell}}}} \cdot   u (t)      \label{eq:xvkkkwiretap}
   \end{align}
where  $\{\delta_{k, \ell}\}_{k, \ell} $ are the parameters which will be  specified  later  by using the Markov chain-based  interference neutralization and  alignment technique.  $\varepsilon$ is a parameter designed as   
 \begin{align}\label{eq:varadef1122}
 \varepsilon \defeq \begin{cases}  
1      & \quad  \rm{ if~\alpha \neq 1}\\
\frac{h_{11} h_{22}- h_{12} h_{21}}{8}     & \quad  \rm{ if~\alpha = 1}.
\end{cases}
   \end{align}
$\tau$ is another parameter designed as 
   \begin{align} \label{eq:tau9787}
 \tau \defeq \begin{cases}  
 1   & \quad  \rm{ if~\alpha = 1} \\
 \big \lceil \frac{1}{2(1-\alpha)} \big \rceil    & \quad  \rm{ if~\alpha < 1}\\
 \big \lceil \frac{1}{2(\alpha-1)} \big \rceil   & \quad  \rm{ if~\alpha > 1}. \\
\end{cases}
   \end{align}
 $u$ is a random variable \emph{independently}  and  \emph{uniformly}  drawn from a  PAM constellation set which will be specified later on.  In this channel, the common randomness $w_c$ is mapped into two random variables, i.e.,  $w'_1$ and $u$, such that $\Hen(w_c) = \Hen(w'_1) +  \Hen(u)$ and  $w'_1$ and $u$ are mutually independent.  Based on our definition, $w'_1$ and $u$ are  available at the transmitters but not at the receivers.

\subsubsection{Signal design}

In the scheme,  each element of the codeword $v^{\bln}$ in \eqref{eq:code23418937Jwiretap} is designed to take the following form
 \begin{align}
   v(t)  =   v_ c(t) +   \sqrt{P^{ - \beta_{p}}}    \cdot  v_{p} (t)  \label{eq:xvkwiretap}  
 \end{align}
 which gives  the  channel input in~\eqref{eq:xvkkkwiretap}  as 
 \begin{align}
  x_{1}  =   \varepsilon    v_{c} +   \varepsilon   \sqrt{P^{ - \beta_{p}}}    \cdot  v_{p}  + \varepsilon \sum_{\ell=0}^{\tau}  \delta_{1,\ell}   \sqrt{P^{ -  \beta_{u_{\ell}}}} \cdot   u     \label{eq:xvkkk1wiretap}  
 \end{align}
(removing the time index for simplicity).  At transmitter~2, it sends  the jamming  signal designed as 
 \begin{align}
  x_{2}  =   \varepsilon   \sum_{\ell=1}^{\tau}  \delta_{2,\ell}    \sqrt{P^{ -  \beta'_{u_{\ell}}}}  \cdot u  \label{eq:xvkkk1wiretap222}  
 \end{align}
where the random variables $u$, $v_{c}$ and  $v_{p}$   are \emph{independently} and \emph{uniformly}  drawn from the corresponding PAM constellation sets
 \begin{align}
   v_{c},  u    &  \in    \Omega ( \xi =  \frac{  \gamma}{Q} ,   \   Q =  P^{ \frac{ \lambda_{c} }{2}} )  \label{eq:constellationGsym1wiretap}   \\ 
   v_{p}      &  \in    \Omega ( \xi =  \frac{\gamma}{2Q} ,   \   Q = P^{ \frac{  \lambda_{p} }{2}} )    \label{eq:constellationGsym2wiretap}    
 \end{align}
  and $\gamma$ is a parameter satisfying the constraint
   \begin{align}
   \gamma  \in \bigl(0,  \ \frac{1}{(\tau +2)\cdot 4^\tau}\big].    \label{eq:gammaWH999}    
 \end{align}
 Table~\ref{tab:wiretappara} provides the designed parameters $\{\beta_{p}, \beta_{u_{\ell}}, \beta'_{u_{\ell}}, \lambda_{c},  \lambda_{p}, \lambda_{u}\}$  under  different regimes. 
Based on our signal design (see \eqref{eq:constellationGsym1wiretap}-\eqref{eq:constellationGsym2wiretap}), and  by following the steps in \eqref{eq:power35781}-\eqref{eq:power3578} we have
\begin{align}
 \E |v_{c}|^2  =  \E |u|^2 \leq   \frac{ 2  \gamma^2 }{3} \quad {\rm{and}}   \quad  \E |v_{p}|^2    \leq   \frac{    \gamma^2 }{6}.  \label{eq:constellationcalculation}
\end{align} 
By combining \eqref{eq:xvkkk1wiretap},  \eqref{eq:xvkkk1wiretap222} and \eqref{eq:constellationcalculation},  one can check that  the power constraints  $\E |x_1|^2 \leq 1$ and   $\E |x_2|^2 \leq 1$  are satisfied.

\subsubsection{Secure rate analysis} We define  the rates  $R_1$ and $R'_1$ as 
\begin{align}
R_1 &\defeq   \Imu(v; y_1) -  \Imu ( v; y_2 ) - \epsilon   \label{eq:Rk623Jwiretap} \\  
R'_1  &\defeq  \Imu ( v; y_2) - \epsilon .  \label{eq:Rk623bJwiretap}  
\end{align}
With our  codebook  and signal design,  the result of  \cite[Theorem~2]{XU:15}   (or \cite[Theorem~2]{LMSY:08}) suggests that the rate $R_1$ defined in  \eqref{eq:Rk623Jwiretap} is achievable  and the transmission of the message $w_1$  is secure, that is,   $\Imu(w_1; y_{2}^{\bln})  \leq  \bln \epsilon$.
For this WTH channel, it  can be treated as  a particular case of the IC-SC channel  if we remove the second transmitter's message (or set it empty). Thus,  the result of \cite[Theorem~2]{XU:15}   (or \cite[Theorem~2]{LMSY:08}) derived for  the IC-SC channel reveals that the secure rate  $R_1$ defined in  \eqref{eq:Rk623Jwiretap} is achievable in this WTH channel.

In what follows we will provide the rate analysis, focusing on the regime of  $ \alpha > 1/2$.  Specifically, we will consider the following three cases:     $\frac{1}{2} <  \alpha < 1$,  $ \alpha >1 $ and $\alpha=1$. 
In the achievability scheme, a Markov chain-based  interference neutralization method is proposed.

\begin{table}
\caption{Parameter design for the WTH channel.}
\begin{center}
{\renewcommand{\arraystretch}{2}
\begin{tabular}{|c|c|c|c|}
  \hline
                     &   $1/2<  \alpha <1$  &  $\alpha =1$  &   $\alpha >1$      \\
    \hline
     $\beta_{u_{0}}$    				  &    $\infty$    &   $\infty  $ &   $0$      \\
    \hline
  $\beta_{u_{\ell}}, \  \ell \in\{1, 2, \cdots, \tau-1\}$    				  &    $2\ell(1-\alpha)$    &   $0  $ &   $2\ell(\alpha-1)$      \\
    \hline
  $\beta_{u_{\tau}}$    				  &    $\infty$    &   $0  $ &   $2\tau(\alpha-1)$      \\
    \hline
 $\beta'_{u_{\ell}}, \  \ell \in\{1, 2, \cdots, \tau\} $    				  &    $(2\ell-1)(1-\alpha) $    &   $0 $ &   $ (2\ell-1)(\alpha-1) $      \\
    \hline
       $\beta_{p} $    		   			   &       $\alpha $    &   $\infty $ &   $ \infty$       \\
    \hline 
   $\lambda_{c}$ 			            &   $\alpha - \epsilon$     &  $  1   - \epsilon$  &  $1  - \epsilon$          \\
  \hline
   $\lambda_{p}$   	    &     $1- \alpha - \epsilon $  &   $0$ &  $0$        \\
    \hline
     $\lambda_{u}$ 			            &   $\alpha - \epsilon$     &  $  1   - \epsilon$  &  $1  - \epsilon$          \\
  \hline
    \end{tabular}
}
\end{center}
\label{tab:wiretappara}
\end{table}

\subsection{$ 1/2 < \alpha < 1$   \label{sec:CJschemewiretap1}}

When $ 1/2 <  \alpha <  1$, the parameters $\{\delta_{k,\ell} \}_{k, \ell}$ are designed  as
\begin{align}
\delta_{1,\ell} &= -  \big(\frac{h_{12} h_{21}}{h_{11} h_{22}}\big)^{\ell} \   \quad \quad  \quad  \quad \ell \in \{ 1, 2, \cdots, \tau-1\} \non\\
 \delta_{2,\ell} & =  \frac{h_{21}}{h_{22}} \cdot   \big(\frac{h_{12} h_{21}}{h_{11} h_{22}}\big)^{\ell-1}   \quad  \quad \ell \in \{1, 2, \cdots, \tau\}.  \non
\end{align}
 In this case,  the transmitted signals take the forms   as
\begin{align}
 x_1  = &    v_{c}  +    \sqrt{P^{ -  \alpha}}  \cdot   v_{p} +    \sum_{\ell=1}^{\tau-1}  \delta_{1,\ell}    \sqrt{P^{ -  \beta_{u_{\ell}}}} \cdot  u  \label{eq:wiretap1x1}  \\
 x_2  = &      \sum_{\ell=1}^{\tau}  \delta_{2,\ell}    \sqrt{P^{ -  \beta'_{u_{\ell}}}}  \cdot u  . \label{eq:wiretap1x2}
 \end{align}
 Note that  the common randomness $u$ is used to generate a certain number of  signals with specific  directions and powers, i.e.,  $\{  \delta_{1, \ell}  \sqrt{P^{ -  \beta_{u_{\ell}}}} u\}_{\ell=0}^{\tau}$ at transmitter~1 and $\{  \delta_{2, \ell}  \sqrt{P^{ -  \beta'_{u_{\ell}}}} u\}_{\ell=1}^\tau$ at transmitter~2.
Then, the received signals are expressed as
\begin{align}
y_{1} &=     \sqrt{P}  h_{11} v_{c}  +      \sqrt{P^{ 1 - \alpha}}  h_{11} v_{p}   +    \sum_{\ell=1}^{\tau-1} \underbrace{ (\sqrt{P^{1-  \beta_{u_{\ell}}}}   \delta_{1,\ell}   h_{11}  \!+\!  \sqrt{P^{ \alpha-  \beta'_{u_{\ell}}}} \delta_{2,\ell}     h_{12}   ) u}_{\rm{interference~neutralization}} +      \underbrace{\sqrt{P^{ 2 \tau \alpha +1 -2\tau}}   \delta_{2, \tau}   h_{12}  u}_{\text{treated as noise}} \!+\! z_{1}  \non\\ 
&=     \sqrt{P}  h_{11} v_{c}  +      \sqrt{P^{ 1 - \alpha}}  h_{11} v_{p}  +      \sqrt{P^{ 2 \tau \alpha +1 -2\tau}}   \delta_{2, \tau}   h_{12}  u + z_{1}  \label{eq:wiretap1y1}  \\
y_{2} &=   \underbrace{   \sqrt{P^{ \alpha }}  h_{21}     v_{c}  +    \sqrt{P^{1-  \beta'_{u_1}}} \delta_{2,1}      h_{22}  u}_{\rm{aligned}}  +   h_{21} v_{p}  +    \sum_{\ell=1}^{\tau-1} \underbrace{ (\sqrt{P^{ \alpha-  \beta_{u_{\ell}}}}   \delta_{1,\ell}   h_{21}  +  \sqrt{P^{1-  \beta'_{u_{\ell+1}}}}  \delta_{2,\ell+1}      h_{22}   ) u}_{\rm{interference~neutralization}} +  z_{2}  \non\\
&=    \sqrt{P^{ \alpha }}  h_{21}  (   v_{c}  +  u) +   h_{21} v_{p}  + z_{2}   . \label{eq:wiretap1y2} 
\end{align}

At the receivers, a Markov chain-based  interference neutralization method is used to remove the interference.  In the above expressions  of $y_1$ and $y_2$, we can see that  
the signal $\delta_{2,1}    \sqrt{P^{ -  \beta'_{u_{1}}}}   u  $ from transmitter~2  is used to 
jam the information signal $v_{c}$ at receiver~2 but it will create an interference at receiver~1; this interference will be neutralized by the signal  $\delta_{1, 1}  \sqrt{P^{ -  \beta_{u_{1}}}}  u$ from transmitter~1; the added signal $\delta_{1, 1}  \sqrt{P^{ -  \beta_{u_{1}}}}  u$ also creates another signal at receiver~2 but will be neutralized by the next generated signal $\delta_{2,2}    \sqrt{P^{ -  \beta'_{u_{2}}}}  u  $; this process repeats until the residual interference can be treated as noise.

For the proposed scheme, the rate $R_1$ defined in  \eqref{eq:Rk623Jwiretap} can be achieved, which can be expressed as   
\begin{align}
R_1  & =     \Imu(v; y_1) -  \Imu ( v; y_{2} )      \label{eq:lbrate1783} 
\end{align}
for $\epsilon \to 0$.  In the following we will bound  the secure rate $R_1$. 
Let $\hat{v}_{c}$  and $\hat{v}_{p}$ be the estimates of $v_{c}$ and  $v_{p}$  respectively from  $y_1$, and let $\text{Pr} [  \{ v_{c} \neq \hat{v}_{c} \} \cup  \{ v_{p} \neq \hat{v}_{p} \}  ] $ denote the corresponding error  probability of this estimation.
With this,  the term $\Imu(v; y_1)$ in \eqref{eq:lbrate1783} has the following bound
  \begin{align}
  \Imu(v; y_1)   &\geq   \Imu(v; \hat{v}_{c}, \hat{v}_{p})  \label{eq:rate5544}     \\
  &=   \Hen(v) -   \Hen(v  |  \hat{v}_{c}, \hat{v}_{p})    \non    \\
    &\geq    \Hen(v) -       \bigl( 1+    \text{Pr} [  \{ v_{c} \neq \hat{v}_{c} \} \cup  \{ v_{p} \neq \hat{v}_{p} \}  ] \cdot \Hen(v) \bigr)  \label{eq:rate28536}     \\
        & =     \bigl( 1 -   \text{Pr} [  \{ v_{c} \neq \hat{v}_{c} \} \cup  \{ v_{p} \neq \hat{v}_{p} \}  ] \bigr)   \cdot \Hen(v)  - 1  \label{eq:rate2256}     
 \end{align}
where  \eqref{eq:rate5544} stems from  the Markov property of $v \to y_1 \to  \{\hat{v}_{c}, \hat{v}_{p} \}$; and  \eqref{eq:rate28536} stems  from Fano's inequality.
Given  that    $v_{c} \in    \Omega (\xi   =  \frac{ \gamma}{Q},   \   Q =  P^{ \frac{ \alpha  - \epsilon}{2}} ) $ and  $v_{p}  \in    \Omega (\xi   =  \frac{ \gamma}{2Q},   \   Q = P^{ \frac{ 1 - \alpha - \epsilon}{2}} ) $,   we have 
   \begin{align}
 \Hen(v) & =  \Hen(v_{c}) +   \Hen(v_{p})  \non\\ 
        &=  \log (2 \cdot P^{ \frac{ \alpha  - \epsilon}{2}} +1)    +  \log (2 \cdot P^{ \frac{ 1 - \alpha - \epsilon}{2}} +1)   \non\\
        &= \frac{ 1  - 2\epsilon}{2} \log P + o(\log P)  .      \label{eq:rate39737}                              
 \end{align}  
Note  that, with our signal design,  we can construct  $\{ v_{c}, v_{p}\}$ based on $v$,  and vice versa. 
Below we have a result on the error probability appeared in \eqref{eq:rate2256}.
 \begin{lemma}  \label{lm:errorcase1}
 With design in \eqref{eq:xvkkk1wiretap}-\eqref{eq:constellationGsym2wiretap} and Table~\ref{tab:wiretappara}, for $1/2 < \alpha < 1$,  the error probability of decoding $\{v_{c}, v_{p} \}$ based on $y_1$ is vanishing when $P$ is large, i.e.,
 \begin{align}
\text{Pr} [  \{ v_{c} \neq \hat{v}_{c} \} \cup  \{ v_{p} \neq \hat{v}_{p} \}  ]    \to 0         \quad \text {as}\quad  P\to \infty .   \label{eq:epcase1}
 \end{align}
 \end{lemma}
\begin{proof}
See~Appendix~\ref{sec:errorcase1}. 
\end{proof}

By incorporating the results of  Lemma~\ref{lm:errorcase1} and  \eqref{eq:rate39737} into  \eqref{eq:rate2256},  it gives 
  \begin{align}
\Imu(v; y_1)  &\geq   \frac{ 1  - 2\epsilon}{2} \log P + o(\log P) .   \label{eq:rb286227}  
 \end{align}
For the last term appeared in  \eqref{eq:lbrate1783}, we can treat it  as a rate penalty. 
This penalty can be  bounded by 
  \begin{align}
  \Imu ( v; y_{2} )    
\leq &  \Imu(v; y_2,  v_{c} + u )    \non \\
=   & \Imu(v;    v_{c} + u  )   +  \Imu(v;    h_{21}  v_{p}        +  z_{2}   |  v_{c} + u )    \label{eq:rb337257}   \\
=   & \Hen( v_{c} + u )  - \Hen(u  )      +  \hen(   h_{21}  v_{p}        +  z_{2}  )   -  \hen(   z_{2} )   \non    \\
\leq &   \underbrace{\log (4  \cdot P^{ \frac{ \alpha - \epsilon}{2}} +1)    -   \log (2  \cdot P^{ \frac{ \alpha  - \epsilon}{2}} +1) }_{\leq 1 }    +   \frac{1}{2} \log ( 2 \pi e (  \underbrace{| h_{21}|^2        +  1 }_{\leq 5}))       -  \frac{1}{2}\log (2 \pi e)    \label{eq:rb2526}   \\
\leq  & \log (2 \sqrt{5})   \label{eq:rb26277}  
 \end{align}
where  
\eqref{eq:rb2526} uses an identity that $v_{c} + u \in  \Omega (\xi   = \gamma \cdot P^{ - \frac{ \alpha  - \epsilon}{2}},   \   Q = 2\cdot  P^{ \frac{ \alpha  - \epsilon}{2}} ) $.
Finally,  by  incorporating the results of  \eqref{eq:rb286227}  and \eqref{eq:rb26277} into \eqref{eq:lbrate1783}, we have
\begin{align}
R_1       \geq   \frac{ 1  - 2\epsilon}{2} \log P + o(\log P)   \label{eq:rbcase1f}   
\end{align}
which suggests that when  $1/2 < \alpha < 1$ the proposed scheme achieves  $d (\alpha)  =  1$  by using $d_c = \alpha$ GDoF of common randomness (mainly due to $u$, and $d_c = \lambda_{u} = \alpha - \epsilon$ with $\epsilon \to 0$). 

\subsection{$ \alpha >1$   \label{sec:CJschemewiretap2}}

In this case with $  \alpha >1 $,   the parameters $\{\delta_{k,\ell} \}_{k, \ell}$ are designed  as
\begin{align}
\delta_{1,\ell} &=  \big(\frac{h_{11} h_{22}}{h_{12} h_{21}}\big)^{\ell} \quad  \quad  \quad \quad  \quad   \quad   \quad \quad  \ell \in \{ 0, 1, 2, \cdots, \tau\}    \label{eq:deltaWH888}\\
 \delta_{2,\ell} & = - \frac{h_{11}}{h_{12}} \cdot   \big(\frac{h_{11} h_{22}}{h_{12} h_{21}}\big)^{\ell-1}  \quad  \quad  \quad  \quad \ell \in \{1, 2, \cdots, \tau\} .   \label{eq:deltaWH777}
\end{align}
The transmitted signals in this case have the following expressions 
\begin{align}
 x_1  = &    v_{c}  +  \sum_{\ell=0}^{\tau}  \delta_{1,\ell}  \sqrt{P^{ -  \beta_{u_{\ell}}}}   \cdot   u   \label{eq:wiretap2x1}  \\
 x_2  = &    \sum_{\ell=1}^{\tau}  \delta_{2,\ell}   \sqrt{P^{ -  \beta'_{u_{\ell}}}}    \cdot   u  .  \label{eq:wiretap2x2}
 \end{align}
 Similarly to the previous case,  the common randomness $u$ is used to generate a certain number of  signals with specific  directions and powers. 
Then, the received signals are expressed as
\begin{align}
y_{1} &=      \sqrt{P}  h_{11} v_{c}  +      \sum_{\ell=0}^{\tau-1} \underbrace{ (\sqrt{P^{1-  \beta_{u_{\ell}}}}   \delta_{1,\ell}   h_{11}  +  \sqrt{P^{ \alpha-  \beta'_{u_{\ell+1}}}}  \delta_{2,\ell+1}      h_{12}   ) u}_{\rm{interference~neutralization}}+  \underbrace{  \sqrt{P^{-2\tau \alpha+2\tau +1}}   \delta_{1, \tau}  h_{11}  u}_{\text{treated as noise}} +  z_{1}  \non\\
&=    \sqrt{P}  h_{11} v_{c}   +    \sqrt{P^{-2\tau \alpha+2\tau +1}}   \delta_{1, \tau}  h_{11}  u +  z_{1}   \label{eq:wiretap2y1}  \\
y_{2} &=  \underbrace{   \sqrt{P^{ \alpha }}  h_{21}     v_{c}  +    \sqrt{P^{ \alpha -  \beta_{u_0}}} \delta_{1, 0}      h_{21}  u}_{\rm{aligned}} +    \sum_{\ell=1}^{\tau} \underbrace{ (\sqrt{P^{ \alpha-  \beta_{u_{\ell}}}}   \delta_{1,\ell}   h_{21}  +  \sqrt{P^{1-  \beta'_{u_{\ell}}}} \delta_{2,\ell}     h_{22}   ) u}_{\rm{interference~neutralization}}  +  z_{2} \non\\
&=    \sqrt{P^{ \alpha}}  h_{21}  (   v_{c}  +  u)  +  z_{2}  .   \label{eq:wiretap2y2} 
\end{align}

At the receivers, a Markov chain-based  interference neutralization method is used to remove the interference.  
From the expression of \eqref{eq:wiretap2y1}, we can estimate  $v_{c}$  from $y_1$ by treating the other signal as noise.  Note that in this case $v_{c}   \in    \Omega ( \xi =  \frac{  \gamma}{Q} ,   \   Q =  P^{ \frac{ 1- \epsilon }{2}} ) $,  and the term  $ \sqrt{P^{-2\tau \alpha+2\tau +1}}   \delta_{1, \tau}  h_{11}  u$ in \eqref{eq:wiretap2y1} can be treated as noise because  $| \sqrt{P^{-2\tau \alpha+2\tau +1}}   \delta_{1, \tau}  h_{11}  u |  \leq    4^\tau  \times 2  \times \gamma      \leq  \frac{4\alpha -4}{4\alpha -3}$,  recalling that  $| u| \leq \gamma$  (cf.~\eqref{eq:constellationGsym1wiretap}),  $\gamma  \in \bigl(0, \frac{1}{(\tau+2)4^\tau}]$  (cf.~\eqref{eq:gammaWH999})  and  $|\delta_{1, \tau}| \leq 4^\tau$  (cf.~\eqref{eq:deltaWH888}).  Therefore, by treating the other signal as noise, one can easily prove that the error probability of decoding $v_{c}$ from $y_1$  is vanishing when $P$ is large,  that is,  
\begin{align}
 \text{Pr} [  v_{c} \neq \hat{v}_{c}  ]   \to 0         \quad \text {as}\quad  P\to \infty  \label{eq:error po3898}
\end{align}
where $\hat{v}_{c}$ is the estimate of $v_{c}$. 

Note that the rate of $v$ is computed as   $\Hen(v)  =  \Hen(v_{c})  = \frac{    1 - \epsilon}{2} \log P + o(\log P) $.
 From \eqref{eq:rate5544}-\eqref{eq:rate2256},  we can bound  $\Imu(v; y_1)$ appeared in \eqref{eq:lbrate1783}  as 
 \begin{align}
\Imu(v; y_1) &  \geq  \bigl( 1 -   \text{Pr} [ \{ v_{c} \neq \hat{v}_{c} \}   ] \bigr)   \cdot \Hen(v)  - 1 \non\\
  &=    \bigl( 1 -   \text{Pr} [ \{ v_{c} \neq \hat{v}_{c} \} ] \bigr)   \cdot  \frac{    1- \epsilon}{2} \log P   + o(\log P) .
   \label{eq:rate34647} 
\end{align} 
With  \eqref{eq:error po3898}  and    \eqref{eq:rate34647},  it gives
\begin{align}
 \Imu(v; y_1) \geq  \frac{    1- \epsilon}{2} \log P   + o(\log P).   \label{eq:rate3464779357}
\end{align}
From \eqref{eq:rb337257}-\eqref{eq:rb26277}, we can also bound  $ \Imu ( v; y_{2} )$  in \eqref{eq:lbrate1783} as
  \begin{align}
  \Imu ( v; y_{2} )  \leq   o(\log P).    \label{eq:rb26277897}  
 \end{align}
Given \eqref{eq:rate3464779357} and \eqref{eq:rb26277897}, we have  
 \[R_1 \geq \frac{    1- \epsilon}{2} \log P   + o(\log P). \] 
 This suggests that, when  $\alpha >1$,  the proposed scheme achieves  $d(\alpha) =  1$ by using $d_c = 1$ GDoF of common randomness.

\subsection{$ \alpha =1$   \label{sec:CJschemewiretapeq}}

When $ \alpha =1$,  by setting
\[\delta_{1, 1}=  - \frac{h_{12}h_{21}}{h_{11} h_{22}- h_{12} h_{21}}, \quad \delta_{2, 1}=  \frac{h_{11}h_{21}}{h_{11} h_{22}- h_{12} h_{21}}, \]
  the transmitted signals take the following forms
\begin{align}
 x_1  = &  \varepsilon v_{c}  -\varepsilon \frac{h_{12}h_{21}}{h_{11} h_{22}- h_{12} h_{21}} \cdot u    \label{eq:wiretapeqx1}  \\
 x_2  = & \varepsilon \frac{h_{11}h_{21}}{h_{11} h_{22}- h_{12} h_{21}} \cdot u  .   \label{eq:wiretapeqx2}
 \end{align}
Note that in this case, $\tau=1$ and $\varepsilon = \frac{h_{11} h_{22}- h_{12} h_{21}}{8} $.
Then, the received signals are expressed as
\begin{align}
y_{1} &=   \varepsilon \sqrt{P}  h_{11} v_{c} +  z_{1}  \label{eq:wiretapeqy1}  \\
y_{2} &=  \varepsilon  \sqrt{P}  h_{21}  (   v_{c}  +  u)   +  z_{2} .  \label{eq:wiretapeqy2} 
\end{align}
Let us consider the bound of the secure rate $R_1$ expressed in \eqref{eq:lbrate1783}. In this case, the rate of $v$ is computed as   $\Hen(v)  =  \Hen(v_{c})  = \frac{    1 - \epsilon}{2} \log P + o(\log P) $.
 By following the steps  in \eqref{eq:rate5544}-\eqref{eq:rb286227},  $\Imu(v; y_1)$   in \eqref{eq:lbrate1783} is   bounded by 
\begin{align}
 \Imu(v; y_1) \geq  \frac{1- \epsilon}{2} \log P   + o(\log P).   \label{eq:rate346477935897}
\end{align}
From  \eqref{eq:rb337257}-\eqref{eq:rb26277}, we can also bound  $ \Imu ( v; y_{2} )$  in \eqref{eq:lbrate1783} as
  \begin{align}
  \Imu ( v; y_{2} )      \leq  1.   \label{eq:rb262778497897}  
 \end{align}
At this point,  given \eqref{eq:rate346477935897} and \eqref{eq:rb262778497897}, we have $R_1  \geq  \frac{    1- \epsilon}{2} \log P   + o(\log P)$. This bound suggests  that, when $\alpha =1$, the proposed scheme achieves   $d(\alpha)=  1$   by using $d_c = 1$ GDoF of common randomness.

\section{Achievability  for multiple access wiretap channel  \label{sec:CJGauMAC} }

In this section, we will provide the achievability  proof of  Theorem~\ref{thm:GDoFmawc}  for MAC-WT  channel defined in Section~\ref{sec:sysmawc}. The following lemma will be used in the proof.

 \begin{lemma}  \label{lm:mac1switchalpha}
Given  the symmetric  Gaussian  MAC-WT channel  defined in Section~\ref{sec:sysmawc},  for any tuple $(d^{'}_1, d^{'}_2, d^{'}_c)$ such that  $(d^{'}_1, d^{'}_2, d^{'}_c) \in \bar{\mathcal{D}} ( \alpha)$,  then 
 \begin{align}
(\frac{1}{\alpha}d^{'}_2, \frac{1}{\alpha}d^{'}_1, \frac{1}{\alpha}d^{'}_c) \in \bar{\mathcal{D}} (\frac{1}{\alpha}). 
 \label{eq:mac1switchalpha}
 \end{align}
 \end{lemma}
\begin{proof}
See Appendix~\ref{sec:mac1switchalphaproof}.
\end{proof}

In what follows, we will first focus on  the case of $0 \leq \alpha \leq 1$ and prove  that the optimal secure GDoF region ${\mathcal{D}}^{*} ( \alpha) = \{(d_1, d_2) | d_1 + d_2 \leq \max\{1, \alpha\}, 0 \leq d_1 \leq   1, 0 \leq d_2 \leq \alpha\}$ is achievable for the $0 \leq \alpha \leq 1$ case\footnote{When $\alpha=0$, ${\mathcal{D}} ( 0) = \{(d_1, d_2) | 0 \leq d_1 \leq   1,  d_2 =0 \}$  is achievable by  using a  single-user transmission scheme. Therefore, without loss of generality, we will focus on the case with $\alpha >0$.}. In Section~\ref{sec:macwiretapschemeaaa000},  we will prove that ${\mathcal{D}}^{*} ( \alpha)$ is achievable for $\alpha>1$ by using the result of  Lemma~\ref{lm:mac1switchalpha}.
In the proposed scheme,  pulse  amplitude modulation,   Markov chain-based  interference neutralization and  alignment technique will be used in the signal design. The details of the proposed scheme are given as follows.

\subsubsection{Codebook}

The codebook generation is the same as that of the interference channel in Section~\ref{sec:CJGauIC} (see \eqref{eq:code2341J} and \eqref{eq: codebook89284}). 
In this setting, the channel input takes the following form 
 \begin{align}
  x_{k}(t) =     \varepsilon  v_{k} (t)    + \varepsilon \sum_{{\ell}=1}^{\tau/2}  \delta_{k, {\ell}}  \sqrt{P^{ -  \beta_{1, k, {\ell}} }}   \cdot  u_1(t)    +  \varepsilon \sum_{{\ell}=1}^{\tau/2}  \delta_{k, {\ell}}    \sqrt{P^{-  \beta_{2, k, {\ell}} }}   \cdot   u_2 (t)       \label{eq:macwt2x189283} 
\end{align}
for $ k=1,2$, where $v_{k} (t)$ denotes the $t$th element of  codeword; $\{ \delta_{k,  {\ell}}\}_{k, \ell}$ and $\tau$ are  parameters that will be designed specifically later on for different cases of $\alpha$; $\varepsilon$  is a parameter designed as  \begin{align}\label{eq:varadef2222}
 \varepsilon \defeq \begin{cases}  
1      & \quad  \rm{ if~\alpha \neq 1}\\
\frac{h_{11} h_{22}- h_{12} h_{21}}{8}     & \quad  \rm{ if~\alpha = 1}.
\end{cases}
   \end{align}

\subsubsection{PAM constellation and signal alignment}

In this setting, all the elements of the codewords are designed to take the following forms (without time index) for transmitter~1 and transmitter~2, respectively,
 \begin{align}
   v_{1}  &=   v_{1, c} +   \sqrt{P^{ -  \beta_{1, p} }}  v_{1, p}   \label{eq:xv11}  \\
   v_{2}  &=   \sqrt{P^{ -  \beta_{2,c}}} \frac{h_{21}}{h_{22}} v_{2,c}   +   \sqrt{P^{ -  \beta_{2, m}}}    \frac{h_{21}}{h_{22}}v_{2, m}.   \label{eq:xv22}  
 \end{align}
 Then, we can rewrite   the  channel input in~\eqref{eq:macwt2x189283}  as 
\begin{align}
 x_1  = &   \varepsilon v_{1, c} +  \varepsilon \sqrt{P^{ -  \beta_{1, p} }}  v_{1, p}   + \varepsilon \sum_{{\ell}=1}^{\tau/2}  \delta_{1, {\ell}}  \sqrt{P^{ -  \beta_{1, 1, {\ell}} }}   \cdot  u_1    +  \varepsilon \sum_{{\ell}=1}^{\tau/2}  \delta_{1, {\ell}}    \sqrt{P^{-  \beta_{2, 1, {\ell}} }}   \cdot   u_2  \label{eq:macwt2x1} \\
 x_2  =&   \varepsilon  \sqrt{P^{ -  \beta_{2,c}}} \frac{h_{21}}{h_{22}} v_{2,c}   +  \varepsilon \sqrt{P^{ -  \beta_{2, m}}}    \frac{h_{21}}{h_{22}}v_{2, m}  + \varepsilon \sum_{{\ell}=1}^{\tau/2}  \delta_{2, {\ell}}  \sqrt{P^{ -  \beta_{1, 2, {\ell}} }}   \cdot  u_1   +  \varepsilon \sum_{{\ell}=1}^{\tau/2}  \delta_{2, {\ell}}    \sqrt{P^{-  \beta_{2, 2, {\ell}} }}   \cdot   u_2 \label{eq:macwt2x2}   
\end{align}
 where the random variables $\{v_{k,c}, v_{1, p}, v_{2, m}, u_k\}_{k=1, 2}$ are   \emph{independent} and \emph{uniformly}  drawn from the PAM constellation sets
 \begin{align}
   v_{1,c}  &  \in    \Omega ( \xi =  \frac{  \eta_{1,c} \gamma}{Q} ,    \   Q = P^{ \frac{  \lambda_{1,c} }{2}})  \label{eq:constellationGsymmac393}   \\    
v_{1, p}   &  \in    \Omega ( \xi =  \frac{   \gamma}{2Q} ,   \   Q =  P^{ \frac{ \lambda_{1, p} }{2}} )  \label{eq:constellationGsymmac39204}   \\ 
   v_{2, c}      &  \in    \Omega ( \xi =  \frac{ \eta_{2,c}  \gamma}{Q} ,   \   Q = P^{ \frac{  \lambda_{2,c} }{2}} )    \label{eq:constellationGsymmac28301}    \\
 u_1    &  \in    \Omega ( \xi =  \frac{\gamma}{Q} ,    \   Q = P^{ \frac{ \max\{ \lambda_{1,c}, \lambda_{2,c}\} }{2}})  \label{eq:constellationGsymmac83374}  \\
  v_{2, m},  u_2    &  \in    \Omega ( \xi =  \frac{   \gamma}{2Q} ,   \   Q =  P^{ \frac{ \lambda_{2, m} }{2}} )  \label{eq:constellationGsymac320}   \
 \end{align}
 respectively, where  $\gamma  \in \bigl(0, \frac{1}{\tau \cdot 2^{\tau}}]$. 
The  parameters $\{ \lambda_{k,c},  \lambda_{1,   p},  \lambda_{2, m}, \beta_{2, c}, \beta_{1,  p},  \beta_{2, m},  \beta_{1, k, {\ell}}, \beta_{2, k, {\ell}}\}_{k, \ell}$ are designed as  in  Table~\ref{tab:wiremacpara}.  Note that $B$ is a  parameter within a specific range, which will be specified later on for different cases of $\alpha$. 
The parameters $\eta_{1,c}$ and $\eta_{2,c}$ are designed to ensure that $v_{1, c}$ and $v_{2, c}$ have a certain integer relationship on  the minimum distances of their PAM constellation sets\footnote{For a PAM constellation set defined as $\Omega (\xi,  Q)  \defeq   \{ \xi  a :   \    a \in   [-Q,   Q] \cap \Zc \}$, the minimum distance of the constellation is $\xi$.}.  Specifically, $\eta_{1,c}$ and $\eta_{2,c}$ are  designed as
\begin{align}
\eta_{i_m,c}&=1,  \quad \quad \quad \eta_{i_n,c}= \Big \lceil \sqrt{P^{ \lambda_{i_m,c} -  \lambda_{i_n,c}}} \Big \rceil \Big/  \sqrt{P^{ \lambda_{i_m,c}-  \lambda_{i_n,c}}}      \label{eq:para001}  \\
\text{where} \quad   i_m& =  \argmax_{i\in\{1, 2\}}   \lambda_{i,c},  \quad \quad i_n \neq i_m ,    \quad  i_m,  i_n \in\{1,2\}.    \label{eq:para010} 
\end{align}
With this design, 
the ratio of the minimum distance of the constellation for $v_{i_n, c}$ and the minimum distance of the  constellation for $v_{i_m, c}$,  i.e.,  $\frac{\eta_{i_n,c} \gamma}{P^{ \frac{  \lambda_{i_n,c} }{2}}} \big/  \frac{\eta_{i_m,c} \gamma}{P^{ \frac{  \lambda_{i_m,c} }{2}}}$, is an integer, where  $1 \leq \eta_{i_n,c}   < 2$. 
By following the steps in \eqref{eq:power35781}-\eqref{eq:power3578}, it is easy to check  that  the average power constraints $\E |x_1|^2 \leq 1$ and $\E |x_2|^2 \leq 1$ are satisfied.
In our scheme, the parameter $\tau$ is designed as 
 \begin{align}\label{eq:taudef3245}
\tau \defeq \begin{cases}  
2\Big \lceil  {\max \big \{ \big \lceil \frac{\alpha}{1- \alpha} \big \rceil,  \big \lceil \frac{1- \alpha +B}{1- \alpha} \big \rceil \big \}/2}   \Big \rceil        & \quad  \rm{ if~\alpha \neq 1}\\
2     & \quad  \rm{ if~\alpha = 1}.
\end{cases}
   \end{align}
The parameters  $\{ \delta_{k,  {\ell}}\}_{k, \ell}$ are designed as 
\begin{align}\label{eq:delta8344}
\delta_{1,  {\ell}}  \defeq \begin{cases}  
-  \big(\frac{h_{12} h_{21}}{h_{11} h_{22}}\big)^{\ell}       & \quad  \rm{ if~\alpha \neq 1}\\
  - \frac{h_{12}h_{21}}{h_{11} h_{22}- h_{12} h_{21}}      & \quad  \rm{ if~\alpha = 1}
\end{cases}
   \end{align}
and 
\begin{align}\label{eq:delta8344b}
\delta_{2,  {\ell}}  \defeq \begin{cases}  
 \frac{h_{21}}{h_{22}} \cdot   \big(\frac{h_{12} h_{21}}{h_{11} h_{22}}\big)^{{\ell}-1}   & \quad  \rm{ if~\alpha \neq 1}\\
  \frac{h_{11}h_{21}}{h_{11} h_{22}- h_{12} h_{21}}    & \quad  \rm{ if~\alpha = 1}
\end{cases}
   \end{align}
for $\ell \in \{1, 2, \cdots \tau/2\}$.

\begin{table}
\caption{Designed  parameters for MAC-WT channel when $\alpha \leq 1$. The last two rows correspond to the design of the parameters $\beta_{1, k, \ell}$ and $\beta_{2, k, \ell}$, for  $k  \in  \{1, 2\}$ and  $\ell  \in \{1,  2,  \cdots,   \frac{\tau}{2}\}$.}
\begin{center}
{\renewcommand{\arraystretch}{2}
\begin{tabular}{|c|c|c|c|c|c|}
  \hline
&\multicolumn{2}{|c|}{$0 \leq \alpha \leq \frac{2}{3}$} & \multicolumn{2}{|c|}{$\frac{2}{3}  \leq \alpha < 1$}  &   $\alpha=1$ \\
  \hline
                     &   $ 0\leq B \leq (2\alpha\!-\!1)^{+}$  &  $ (2\alpha\!-\!1)^{+} <B \leq \alpha$  &   $0 \leq B \leq 2\alpha-1$  &  $  2\alpha \!-\!1 < B \leq  \alpha$  &  $(d^{'}_1, d^{'}_2) \!\in\! {\mathcal {D}}^* (1)$ \\
 \hline
  $\beta_{2, c}$    	&  	$1- \alpha$   &  	$1- \alpha$   &  	$1- \alpha$   &  	$1- \alpha$  & 0 \\
    \hline
  $\beta_{2, m}$    		\!		  &   $\infty$     &   $\alpha -B   $ &   $\infty$     &   $\alpha -B   $   & $\infty$\\
    \hline
  $\beta_{1, p}$    				  &   $\alpha$    &    $\alpha$  &   $\alpha$  &  $\infty$   & $\infty$    \\
\hline
   $\lambda_{2,c}$   	    &     $B- \epsilon$ & $(2\alpha -1)^{+}- \epsilon $  &     $B- \epsilon$ & $2\alpha -1- \epsilon $  &  $d^{'}_2- \epsilon $  \\
    \hline
 $\lambda_{2,m}$   	    &     $0 $  &   $B-(2\alpha -1)^{+}-\epsilon$ &     $0 $  &   $B\!-\!2\alpha \!+\!1\!-\! \epsilon$   &  0 \\
    \hline
   $\lambda_{1,p}$   	      &   $1\!-\!\alpha \!-\!B\!-\! \epsilon $    &   $(\max\{B, 1\!-\!\alpha\} \! \!-\! \!B  \!- \! \epsilon)^{+}  $  &    $\min\{1\!-\!\alpha, 2\alpha\!-\!1\!\!-\!\!B\} \!- \!\epsilon$ & 0  &  0  \\
    \hline
$\lambda_{1, c}$ 			            & $\alpha - \epsilon$      &$(1-B- \lambda_{1,p})^{+} - 2\epsilon$   &  $(1-B- \lambda_{1,p})^{+} - 2\epsilon$  & $1-B-\epsilon$  &  $d^{'}_1- \epsilon $     \\
    \hline
      $\beta_{2, k, {\ell}}$ &  $\infty$  &  $\beta_{1,  k, {\ell}} \!-\! (1+B-2\alpha)$ &  $\infty$  &  $\beta_{1,k, {\ell}} \!-\! (\!1\!+\!B\!-\!2\alpha\!) $ &  $\infty$ \\ 
\hline
$\beta_{1, k, {\ell}}$ &\multicolumn{4}{|c|}{$(2{\ell} -k+1)(1 - \alpha) $} & 0\\
      \hline
    \end{tabular}
}
\end{center}
\label{tab:wiremacpara}
\end{table} 

\subsubsection{Secure rate analysis} 
Given the codebook design and signal mapping, the result of  \cite[Theorem~1]{BMKarxiv:10} implies that we can achieve the following secure rate region
\begin{align}  \label{eq:macsecregion}
\big\{ (R_1, R_2): \sum_{k=1}^{2} R_k     \leq  \big( \Imu(v_1, v_2; y_1) -\Imu(v_1, v_2; y_2)  \big)^{+}, \  R_1   \leq   \Imu(v_1; y_1|v_2), \  R_2   \leq   \Imu(v_2; y_1|v_1) \big\}.
\end{align} 
In the following subsections we will provide the analysis of  the rate region  under three different  cases, i.e.,  $0  \leq \alpha \leq \frac{2}{3}$,  $\frac{2}{3} \leq \alpha < 1$ and $\alpha=1$.   In the proposed scheme, a Markov chain-based  interference neutralization method is used.

\subsection{$ 0 \leq \alpha \leq \frac{2}{3}$ \label{sec:macwiretapschemeaaa1}}

For this case of $ 0 \leq  \alpha \leq \frac{2}{3}$, we will divide the analysis into two cases and show that the secure GDoF region ${\mathcal{D}}^{*} ( \alpha) $ is achievable.

\subsubsection{$ 0 \leq  B \leq (2\alpha-1)^{+}$}
In the case with  $ 0 \leq \alpha \leq \frac{2}{3}$   and $ 0 \leq  B \leq (2\alpha-1)^{+}$, based  on the parameter design in  \eqref{eq:macwt2x1}-\eqref{eq:delta8344b}  and Table~\ref{tab:wiremacpara},
 the transmitted signals take the following forms 
\begin{align}
 x_1  = &    v_{1, c} +   \sqrt{P^{ -  \alpha }}  v_{1, p}  + \sum_{{\ell}=1}^{\tau/2}  \delta_{1, {\ell}}  \sqrt{P^{ -  \beta_{1, 1, {\ell}} }}   \cdot   u_1   \label{eq:macwt2x1casea11}   \\
 x_2  =&     \sqrt{P^{ \alpha- 1}} \frac{h_{21}}{h_{22}} v_{2,c}  +
  \sum_{{\ell}=1}^{\tau/2}  \delta_{2, {\ell}}  \sqrt{P^{ -  \beta_{1, 2, {\ell}} }}   \cdot  u_1.  \label{eq:macwt2x2casea12}
\end{align}
Then the received signals take the  forms as
\begin{align}
y_{1} 
&  \! =  \!      \sqrt{P}  h_{11} v_{1, c}  \!  + \!      \sqrt{P^{2\alpha  \!-  \! 1 }}  \frac{ h_{12}h_{21}}{h_{22}}  v_{2, c}  \!   +  \!   \sqrt{P^{ 1  \!-  \!  \alpha }} h_{11}  v_{1, p} 
 \! + \!    \sum_{{\ell}=1}^{\tau/2}   \underbrace{(\sqrt{P^{1  \!- \!  \beta_{1, 1,  {\ell}}}}   \delta_{1, {\ell}}   h_{11} \!  +  \!  \sqrt{P^{ \alpha  \!-  \! \beta_{1, 2,  {\ell} } }} \delta_{2, {\ell}}     h_{12})}_{\rm{interference~neutralization}}   u_1  \! +  \! z_{1}   \non \\
&=   \sqrt{P}  h_{11} v_{1, c}  +     \sqrt{P^{2\alpha-1 }   } \frac{ h_{12}h_{21}}{h_{22}}  v_{2, c}    +   \sqrt{P^{ 1-   \alpha }} h_{11}  v_{1, p}  +  z_{1} ,  \label{eq:macwiretap2y1casea11}  \\
y_{2} 
&=  \underbrace{   \sqrt{P^{ \alpha }}  h_{21}     v_{1, c}  +     \sqrt{P^{\alpha }}   h_{21}     v_{2, c} +    \sqrt{P^{1-\beta_{1, 2,  1} }   }   h_{22}   \delta_{2, 1}   u_1}_{\rm{aligned}}  \non\\
&\quad +   \sum_{{\ell}=1}^{\frac{\tau}{2}-1} \underbrace{ (\sqrt{P^{ \alpha-  \beta_{1, 1,  {\ell}}}}   \delta_{1,  {\ell}}   h_{21}  +  \sqrt{P^{1-  \beta_{1, 2, {\ell} +1}}} \delta_{2, {\ell}+1}     h_{22}   )}_{\rm{interference~neutralization}}  u_1 + \sqrt{P^{ \alpha-  \beta_{1, 1, \tau/2}}}   \delta_{1, \tau/2}   h_{21} u_1+    h_{12}  v_{1, p} + z_{2}    \non\\
&=    \sqrt{P^{ \alpha}}  h_{21}  (   v_{1, c}  +  v_{2, c} + u_1)+   \sqrt{P^{ \alpha-  \beta_{1, 1, \tau/2}}}   \delta_{1, \tau/2}   h_{21} u_1+    h_{12}  v_{1, p} + z_{2}   . \label{eq:macwiretap2y1casea12}
\end{align}
 In the above expressions of $y_1$ and $y_2$,  the interference is removed by using the  Markov chain-based interference neutralization method.
For the secure rate region in \eqref{eq:macsecregion},  we will prove that $\Imu(v_1, v_2; y_1)  -  \Imu(v_1, v_2; y_2) \geq  \frac{1-3\epsilon}{2}  \log P + o( \log P )$, $ \Imu(v_1; y_1|v_2)  \geq \frac{1-B-2\epsilon }{2}  \log P + o( \log P )$  and $ \Imu(v_2; y_1|v_1)    \geq    \frac{B- \epsilon}{2} \log P + o(\log P)$, which will imply that the   GDoF region  $\{(d_1, d_2) | d_1 + d_2 \leq 1, 0 \leq d_1 \leq   1-B, 0 \leq d_2 \leq B\}$ is achievable, for almost all the channel coefficients  $\{h_{k\ell}\} \in (1, 2]^{2\times 2}$.

First we focus on the lower bound of $\Imu(v_1, v_2; y_1)  -  \Imu(v_1, v_2; y_2)$. Let $\hat v_{1, c}, \hat v_{2, c}$ and $ \hat v_{1,p}$ be the estimates of  $v_{1, c}, v_{2, c}$ and $ v_{1,p}$ from $y_1$, respectively,
Let $ \text{Pr} [  \{ v_{1,c} \neq \hat{v}_{1,c} \} \cup  \{ v_{1,p} \neq \hat{v}_{1,p} \} \cup  \{ v_{2,c} \neq \hat{v}_{2,c} \}  ] $ denote the corresponding  error probability of this estimation.
Then the term $\Imu(v_1, v_2; y_1)$ can be lower bounded by
 \begin{align}
\Imu(v_1, v_2; y_1)   &\geq \Imu(v_1, v_2;  \hat v_{1,c}, \hat v_{1,p}, \hat v_{2,c}) \label{eq:macrate843952}\\
&= \Hen(v_{1}, v_2) -  \Hen(v_{1}, v_2|\hat v_{1,c}, \hat v_{1,p}, \hat v_{2,c})  \non \\
& \geq \Hen(v_{1}, v_2)  - (1+\text{Pr} [  \{ v_{1,c} \neq \hat{v}_{1,c} \} \cup  \{ v_{1,p} \neq \hat{v}_{1,p} \}\cup  \{ v_{2,c} \neq \hat{v}_{2,c} \} ]) \cdot \Hen(v_{1}, v_2) \label{eq:macrate883904}\\
&\geq  \bigl( 1 -   \text{Pr} [  \{ v_{1,c} \neq \hat{v}_{1,c} \} \cup  \{ v_{1,p} \neq \hat{v}_{1,p} \}\cup  \{ v_{2,c} \neq \hat{v}_{2,c} \} ]  \bigr)   \cdot \Hen(v_{1}, v_2)  - 1  \label{eq:macrate8931548}     
 \end{align}
where \eqref{eq:macrate843952} results from the Markov chain $\{v_1, v_2\} \to y_1 \to  \{\hat{v}_{1,c}, \hat{v}_{1,p}, \hat{v}_{2,c} \}$; \eqref{eq:macrate883904} uses Fano's inequality. For the term  $\Hen(v_{1}, v_2)$ appeared in \eqref{eq:macrate8931548} we have 
\begin{align}
\Hen(v_{1}, v_2) = \Hen(v_{1,c}) +\Hen(v_{1,p}) +\Hen(v_{2,c}) =  \frac{1  -3\epsilon}{2} \log P + o(\log P).  \label{eq: rateomactet41}
\end{align}
Below we provide a result  on the error probability appeared in \eqref{eq:macrate8931548}.
 \begin{lemma}  \label{lm:mac1rateerror49293}
When  $ 0 \leq \alpha \leq \frac{2}{3}$   and $ 0\leq B  \leq (2\alpha-1)^{+}$, given the signal design in Table~\ref{tab:wiremacpara}, \eqref{eq:constellationGsymmac393}-\eqref{eq:constellationGsymac320}  and \eqref{eq:macwt2x1casea11}-\eqref{eq:macwt2x2casea12},  for almost all the channel realizations the error probability of decoding  $\{v_{1,c}, v_{1,p}, v_{2,c} \}$ from $y_1$  is vanishing when $P$ goes large, i.e.,
 \begin{align}
 \text{Pr} [  \{ v_{1,c} \neq \hat{v}_{1,c} \} \cup  \{ v_{1,p} \neq \hat{v}_{1,p} \}\cup  \{ v_{2,c} \neq \hat{v}_{2,c} \}]  \to 0         \quad \text {as}\quad  P\to \infty .   \label{eq:macrateerror85852343}
 \end{align}
 \end{lemma}
\begin{proof}
See Appendix~\ref{sec:errorcasemac323}.
\end{proof}
With the results of \eqref{eq:macrate8931548}, \eqref{eq: rateomactet41} and Lemma~\ref{lm:mac1rateerror49293}, we can bound the term $ \Imu(v_1, v_2; y_1)   $  as 
\begin{align}
 \Imu(v_1, v_2; y_1)    \geq    \frac{1  - 3\epsilon}{2} \log P + o(\log P)   \label{eq:macrateana39895}
\end{align}
for almost all the channel coefficients  $\{h_{k\ell}\} \in (1, 2]^{2\times 2}$. 
For the term $\Imu(v_1, v_2; y_2) $, we can bound it as
\begin{align}
&\Imu(v_1, v_2; y_2) \non\\
& \leq  \Imu(v_1, v_2; y_2, v_{1, c}  +  v_{2, c} + u_1)  \label{eq:macratebegin89023} \\
&=  \Imu(v_1, v_2; v_{1, c}  +  v_{2, c} + u_1) +  \Imu(v_1, v_2; y_2 | v_{1, c}  +  v_{2, c} + u_1) \non\\
& \leq    \Hen(v_{1, c}  +  v_{2, c} + u_1) - \Hen(u_1) +  \hen(  \sqrt{P^{ \alpha-  \beta_{1, 1, \tau/2}}}   \delta_{1, \tau/2}   h_{21} u_1+    h_{12}  v_{1, p} + z_{2}  ) -  \hen(z_{2}) \label{eq:macrateana3i4o}\\
& \leq \underbrace{\log(6Q'+1) - \log(2Q'+1)}_{\leq \log 3} + \frac{1}{2} \log \big(2 \pi e (\frac{8}{3\tau^2}+\frac{2}{3\tau^2\cdot 4^{\tau}}+1) \big) - \frac{1}{2} \log(2 \pi e)  \label{eq:macrateana8981212}\\
& \leq \log \Bigg (3 \sqrt{\frac{8}{3\tau^2}+\frac{2}{3\tau^2\cdot 4^{\tau}}+1}  \Bigg) \label{eq:macrateana8549283}
\end{align}
 where \eqref{eq:macrateana3i4o} follows from the fact that $\{v_{1, c}, v_{2, c}, u_1 \}$ are mutually independent; \eqref{eq:macrateana8981212} stems from the derivation that $  \Hen(v_{1, c}  +  v_{2, c} + u_1) \leq \log(6Q'+1)$ and $\Hen(u_1) =\log(2Q'+1)$, where $Q' \defeq P^{ \frac{ \max\{ \lambda_{1,c}, \lambda_{2,c}\} }{2}}$. Due to our design in \eqref{eq:para001}-\eqref{eq:para010},  the ratio between the minimum distance  of the constellation for  $v_{2, c}$ and the minimum distance  of  the constellation for  $v_{1, c}$  is an integer.
This integer relationship allows us to minimize the value of $\Hen(v_{1, c}  +  v_{2, c} + u_1) $, which can be treated as a GDoF penalty.

Given the results of \eqref{eq:macrateana39895} and \eqref{eq:macrateana8549283}, it reveals that
\begin{align}
 \Imu(v_1, v_2; y_1) -\Imu(v_1, v_2; y_2)   \geq  \frac{1-3\epsilon}{2}  \log P + o( \log P )  \label{eq:macrate1939834}
\end{align}
for almost all the channel realizations. Now we consider the bound of $\Imu(v_1; y_1|v_2) $. Let 
\begin{align}
y^{'}_1 =   \sqrt{P}  h_{11} v_{1, c}  +   \sqrt{P^{ 1-   \alpha }} h_{11}  v_{1, p}  +  z_{1} \label{eq:ratey1new824}
\end{align}
 and  let $\{\hat{v}^{'}_{1, c}, \hat{v}^{'}_{1, p}\}$ be the estimates of $\{v_{1, c}, v_{1, p}\}$ from $y^{'}_1$. Then  we have
 \begin{align}
  \Imu(v_1; y_1|v_2)  & = \Imu(v_1, y^{'}_1) \label{eq:ratebegin8390}\\  
        & \geq      \bigl( 1 -   \text{Pr} [  \{ v_{1, c} \neq \hat{v}^{'}_{1, c} \} \cup  \{ v_{1, p} \neq \hat{v}^{'}_{1, p} \}  ] \bigr)   \cdot \Hen(v_1)  - 1  \label{eq:rat4903420}     
\end{align}
where  \eqref{eq:ratebegin8390} follows from the independence between $v_2$ and $v_1$;
\eqref{eq:rat4903420} follows from the steps in \eqref{eq:macrate843952}-\eqref{eq:macrate8931548}.
 For the term  $\Hen(v_1)$ appeared in \eqref{eq:rat4903420}, we have 
\begin{align}
 \Hen(v_1) = \frac{1  - B-2 \epsilon}{2} \log P + o(\log P). \label{eq: rateomacr114i4}
\end{align}
By following the proof steps of Lemma~\ref{lm:errorcase1}, one can easily prove that error probability of estimating $ v_{1, c} $ and $ v_{1, p} $ based on $y'_1$ is vanishing when $P$ goes large, that is,
\begin{align}
 \text{Pr} [  \{ v_{1,c} \neq \hat{v}^{'}_{1,c} \} \cup  \{ v_{1,p} \neq \hat{v}^{'}_{1,p} \}]  \to 0         \quad \text {as}\quad  P\to \infty .   \label{eq:macratefy1}
\end{align}
With \eqref{eq:rat4903420}, \eqref{eq: rateomacr114i4} and \eqref{eq:macratefy1}, it suggests that
\begin{align}
\Imu(v_1; y_1|v_2)   \geq    \frac{1  - B- 2\epsilon}{2} \log P + o(\log P).   \label{eq:macrateanar17839}
\end{align}
Similarly,  $ \Imu(v_2; y_1|v_1) $ can be bounded by 
\begin{align}
 \Imu(v_2; y_1|v_1)    \geq    \frac{B- \epsilon}{2} \log P + o(\log P).   \label{eq:macrateanar2428396}
\end{align}
By combining the results of  \eqref{eq:macsecregion}, \eqref{eq:macrate1939834}, \eqref{eq:macrateanar17839} and  \eqref{eq:macrateanar2428396}, it implies that  the GDoF region  $\{(d_1, d_2) | d_1 + d_2 \leq 1, 0 \leq d_1 \leq   1-B, 0 \leq d_2 \leq B\}$  is achievable for almost all the channel coefficients, for this case  with $ 0 \leq \alpha \leq \frac{2}{3}$   and $ 0 \leq B \leq (2\alpha-1)^{+}$.

\subsubsection{$(2\alpha-1)^{+} <B \leq \alpha$}

In the case with $ 0 \leq \alpha \leq \frac{2}{3}$   and $(2\alpha-1)^{+} <B \leq \alpha$, the transmitted signals now take the following forms
\begin{align}
 x_1  = &    v_{1, c} +   \sqrt{P^{ -  \alpha }}  v_{1, p}  + \sum_{{\ell}=1}^{\tau/2}  \delta_{1, {\ell}}  \sqrt{P^{ -  \beta_{1, 1, {\ell}} }}   \cdot  \big( u_1  + \sqrt{P^{1+B - 2\alpha }}   \cdot   u_2      \big)  \label{eq:macwt2x1casea21} \\
 x_2  =&     \sqrt{P^{ \alpha-  1}} \frac{h_{21}}{h_{22}} v_{2,c}   +   \sqrt{P^{ B-\alpha}}    \frac{h_{21}}{h_{22}}v_{2, m} +
  \sum_{{\ell}=1}^{\tau/2}  \delta_{2, {\ell}}  \sqrt{P^{ -  \beta_{1, 2, {\ell}} }}   \cdot  \big( u_1  + \sqrt{P^{1+B - 2\alpha }}   \cdot   u_2      \big)   \label{eq:macwt2x2casea22}   
\end{align}
and the  received signals take the following forms
\begin{align}
y_{1} 
&=   \sqrt{P}  h_{11} v_{1, c}  +     \sqrt{P^{2\alpha-1 }   } \frac{ h_{12}h_{21}}{h_{22}}  v_{2, c}    +     \sqrt{P^{B}} \frac{ h_{12}h_{21}}{h_{22}}  v_{2, m}  +   \sqrt{P^{ 1-   \alpha }} h_{11}  v_{1, p}  +  z_{1}  \label{eq:macwiretap2y1casea21} \\
y_{2} 
&=    \sqrt{P^{ \alpha}}  h_{21}  (   v_{1, c}  +  v_{2, c} + u_1)  +     \sqrt{P^{ B+1-\alpha }}   h_{21}  (  v_{2, m} + u_2)  \non\\
&\quad + \sqrt{P^{ \alpha-  \beta_{1, 1, \tau/2}}}   \delta_{1, \tau/2}   h_{21} u_1+   \sqrt{P^{ B+1-\alpha-  \beta_{1, 1,  \tau/2}}}   \delta_{1,\tau/2}   h_{21} u_2  +   h_{12}  v_{1, p} + z_{2} .   \label{eq:macwiretap2y2casea22}
\end{align}

By following the derivations in \eqref{eq:macrate843952}-\eqref{eq:macrate8931548},  the term $\Imu(v_1, v_2; y_1)$ can be lower bounded by
 \begin{align}
  \Imu(v_1, v_2; y_1)   & \geq  \bigl( 1 \! - \!   \text{Pr} [  \{ v_{1,c} \neq \hat{v}_{1,c} \} \cup  \{ v_{1,p} \neq \hat{v}_{1,p} \}\cup  \{ v_{2,c} \neq \hat{v}_{2,c} \} \cup  \{ v_{2,m} \neq \hat{v}_{2,m} \}  ]  \bigr)   \!\cdot\! \Hen(v_{1}, v_2)  \!-\!   1  \label{eq:macrate89318}    \\
& = \frac{1  -4\epsilon}{2} \log P + o(\log P).   \label{eq:macrateana395}  
 \end{align}
where the step in \eqref{eq:macrate89318} derives from \eqref{eq:macrate843952}-\eqref{eq:macrate8931548};
and the last step follows from  Lemma~\ref{lm:mac1rateerror8391028} (see below) and the fact that  $\Hen(v_{1}, v_2) = \Hen(v_{1,c}) +\Hen(v_{1,p}) +\Hen(v_{2,c})+\Hen(v_{2,m})  =  \frac{1  - 4\epsilon}{2} \log P + o(\log P)$.
 \begin{lemma}  \label{lm:mac1rateerror8391028}
When $ 0 \leq \alpha \leq \frac{2}{3}$   and $  (2\alpha-1)^{+} <B \leq \alpha$, given the signal design in Table~\ref{tab:wiremacpara}, \eqref{eq:constellationGsymmac393}-\eqref{eq:constellationGsymac320}  and \eqref{eq:macwt2x1casea21}-\eqref{eq:macwt2x2casea22},  the error probability of estimating  $\{v_{1,c}, v_{1,p}, v_{2,c},  v_{2,m} \}$ from $y_1$  is
 \begin{align}
 \text{Pr} [  \{ v_{1,c} \neq \hat{v}_{1,c} \} \cup  \{ v_{1,p} \neq \hat{v}_{1,p} \}\cup  \{ v_{2,c} \neq \hat{v}_{2,c} \} \cup  \{ v_{2,m} \neq \hat{v}_{2,m} \}  ]  \to 0         \quad \text {as}\quad  P\to \infty .   \label{eq:macrateerror85203}
 \end{align}
 \end{lemma}
\begin{proof}
See Appendix~\ref{sec:errorcasemacr2}.
\end{proof}
By following the steps in \eqref{eq:macratebegin89023}-\eqref{eq:macrateana8549283}, the term $\Imu(v_1, v_2; y_2) $ can be  bounded by
\begin{align}
\Imu(v_1, v_2; y_2)  \leq o(\log P) . \label{eq:macrateana8293}
\end{align}
The bounds in \eqref{eq:macrateana395} and \eqref{eq:macrateana8293} then reveal that  
\begin{align}
 \Imu(v_1, v_2; y_1) -\Imu(v_1, v_2; y_2)   \geq  \frac{1-4\epsilon}{2}  \log P + o( \log P ) . \label{eq:macrate839208}
\end{align}
By following the steps in  \eqref{eq:ratey1new824}-\eqref{eq:macrateanar2428396}, we also have 
\begin{align}
\Imu(v_1; y_1|v_2)  & \geq    \frac{1  - B-2 \epsilon}{2} \log P + o(\log P)   \label{eq:macrateanar1839} \\
 \Imu(v_2; y_1|v_1)   & \geq    \frac{B- 2\epsilon}{2} \log P + o(\log P).   \label{eq:macrateanar2466}
\end{align}
Given the results of  \eqref{eq:macsecregion}, \eqref{eq:macrate839208}, \eqref{eq:macrateanar1839} and  \eqref{eq:macrateanar2466} it  implies that  the GDoF region   $\{(d_1, d_2) | d_1 + d_2 \leq 1, 0 \leq d_1 \leq   1-B, 0 \leq d_2 \leq B\}$  is achievable for this case  with $ 0 \leq  \alpha \leq \frac{2}{3}$   and $  (2\alpha-1)^{+} <B \leq \alpha$.

Finally, by combining the results of the above two cases and by moving $B$ from 0 to $\alpha$, it reveals that   for almost all the channel realizations the proposed scheme achieves  ${\mathcal{D}}^{*} ( \alpha) $ in this case of $ 0 \leq  \alpha \leq \frac{2}{3}$.

\subsection{$ \frac{2}{3}  \leq \alpha < 1$ \label{sec:macwiretapschemebbb1}}

When $ \frac{2}{3}  \leq \alpha < 1$, we will also divide the analysis into two cases.

\subsubsection{$ 0 \leq  B \leq 2\alpha-1$}

In the case with $ \frac{2}{3} \leq \alpha <1$    and $ B  \leq 2\alpha-1$,  the  signals of the transmitters have the same forms as in  \eqref{eq:macwt2x1casea11} and \eqref{eq:macwt2x2casea12}, and the signals of the receivers take the same forms as in  \eqref{eq:macwiretap2y1casea11} and \eqref{eq:macwiretap2y1casea12}.
In this case,  we have 
 \begin{align}
  \Imu(v_1, v_2; y_1)   & \geq  \bigl( 1 -   \text{Pr} [  \{ v_{1,c} \neq \hat{v}_{1,c} \} \cup  \{ v_{1,p} \neq \hat{v}_{1,p} \}\cup  \{ v_{2,c} \neq \hat{v}_{2,c} \}  ]  \bigr)   \cdot \Hen(v_{1}, v_2)  - 1  \label{eq:macrate2902343}   \\
    &  = \frac{1  - 3\epsilon}{2} \log P + o(\log P)  \label{eq:macrateana8987}
 \end{align}
for  almost all the channel coefficients, where \eqref{eq:macrate2902343} follows from the steps in \eqref{eq:macrate843952}-\eqref{eq:macrate8931548};
 the last step stems from Lemma~\ref{lm:mac1rateerror29054} (see below) and the derivation that $\Hen(v_{1}, v_2) = \Hen(v_{1,c}) +\Hen(v_{1,p}) +\Hen(v_{2,c})=  \frac{1  - 3\epsilon}{2} \log P + o(\log P)$.
 \begin{lemma}  \label{lm:mac1rateerror29054}
When $ \frac{2}{3} \leq \alpha < 1$    and $0 \leq B  \leq 2\alpha-1$, given the signal design in Table~\ref{tab:wiremacpara}, \eqref{eq:constellationGsymmac393}-\eqref{eq:constellationGsymac320}  and  \eqref{eq:macwt2x1casea11}-\eqref{eq:macwt2x2casea12},  for almost all the channel coefficients  $\{h_{k\ell}\} \in (1, 2]^{2\times 2}$, the error probability of decoding   $\{v_{1,c},  v_{1,p}, v_{2,c}\}$ from $y_1$  is vanishing when $P$ is large, i.e.,
 \begin{align}
 \text{Pr} [  \{ v_{1,c} \neq \hat{v}_{1,c} \}  \cup  \{ v_{1,p} \neq \hat{v}_{1,p} \}  \cup  \{ v_{2,c} \neq \hat{v}_{2,c} \}  ]  \to 0         \quad \text {as}\quad  P\to \infty .   \label{eq:macrateerror34789}
 \end{align}
 \end{lemma}
\begin{proof}
See Appendix~\ref{sec:errorcasemacr8325}.
\end{proof}
From the steps in \eqref{eq:macratebegin89023}-\eqref{eq:macrateana8549283}, one can easily show that 
\begin{align}
\Imu(v_1, v_2; y_2) \leq o(\log P).  \label{eq:macrateana233412}
\end{align}
With  \eqref{eq:macrateana8987} and \eqref{eq:macrateana233412} we have the following inequality  
\begin{align}
 \Imu(v_1, v_2; y_1) -\Imu(v_1, v_2; y_2)   \geq  \frac{1-3\epsilon}{2}  \log P + o( \log P )  \label{eq:macrate6662}
\end{align}
for almost all the channel coefficients.
From the steps in  \eqref{eq:ratey1new824}-\eqref{eq:macrateanar2428396}, the following two inequalities can be easily derived
\begin{align}
\Imu(v_1; y_1|v_2)   &\geq    \frac{1  - B- 2\epsilon}{2} \log P + o(\log P),   \label{eq:macrateanar23933}  \\
 \Imu(v_2; y_1|v_1)   & \geq    \frac{B- \epsilon}{2} \log P + o(\log P).   \label{eq:macrateanar20032}
\end{align}
From \eqref{eq:macsecregion}, \eqref{eq:macrate6662}, \eqref{eq:macrateanar23933} and  \eqref{eq:macrateanar20032} we can conclude that  the secure GDoF region $\{(d_1, d_2) | d_1 + d_2 \leq 1, 0 \leq d_1 \leq   1-B, 0 \leq d_2 \leq B\}$ is achievable for almost all the channel coefficients  for this case  with $ 0  \leq \alpha \leq \frac{2}{3}$   and $ 0\leq B \leq 2\alpha-1$.

\subsubsection{$2\alpha-1 < B \leq \alpha$}

In the case with $  \frac{2}{3} \leq \alpha < 1$   and $ 2\alpha-1 <B \leq \alpha$, 
based  on the parameter design in  \eqref{eq:macwt2x1}-\eqref{eq:delta8344b}  and Table~\ref{tab:wiremacpara}, the transmitted signals are expressed as 
\begin{align}
 x_1  = &    v_{1, c}+ \sum_{{\ell}=1}^{\tau/2}  \delta_{1, {\ell}}  \sqrt{P^{ -  \beta_{1, 1, {\ell}} }}   \cdot  \big( u_1  + \sqrt{P^{1+B - 2\alpha }}   \cdot   u_2      \big)  \label{eq:macwt2x1caseb21} \\
 x_2  =&     \sqrt{P^{ \alpha-  1}} \frac{h_{21}}{h_{22}} v_{2,c}   +   \sqrt{P^{ B-\alpha}}    \frac{h_{21}}{h_{22}}v_{2, m} +
   \sum_{{\ell}=1}^{\tau/2}  \delta_{2, {\ell}}  \sqrt{P^{ -  \beta_{1, 2, {\ell}} }}   \cdot  \big( u_1  + \sqrt{P^{1+B - 2\alpha }}   \cdot   u_2      \big)   \label{eq:macwt2x2caseb22}   
\end{align}
and the received signals have the following forms
\begin{align}
y_{1} &=   \sqrt{P}  h_{11} v_{1, c}  +     \sqrt{P^{2\alpha-1 }   } \frac{ h_{12}h_{21}}{h_{22}}  v_{2, c}    +     \sqrt{P^{B}} \frac{ h_{12}h_{21}}{h_{22}}  v_{2, m} +  z_{1}  \label{eq:macwiretap2y1caseb21} \\
y_{2} 
&=    \sqrt{P^{ \alpha}}  h_{21}  (   v_{1, c}  +  v_{2, c} + u_1)  +     \sqrt{P^{B+1-\alpha}}   h_{21}  (  v_{2, m} + u_2)  \non\\
&\quad + \sqrt{P^{ \alpha-  \beta_{1, 1, \tau/2}}}   \delta_{1, \tau/2}   h_{21} u_1+   \sqrt{P^{ B+1-\alpha-  \beta_{1, 1,  \tau/2}}}   \delta_{1,\tau/2}   h_{21} u_2 + z_{2} . \label{eq:macwiretap2y2caseb22}
\end{align}
In this case, we have
 \begin{align}
   \Imu(v_1, v_2; y_1)    &\geq  \bigl( 1 -   \text{Pr} [  \{ v_{1,c} \neq \hat{v}_{1,c} \} \cup  \{ v_{2,c} \neq \hat{v}_{2,c} \} \cup  \{ v_{2,m} \neq \hat{v}_{2,m} \}  ]  \bigr)   \cdot \Hen(v_{1}, v_2)  - 1   \label{eq:macrate439030}     \\
   & =    \frac{1  - 3\epsilon}{2} \log P + o(\log P)   \label{eq:macrateana93234}
 \end{align}
where \eqref{eq:macrate439030}  stems from the steps in \eqref{eq:macrate843952}-\eqref{eq:macrate8931548};
\eqref{eq:macrateana93234} uses the facts that $\Hen(v_{1}, v_2) = \Hen(v_{1,c}) +\Hen(v_{2,c})+\Hen(v_{2,m})  =  \frac{1  - 3\epsilon}{2} \log P + o(\log P)$  and that $ \text{Pr} [  \{ v_{1,c} \neq \hat{v}_{1,c} \} \cup  \{ v_{2,c} \neq \hat{v}_{2,c} \} \cup  \{ v_{2,m} \neq \hat{v}_{2,m} \}  ]  \to 0$ as  $P\to \infty$ by following  the proof steps of Lemma~\ref{lm:mac1rateerror8391028}, using a successive decoding method.
From the steps in \eqref{eq:macratebegin89023}-\eqref{eq:macrateana8549283}, we have
\begin{align}
\Imu(v_1, v_2; y_2) \leq o(\log P),  \label{eq:macrateana23341200}
\end{align}
which, together with the result of \eqref{eq:macrateana93234}, reveals that 
\begin{align}
\Imu(v_1, v_2; y_1) -\Imu(v_1, v_2; y_2)   \geq  \frac{1-3\epsilon}{2}  \log P + o( \log P ).  \label{eq:macrate903245}
\end{align}
From the steps in  \eqref{eq:ratey1new824}-\eqref{eq:macrateanar2428396}, the following two inequalities can be easily derived
\begin{align}
\Imu(v_1; y_1|v_2)   &\geq    \frac{1  - B- \epsilon}{2} \log P + o(\log P),  \label{eq:macrateanar11212}  \\
 \Imu(v_2; y_1|v_1)   & \geq    \frac{B- 2\epsilon}{2} \log P + o(\log P).   \label{eq:macrateanar2466532}
\end{align}
With the results in  \eqref{eq:macsecregion}, \eqref{eq:macrate903245}, \eqref{eq:macrateanar11212} and  \eqref{eq:macrateanar2466532}, we can conclude that  the GDoF region   $\{(d_1, d_2) | d_1 + d_2 \leq 1, 0 \leq d_1 \leq   1-B, 0 \leq d_2 \leq B\}$ is achievable for this case  with $\frac{2}{3} \leq \alpha < 1$   and $  2\alpha-1 <B \leq \alpha$.

Finally, by combining the results of the above two cases and by moving $B$ from 0 to $\alpha$, it reveals that for almost all the channel realizations the proposed scheme achieves  ${\mathcal{D}}^{*} ( \alpha) $ in this case of $ 0 \leq  \alpha \leq \frac{2}{3}$.

\subsection{$\alpha=1$} \label{sec:macwiretapschemeaaa000}

In  the case with  $\alpha =1$, for any GDoF pair $(d^{'}_1, d^{'}_2)$ such that  $(d^{'}_1, d^{'}_2) \in  \mathcal{D}^*( 1)$,   we will provide the following scheme and show that the GDoF pair $(d^{'}_1, d^{'}_2)$   is achievable with $d^{'}_c= \max\{d^{'}_1, d^{'}_2\}$ GDoF common randomness.
Based  on the parameter design in  \eqref{eq:macwt2x1}-\eqref{eq:delta8344b}  and Table~\ref{tab:wiremacpara}, 
then the transmitted signals are designed as
\begin{align}
 x_1  = & \varepsilon  v_{1, c}  - \varepsilon  \frac{h_{12}h_{21}}{h_{11} h_{22}- h_{12} h_{21}} \cdot u_1    \label{eq:MACeqx1}  \\
 x_2  = &\varepsilon  \frac{h_{21}}{h_{22}} \cdot  v_{2, c} + \varepsilon \frac{h_{11}h_{21}}{h_{11} h_{22}- h_{12} h_{21}} \cdot  u_1  .    \label{eq:MACeqx2}
 \end{align}
The PAM constellation sets of the random variables $\{v_{1,c}, v_{2,c},  u_1 \}$ are designed as
 \begin{align}
   v_{1,c}   &  \in    \Omega ( \xi =  \frac{ \eta_{1, c}  \gamma}{Q} ,   \    Q = P^{ \frac{   d^{'}_1-\epsilon }{2}})  \label{eq:constellationalpha1}   \\    
   v_{2, c}      &  \in    \Omega ( \xi =  \frac{\eta_{2, c} \gamma}{Q} ,   \   Q = P^{ \frac{  d^{'}_2-\epsilon }{2}} )    \label{eq:constellationalpha2}   \\
   u_1   &  \in    \Omega ( \xi =  \frac{   \gamma}{Q} ,   \    Q = P^{ \frac{ \max\{  d^{'}_1-\epsilon, \ d^{'}_2-\epsilon\} }{2}}) .  \label{eq:constellationalpha3}   
 \end{align}
In terms of the signals at the receivers, we have 
\begin{align}
y_{1} &=   \varepsilon  \sqrt{P}  h_{11} v_{1, c} +\varepsilon  \sqrt{P}  \frac{h_{12}h_{21}}{h_{22}} \cdot  v_{2, c}  +  z_{1}  \label{eq:MACeqy1}  \\
y_{2} &=  \varepsilon  \sqrt{P}  h_{21}  (   v_{1, c}  + v_{2, c} +  u_1)   +  z_{2}.   \label{eq:MACeqy2} 
\end{align}

We  first focus on the bound of $\Imu(v_1, v_2; y_1)  -  \Imu(v_1, v_2; y_2)$ in \eqref{eq:macsecregion}. 
By following the derivations in \eqref{eq:macrate843952}-\eqref{eq:macrate8931548},  the term $\Imu(v_1, v_2; y_1)$ can be lower bounded by
 \begin{align}
  \Imu(v_1, v_2; y_1)    \geq  \bigl( 1 -   \text{Pr} [  \{ v_{1,c} \neq \hat{v}_{1,c} \} \cup  \{ v_{2,c} \neq \hat{v}_{2,c} \}  ]  \bigr)   \cdot \Hen(v_1, v_2)  - 1  \label{eq:macrateals13993}     
 \end{align}
where $v_k=  v_{k,c}$ for $k\in \{1, 2\}$, and the entropy  $ \Hen(v_1, v_2)$ in \eqref{eq:macrateals13993} is derived as
\begin{align}
\Hen(v_1, v_2) = \Hen(v_{1,c}) +\Hen(v_{2,c})=  \frac{d^{'}_1+ d^{'}_2 - 2\epsilon}{2} \log P + o(\log P). \label{eq: rateomacal214}
\end{align}
The lemma below states a result  on the error probability appeared in \eqref{eq:macrateals13993}.
 \begin{lemma}  \label{lm:mac1rateerroraleq1}
When $\alpha = 1$, given the signal design in Table~\ref{tab:wiremacpara}  and  \eqref{eq:MACeqx1}-\eqref{eq:constellationalpha3},  for almost all the channel realizations, the error probability of decoding  $\{v_{1,c}, v_{2,c}\}$ from $y_1$  is vanishing, that is,
 \begin{align}
 \text{Pr} [  \{ v_{1,c} \neq \hat{v}_{1,c} \}  \cup  \{ v_{2,c} \neq \hat{v}_{2,c} \}  ]  \to 0         \quad \text {as}\quad  P\to \infty .   \label{eq:mac1rateerroraleq1}
 \end{align}
 \end{lemma}
\begin{proof}
See Appendix~\ref{sec:mac1rateerroraleq1}.
\end{proof}
By combining the results of \eqref{eq:macrateals13993}, \eqref{eq: rateomacal214} and Lemma~\ref{lm:mac1rateerroraleq1},  $ \Imu(v_1, v_2; y_1)$  can be lower bounded by 
\begin{align} 
 \Imu(v_1, v_2; y_1)    \geq    \frac{d^{'}_1+ d^{'}_2   - 2\epsilon}{2} \log P + o(\log P)   \label{eq:macrateaaleq1}
\end{align}
for almost all the channel coefficients  $\{h_{k\ell}\} \in (1, 2]^{2\times 2}$.
For the term $\Imu(v_1, v_2; y_2) $,  it can be upper bounded by 
\begin{align}
\Imu(v_1, v_2; y_2) &\leq  \Imu(v_1, v_2; y_2, v_{1, c}  +  v_{2, c} + u_1) \non  \\
&=  \Imu(v_1, v_2; v_{1, c}  +  v_{2, c} + u_1) +  \Imu(v_1, v_2; y_2 | v_{1, c}  +  v_{2, c} + u_1) \non\\
& \leq   \Hen(v_{1, c}  +  v_{2, c} + u_1) - \Hen(u_1) + \hen(y_2| v_{1, c}  +  v_{2, c} + u_1)  -  \hen(y_2|v_{1, c}  +  v_{2, c} + u_1, v_1, v_2)  \label{eq:macrateana8849384} \\
& \leq \log(6P^{ \frac{  \max\{  d^{'}_1, \ d^{'}_2\} -\epsilon}{2}}+1) - \log(2P^{ \frac{   \max\{  d^{'}_1, \ d^{'}_2\}-\epsilon }{2}}+1)  \label{eq:macrateana89238}\\
& \leq \log 3  \label{eq:macrateanaaleq18398}
\end{align}
where  \eqref{eq:macrateana89238} follows from \eqref{eq:macrateana8981212}. Due to our design in \eqref{eq:para001}-\eqref{eq:para010},  there is an integer relationship between  the minimum distance  of the constellation for  $v_{1, c}$ and the minimum distance  of  the constellation for  $v_{2, c}$.
This integer relationship allows to minimize the value of $\Hen(v_{1, c}  +  v_{2, c} + u_1) $ in \eqref{eq:macrateana8849384}, which can be treated as a GDoF penalty.
Then the results of \eqref{eq:macrateaaleq1} and \eqref{eq:macrateanaaleq18398} give
\begin{align}
\Imu(v_1, v_2; y_1) -\Imu(v_1, v_2; y_2)  \geq  \frac{d^{'}_1+ d^{'}_2 -2\epsilon}{2}  \log P + o( \log P )  \label{eq:macrateqleq1666}
\end{align}
for almost all the channel coefficients. Let us now consider the term  $\Imu(v_1; y_1|v_2) $ in \eqref{eq:macsecregion} and find its bound. By following the steps in \eqref{eq:ratebegin8390}-\eqref{eq:macrateanar17839}, the term  $\Imu(v_1; y_1|v_2) $ can be lower bounded by
\begin{align}
\Imu(v_1; y_1|v_2)   \geq    \frac{d^{'}_1- \epsilon}{2} \log P + o(\log P).   \label{eq:macrateaaleq1282}
\end{align}
Similarly, the term $ \Imu(v_2; y_1|v_1) $   in \eqref{eq:macsecregion} can be bounded by 
\begin{align}
 \Imu(v_2; y_1|v_1)    \geq    \frac{d^{'}_2- \epsilon}{2} \log P + o(\log P).   \label{eq:macrateanaleq12988}
\end{align}

Finally, by incorporating  the results of  \eqref{eq:macrateqleq1666}-\eqref{eq:macrateanaleq12988} into \eqref{eq:macsecregion}, it suggests that  the secure GDoF pair  $(d^{'}_1, d^{'}_2)$ is achievable by using $d^{'}_c= \max\{d^{'}_1, d^{'}_2\}$ GDoF of common randomness (mainly due to $u_1$),  for almost all the channel coefficients  $\{h_{k\ell}\} \in (1, 2]^{2\times 2}$.  
By moving $d^{'}_1$ from 0 to $1-d^{'}_2$ and moving  $d^{'}_2$ from 0 to 1, then we can conclude that  any GDoF pair  $(d^{'}_1, d^{'}_2) \in  \mathcal{D}^* ( 1)$ is achievable  by using $d^{'}_c= \max\{d^{'}_1, d^{'}_2\}$ GDoF of common randomness, for almost all the channel coefficients, in this case  with $ \alpha =1$.

\subsection{$\alpha>1$} \label{sec:macwiretapschemeal1893}

We have proved in Sections~\ref{sec:macwiretapschemeaaa1}-\ref{sec:macwiretapschemeaaa000} that the 
optimal secure GDoF region ${\mathcal{D}}^* ( \alpha)$ is achievable by the proposed scheme when  $ \alpha \leq 1$,  where  ${\mathcal{D}}^{*} ( \alpha) = \{(d_1, d_2) | d_1 + d_2 \leq \max\{1, \alpha\}, 0 \leq d_1 \leq   1, 0 \leq d_2 \leq \alpha\}$.

Let us consider a secure GDoF pair $(d^{'}_1, d^{'}_2)$ such that $(d^{'}_1, d^{'}_2) \in {\mathcal{D}}^* ( \alpha)$, with conditions $0\leq d^{'}_1\leq 1- d^{'}_2$  and  $0\leq d^{'}_2 \leq \alpha$,  for $  \alpha \leq 1$.  From Sections~\ref{sec:macwiretapschemeaaa1}-\ref{sec:macwiretapschemeaaa000},  it reveals that the proposed scheme is able to achieve this secure GDoF pair $(d^{'}_1, d^{'}_2)$  with a certain amount of GDoF common randomness.  For notationally convenience let us use $d^{'}_c $ to denote that amount of GDoF common randomness for achieving  the corresponding GDoF pair $(d^{'}_1, d^{'}_2)$ in the proposed scheme.  For this  secure GDoF tuple $(d^{'}_1, d^{'}_2, d^{'}_c)$,  it holds true that  
\begin{align}
(d^{'}_1, d^{'}_2, d^{'}_c) \in \bar{\mathcal{D}} ( \alpha)  \label{eq:ratechange0}
\end{align}
since it is achievable by the proposed scheme, for $ \alpha \leq 1$.
From the result of Lemma~\ref{lm:mac1switchalpha},  it also holds true that
 \begin{align}
\big(d^{'}_2/\alpha,  \   d^{'}_1/\alpha,  \ d^{'}_c/\alpha \big) \in \bar{\mathcal{D}} ( 1/\alpha). \label{eq:ratechange1}
\end{align}
In other words,  the secure GDoF tuple $(d^{'}_2/\alpha,  \   d^{'}_1/\alpha,  \ d^{'}_c/\alpha )$ is included in the region $\bar{\mathcal{D}} ( 1/\alpha)$ and  the secure GDoF pair $(d^{'}_2/\alpha,  \   d^{'}_1/\alpha )$ is included in the region ${\mathcal{D}}^*(\frac{1}{ \alpha})$,  for $ \alpha \leq 1$. 
Then, by moving $d^{'}_2$ from $0$ to $\alpha$ and moving $d^{'}_1$ from $0$ to $1-d^{'}_2$, it implies that  any point  in ${\mathcal{D}}^*(\frac{1}{ \alpha}) = \{(d_1, d_2) | d_1 + d_2 \leq \frac{1}{ \alpha}, 0 \leq d_1 \leq   1, 0 \leq d_2 \leq  \frac{1}{ \alpha}\}$ is achievable for $ \alpha \leq 1$. Let $\alpha^{'} = 1/\alpha$, we finally conclude that the optimal secure GDoF region  ${\mathcal{D}}^{*} (\alpha^{'}) = \{(d_1, d_2) | d_1 + d_2 \leq \alpha^{'}, 0 \leq d_1 \leq   1, 0 \leq d_2 \leq  \alpha^{'}\}$ is achievable for $\alpha^{'} > 1$.

\section{Converse  \label{sec:converse} }

In this section we  provide the converse proofs for Theorems~\ref{thm:ICrGDoFcr}-\ref{thm:GDoFmawccr}, regarding the minimal GDoF of the common randomness to achieve the maximal secure sum GDoF, secure  GDoF, and  the  maximal  secure  GDoF  region for interference channel,  wiretap channel with a helper, and multiple access wiretap channel, respectively. 
Let us define that \[s_{k{\ell}}(t) \defeq \sqrt{P^{\alpha_{k{\ell}}}} h_{k{\ell}} x_{\ell}(t) + z_{k}(t) \] for $k, {\ell} \in \{1,2\}, k \neq {\ell}$. Let $s_{k{\ell}}^{n} \defeq \{s_{k{\ell}}(t) \}_{t=1}^{n}$.

\subsection{Converse for two-user  interference channel} \label{sec:converseIC}

We begin with  the converse proof of Theorem~\ref{thm:ICrGDoFcr}, for the two-user  interference channel defined in Section~\ref{sec:sysICtwouser}. 
The following lemma reveals a bound on the minimal GDoF  of  common randomness $\dco (\alpha)$, for achieving the maximal secure sum GDoF $\dsum (\alpha)$.

\begin{lemma}  \label{lm:ICrGDoFcr}
Given the two-user symmetric  Gaussian  IC-SC channel  
(see~Section~\ref{sec:sysICtwouser}),  the  minimal GDoF of the common randomness $\dco (\alpha)$   for achieving the maximal secure sum GDoF $\dsum (\alpha)$   satisfies  the following inequality
 \begin{align}
\dco (\alpha) \geq   \dsum(\alpha)/2 -   (1- \alpha)^+      \quad  \quad  \alpha \in [0,  \infty). 
\end{align}
\end{lemma}
In what follows, we will prove Lemma~\ref{lm:ICrGDoFcr}.
This proof  will use the secrecy constraints and Fano's inequality. 
Starting  with the secrecy constraint $\Imu(w_1; y_{2}^{\bln})  \leq  \bln \epsilon$,  and with the identity of $\Imu(w_1; y_{2}^{\bln})= \Imu(w_1; w_c, w_2, y_{2}^{\bln}) - \Imu(w_1; w_c, w_2 |  y_{2}^{\bln})$, we have 
\begin{align}
\Imu(w_1; w_c, w_2, y_{2}^{\bln})    \leq   \Imu(w_1; w_c, w_2 |  y_{2}^{\bln}) +   \bln \epsilon  .  \label{eq:ICGaup1745}
\end{align}
The first term  in the right-hand side of  \eqref{eq:ICGaup1745} is bounded as  
\begin{align}
\Imu(w_1; w_c, w_2 |  y_{2}^{\bln})   
& = \Hen(w_c, w_2 |  y_{2}^{\bln}) - \Hen(w_c, w_2 |  y_{2}^{\bln}, w_1)    \non\\
& \leq \Hen(w_c) + \Hen( w_2 |  y_{2}^{\bln})     \label{eq:ICGaup292355} \\
& \leq \Hen(w_c) + \bln {\epsilon_n}  \label{eq: ICGaup111}
\end{align}
where \eqref{eq:ICGaup292355} uses the fact that conditioning reduces entropy; 
and \eqref{eq: ICGaup111} follows from Fano's inequality.
On the other hand,  the term  in the left-hand side of  \eqref{eq:ICGaup1745}   can be rewritten as
\begin{align}
 \Imu(w_1; w_c, w_2, y_{2}^{\bln})
& =  \Imu(w_1; y_{2}^{\bln} | w_c, w_2)  \non \\
& =  \Hen(w_1 ) \!- \! \Hen(w_1 | w_c, w_2,  y_{2}^{\bln}) \label{eq: ICGaup1111B1}
\end{align}
using the independence between  $w_1$, $w_c$ and $w_2$. 
By incorporating  \eqref{eq: ICGaup111} and  \eqref{eq: ICGaup1111B1}  into \eqref{eq:ICGaup1745}, it gives 
\begin{align}
\Hen(w_1 )   \leq  \Hen(w_c)  + \Hen(w_1 | w_c, w_2,  y_{2}^{\bln})  +  \bln {\epsilon_n}  +   \bln \epsilon  .    \label{eq: ICGaup1111D}
\end{align}
For the second term in the right-hand side of \eqref{eq: ICGaup1111D}, we have 
\begin{align}
&   \Hen(w_1 | w_c, w_2,  y_{2}^{\bln}) \non\\
 = &\Hen(w_1 | w_c, w_2, y_{2}^{\bln},  s_{21}^{\bln}) \label{eq: ICGaup13424} \\
 \leq   & \Hen(w_1 | w_c, w_2, y_{2}^{\bln},  s_{21}^{\bln})      -  \Hen( w_1 |  y_{1}^{\bln})  +  \bln {\epsilon}_n  \label{eq: ICGaup112A}\\
 \leq & \Hen(w_1 | w_c, w_2,  s_{21}^{\bln})  -  \Hen( w_1 |  y_{1}^{\bln}, w_c, w_2,  s_{21}^{\bln})  +  \bln {\epsilon}_n  \non\\
 =  & \Imu (w_1 ;  y_{1}^{\bln} | w_c, w_2,  s_{21}^{\bln})   +  \bln {\epsilon}_n   \non \\
 = &  \Imu (w_1;  \{y_1(t) - \sqrt{P^{1-\alpha}} \cdot \frac{h_{11}}{h_{21}} s_{21}(t)    -   \sqrt{P^{\alpha}}h_{12}x_2(t)\}_{t=1}^{\bln}  | w_c,  w_2,  s_{21}^{\bln} )  +  \bln {\epsilon}_n \label{eq: ICGaup112B}\\
= & \Imu (w_1 ;  \{ -\sqrt{P^{1-\alpha}}\cdot \frac{h_{11}}{h_{21}} z_2 (t) + z_1(t)\}_{t=1}^{\bln} | w_c, w_2,  s_{21}^{\bln})   +  \bln {\epsilon}_n  \non \\
=&  \hen (\{ -\sqrt{P^{1-\alpha}}\cdot \frac{h_{11}}{h_{21}} z_2 (t) + z_1(t)\}^{\bln} | w_c, w_2,  s_{21}^{\bln})      -  \hen ( \{ -\sqrt{P^{1-\alpha}} \frac{h_{11}}{h_{21}} z_2 (t) + z_1(t)\}^{\bln} | w_1, w_c, w_2,  s_{21}^{\bln})  +  \bln {\epsilon}_n   \non\\
= & \hen (\{ -\sqrt{P^{1-\alpha}}\cdot \frac{h_{11}}{h_{21}} z_2 (t) + z_1(t)\}^{\bln} | w_c, w_2,  s_{21}^{\bln}) -  \hen (z_1^{\bln})  +  \bln {\epsilon}_n    \label{eq: ICGaup12825}\\
  \leq  & \frac{n}{2} \log  (1+P^{1-\alpha}\cdot \frac{|h_{11}|^2}{|h_{21}|^2})    +  \bln {\epsilon}_n \label{eq: ICGaup113}
\end{align}
where \eqref{eq: ICGaup13424} follows from the fact that $s_{21}^{\bln}$  can be reconstructed by $\{w_c, w_2,  y_{2}^{\bln}\}$; 
\eqref{eq: ICGaup112A} is from Fano's inequality;
\eqref{eq: ICGaup112B} uses the fact that  $x_2^{\bln}$ is a function of $(w_c,  w_2)$;
\eqref{eq: ICGaup12825} results from the fact that $z_2^{\bln}$ can be reconstructed from $\{w_1, w_c, w_2,  s_{21}^{\bln}\}$; 
\eqref{eq: ICGaup113} follows from the identity that conditioning reduces differential entropy and the identity that $\hen (z_1^{\bln}) = \frac{n}{2} \log(2\pi e)$. 
Finally, given that $\Hen(w_1 ) =  nR_1$ and  $\Hen(w_c ) =  nR_c$,  combining  the results of \eqref{eq: ICGaup1111D} and \eqref{eq: ICGaup113}  gives the following inequality 
\begin{align}
  nR_1    -  \frac{n}{2} \log  (1+P^{1-\alpha}\cdot \frac{|h_{11}|^2}{|h_{21}|^2})  - \bln {\epsilon'_n}   \leq  nR_c   \label{eq: ICGaup114A}
\end{align}
for $\epsilon'_n = 2\epsilon_n + \epsilon$. 
Due to the symmetry, by exchanging the roles of user~1 and user~2, we also have
\begin{align}
   nR_2  -   \frac{n}{2} \log  (1+P^{1-\alpha}\cdot \frac{|h_{22}|^2}{|h_{12}|^2})  - \bln \epsilon'_n  \leq    nR_c.  \label{eq: ICGaup114B}
\end{align}
Based on the definitions of  $\dco (\alpha)$ and $\dsum (\alpha)$ in Section~\ref{sec:sysICtwouser}, combining the results of \eqref{eq: ICGaup114A} and \eqref{eq: ICGaup114B} produces the following bound  
 \begin{align}
  \dsum(\alpha)/2 -   (1- \alpha)^+   \leq \dco (\alpha),   \quad  \quad \forall \alpha \in [0,  \infty)
\end{align}
which completes the proof of Lemma~\ref{lm:ICrGDoFcr}.

\subsection{Converse for the wiretap channel with a helper} \label{sec:conversewiretap}

Let us now focus on  the converse proof of Theorem~\ref{thm:GDoFwthrcr} for the wiretap channel with a helper.  
The following lemma reveals a bound on the minimal GDoF  of  common randomness $\dco (\alpha)$, for achieving the maximal secure  GDoF, i.e.,  $\dso (\alpha)=1$ for any $\alpha \in [0,  \infty)$ (see Theorem~\ref{thm:GDoFwthr}).

\begin{lemma}\label{lm:GDoFwthrcr}
Given  the \emph{symmetric}  Gaussian  WTH channel (see~Section~\ref{sec:syswchr}), 
the  minimal GDoF of the common randomness $\dco (\alpha)$   for achieving the maximal secure  GDoF satisfies  the following inequality
 \begin{align}
\dco (\alpha) \geq  1-   (1- \alpha)^+      \quad \quad  \alpha \in [0,  \infty).
\end{align}
\end{lemma}

The proof of Lemma~\ref{lm:GDoFwthrcr} follows closely from that of Lemma~\ref{lm:ICrGDoFcr}.
From the secrecy constraint $\Imu(w_1; y_{2}^{\bln})  \leq  \bln \epsilon$,  and  the identity of $\Imu(w_1; y_{2}^{\bln})= \Imu(w_1; w_c,  y_{2}^{\bln}) - \Imu(w_1; w_c |  y_{2}^{\bln})$, we have 
\begin{align}
\Imu(w_1; w_c, y_{2}^{\bln})    \leq   \Imu(w_1; w_c |  y_{2}^{\bln}) +   \bln \epsilon  .  \label{eq:wiretapGaup1745}
\end{align}
The first term  in the right-hand side of  \eqref{eq:wiretapGaup1745} is bounded as  
\begin{align}
\Imu(w_1; w_c |  y_{2}^{\bln})   
& = \Hen(w_c |  y_{2}^{\bln}) - \Hen(w_c |  y_{2}^{\bln}, w_1)    \leq \Hen(w_c)   \label{eq: wiretapGaup111}
\end{align}
On the other hand,  the term  in the left-hand side of  \eqref{eq:wiretapGaup1745}   can be rewritten as
\begin{align}
 \Imu(w_1; w_c, y_{2}^{\bln})
& =  \Imu(w_1; y_{2}^{\bln} | w_c)  \non \\
& =  \Hen(w_1 | w_c) \!- \! \Hen(w_1 | w_c,  y_{2}^{\bln}) \non\\
& =  \Hen(w_1) \!- \! \Hen(w_1 | w_c,  y_{2}^{\bln}) \label{eq: wiretapGaup1111B1}
\end{align}
using the independence between  $w_1$ and $w_c$. 
By incorporating  \eqref{eq: wiretapGaup111} and  \eqref{eq: wiretapGaup1111B1}  into \eqref{eq:wiretapGaup1745}, it gives 
\begin{align}
\Hen(w_1 )   \leq  \Hen(w_c)  + \Hen(w_1 | w_c, y_{2}^{\bln})    +   \bln \epsilon  .    \label{eq: wiretapGaup1111D}
\end{align}
For the second term in the right-hand side of \eqref{eq: wiretapGaup1111D}, by following the steps in \eqref{eq: ICGaup13424}-\eqref{eq: ICGaup113} we  have
\begin{align}
&   \Hen(w_1 | w_c,  y_{2}^{\bln}) \non\\
 = &\Hen(w_1 | w_c,  y_{2}^{\bln},  s_{21}^{\bln}) \label{eq: wiretapGaup13424} \\
 \leq   & \Hen(w_1 | w_c,  y_{2}^{\bln},  s_{21}^{\bln})      -  \Hen( w_1 |  y_{1}^{\bln})  +  \bln {\epsilon}_n  \label{eq: wiretapGaup112A}\\
 \leq & \Hen(w_1 | w_c, s_{21}^{\bln})  -  \Hen( w_1 |  y_{1}^{\bln}, w_c,  s_{21}^{\bln})  +  \bln {\epsilon}_n  \non\\
 =  & \Imu (w_1 ;  y_{1}^{\bln} | w_c,  s_{21}^{\bln})   +  \bln {\epsilon}_n   \non \\
 = &  \Imu (w_1;  \{y_1(t) - \sqrt{P^{1-\alpha}} \cdot \frac{h_{11}}{h_{21}} s_{21}(t)   -   \sqrt{P^{\alpha}}h_{12}x_2(t)\}_{t=1}^{\bln}  | w_c,  s_{21}^{\bln} )  +  \bln {\epsilon}_n \label{eq: wiretapGaup112B}\\
= & \Imu (w_1 ;  \{ -\sqrt{P^{1-\alpha}}\cdot \frac{h_{11}}{h_{21}} z_2 (t) + z_1(t)\}_{t=1}^{\bln} | w_c,  s_{21}^{\bln})   +  \bln {\epsilon}_n  \non \\
 \leq  & \frac{n}{2} \log  (1+P^{1-\alpha}\cdot \frac{|h_{11}|^2}{|h_{21}|^2})    +  \bln {\epsilon}_n \label{eq: wiretapGaup113}
\end{align}
where \eqref{eq: wiretapGaup13424} follows from the fact that $s_{21}^{\bln}$  can be reconstructed by $\{w_c,  y_{2}^{\bln}\}$; 
\eqref{eq: wiretapGaup112A} is from Fano's inequality;
\eqref{eq: wiretapGaup112B} uses the fact that  $x_2^{\bln}$ is a function of $w_c$;
\eqref{eq: wiretapGaup113} follows from the identity that conditioning reduces differential entropy. 
Finally, given that $\Hen(w_1 ) =  nR_1$ and  $\Hen(w_c ) =  nR_c$,  combining  the results of \eqref{eq: wiretapGaup1111D} and \eqref{eq: wiretapGaup113}  gives the following inequality 
\begin{align}
  nR_1    -  \frac{n}{2} \log  (1+P^{1-\alpha}\cdot \frac{|h_{11}|^2}{|h_{21}|^2})  - \bln {\epsilon'_n}   \leq  nR_c   \label{eq: wiretapGaup114A}
\end{align}
for $\epsilon'_n =\epsilon_n + \epsilon$. 
Based on the definitions of  $\dco (\alpha)$ and $ \dso (\alpha)$ in Section~\ref{sec:syswchr}, and given the result of Theorem~\ref{thm:GDoFwthr}, i.e.,   $\dso (\alpha) =1, \forall \alpha \in [0,  \infty) $, the result of  \eqref{eq: wiretapGaup114A} gives the following bound  
 \begin{align}
1 -   (1- \alpha)^+   \leq \dco (\alpha),   \quad  \quad \forall \alpha \in [0,  \infty)
\end{align}
which completes the proof of Lemma~\ref{lm:GDoFwthrcr}.

\subsection{Converse for two-user multiple access wiretap channel} \label{sec:conversewiretapmac}

Let us consider the converse proof of Theorem~\ref{thm:GDoFmawccr} for the two-user   multiple access wiretap channel. 
The following lemma gives a bound on  the minimal GDoF of  common randomness $\dco (\alpha, d_1, d_2)$,  for achieving   any given GDoF pair $(d_1, d_2)$ in the  maximal  secure  GDoF  region  $\mathcal{D}^* ( \alpha) $.

\begin{lemma} \label{lm:GDoFmawccr}
For  the \emph{symmetric}  Gaussian MAC-WT channel with common randomness (see~Section~\ref{sec:sysmawc}),   the minimal GDoF of the common randomness $\dco (\alpha, d_1, d_2)$, for achieving  any given GDoF pair $(d_1, d_2)$ in the  maximal  secure  GDoF  region  $\mathcal{D}^* ( \alpha) $,   satisfies the  following inequality
 \begin{align}
\dco (\alpha, d_1, d_2) \geq  \max\{ d_1 - (1-\alpha)^+,  d_2- (\alpha-1)^+ \}     \quad  \text{for} \quad   (d_1, d_2) \in  \mathcal{D}^* ( \alpha).     \non 
\end{align}
\end{lemma}

The following corollary is directly from Lemma~\ref{lm:GDoFmawccr}  by considering  the specific case  with $\alpha =1 $.

\begin{corollary} \label{cor:GDoFmawccr}
For  the \emph{symmetric}  Gaussian  MAC-WT  channel with common randomness  defined in Section~\ref{sec:sysmawc}, and for $\alpha=1$, the minimal GDoF of the common randomness $\dco (1, d_1, d_2)$,  for achieving   any given GDoF pair $(d_1, d_2)$ in the  maximal  secure  GDoF  region  $\mathcal{D}^* ( 1) $,   satisfies the  following inequality
 \begin{align}
\dco (1, d_1, d_2) \geq  \max\{d_1,  d_2  \}     \quad  \text{for} \quad   (d_1, d_2) \in  \mathcal{D}^* ( 1) ,  \  \alpha =1 .     \non 
\end{align}
\end{corollary}

Let us now prove Lemma~\ref{lm:GDoFmawccr}.
This proof  will also use the secrecy constraints and Fano's inequality. 
Starting  with the secrecy constraint $\Imu(w_1, w_2; y_{2}^{\bln})  \leq  \bln \epsilon$,  and with the identity of $\Imu(w_1, w_2; y_{2}^{\bln})= \Imu(w_1, w_2; w_c,  y_{2}^{\bln}) - \Imu(w_1, w_2; w_c |  y_{2}^{\bln})$, we have 
\begin{align}
\Imu(w_1,  w_2; w_c, y_{2}^{\bln})    \leq   \Imu(w_1 , w_2; w_c |  y_{2}^{\bln}) +   \bln \epsilon  .  \label{eq: mawcGaup1745}
\end{align}
The first term  in the right-hand side of  \eqref{eq: mawcGaup1745} is bounded as  
\begin{align}
\Imu(w_1,  w_2; w_c |  y_{2}^{\bln})   
& = \Hen(w_c |  y_{2}^{\bln}) - \Hen(w_c |  y_{2}^{\bln}, w_1,  w_2)    \non\\
& \leq \Hen(w_c) . \label{eq: mawcGaup111}
\end{align}
On the other hand,  the term  in the left-hand side of  \eqref{eq: mawcGaup1745}   can be rewritten as
\begin{align}
 \Imu(w_1,  w_2; w_c, y_{2}^{\bln})  
& =  \Imu(w_1, w_2; y_{2}^{\bln} | w_c)  \non \\
& =  \Hen(w_1, w_2 ) \!- \! \Hen(w_1,  w_2 | w_c,  y_{2}^{\bln})  \non\\ 
&  = \Hen(w_1) + \Hen( w_2 ) \!- \! \Hen(w_1,  w_2 | w_c,  y_{2}^{\bln})  \label{eq: mawcGaup1111B1}
\end{align}
using the independence between $w_1$, $w_c$ and $w_2$. 
By incorporating  \eqref{eq: mawcGaup111} and  \eqref{eq: mawcGaup1111B1}  into \eqref{eq: mawcGaup1745}, it gives 
\begin{align}
\Hen(w_2 )   \leq   &\Hen(w_c) -  \Hen(w_1) + \Hen(w_1,  w_2 | w_c,  y_{2}^{\bln})  +   \bln \epsilon   \non \\ 
 \leq & \Hen(w_c)  -  \Hen(w_1) +  \Hen(w_1,  w_2 | w_c,  y_{2}^{\bln})     -  \Hen( w_1, w_2 |  y_{1}^{\bln})  +  \bln {\epsilon}_n  +   \bln \epsilon  \label{eq: mawcGaup112A}\\
=  & \Hen(w_c)  \underbrace{-  \Hen(w_1) +  \Hen(w_1 | w_c,  y_{2}^{\bln}) }_{\leq 0}+  \Hen( w_2 | w_1,  w_c,  y_{2}^{\bln})  \underbrace{-  \Hen( w_1 |  y_{1}^{\bln})}_{\leq 0} - \Hen( w_2 | w_1,  y_{1}^{\bln}) + \bln  \epsilon'_n  \non \\
\leq  & \Hen(w_c) +  \Hen( w_2 | w_1,  w_c,  y_{2}^{\bln}) - \Hen( w_2 | w_1,  y_{1}^{\bln}) + \bln \epsilon'_n   \label{eq: mawcGaup1111D} 
\end{align}
for $\epsilon'_n \defeq \epsilon_n + \epsilon$, where \eqref{eq: mawcGaup112A} stems from Fano's equality.
For the second and third terms in the right-hand side of \eqref{eq: mawcGaup1111D}, we have 
\begin{align}
&  \Hen( w_2 | w_1,  w_c,  y_{2}^{\bln}) - \Hen( w_2 | w_1,  y_{1}^{\bln}) \non \\
\leq & \Hen( w_2 | \{\sqrt{P}h_{22}x_2(t) + z_2(t)\}_{t=1}^{\bln}, w_1,  w_c,  y_{2}^{\bln}) - \Hen( w_2 |  \{\sqrt{P}h_{22}x_2(t) + z_2(t)\}_{t=1}^{\bln}, w_1, w_c,  y_{1}^{\bln})  \label{eq: mawcGaup112455B} \\
\leq & \Hen( w_2 |  \{\sqrt{P}h_{22}x_2(t) + z_2(t)\}_{t=1}^{\bln}, w_1,  w_c) - \Hen( w_2 | \{\sqrt{P}h_{22}x_2(t) + z_2(t)\}_{t=1}^{\bln}, w_1, w_c,  y_{1}^{\bln})    \label{eq: mawcGaup1124B} \\
=  & \Imu (w_2 ;  y_{1}^{\bln} | \{\sqrt{P}h_{22}x_2(t) + z_2(t)\}_{t=1}^{\bln}, w_1,  w_c)     \non\\
 = &  \Imu (w_2;  \{y_1(t) - \sqrt{P^{\alpha-1}} \cdot \frac{h_{12}}{h_{22}} (\sqrt{P}h_{22}x_2(t) + z_2(t))   -   \sqrt{P}h_{11}x_1(t)\}_{t=1}^{\bln}  | w_c,  w_1,  \{\sqrt{P}h_{22}x_2(t) + z_2(t)\}_{t=1}^{\bln} )  \label{eq: mawcGaup112B}\\
= & \Imu (w_2 ;  \{ -\sqrt{P^{\alpha-1}}\cdot \frac{h_{12}}{h_{22}} z_2 (t) + z_1(t)\}_{t=1}^{\bln} | w_c, w_1,  \{\sqrt{P}h_{22}x_2(t) + z_2(t)\}_{t=1}^{\bln})   \non \\
 \leq  & \frac{n}{2} \log  (1+P^{\alpha-1}\cdot \frac{|h_{12}|^2}{|h_{22}|^2})     \label{eq: mawcGaup113}
\end{align}
where \eqref{eq: mawcGaup112455B} follows from the fact that $\{\sqrt{P}h_{22}x_2(t) + z_2(t)\}_{t=1}^{\bln}$  can be reconstructed by $\{w_1, w_c,  y_{2}^{\bln}\}$ and the fact that conditioning reduces entropy;
\eqref{eq: mawcGaup1124B} results from the  fact that conditioning reduces entropy;
\eqref{eq: mawcGaup112B} uses the fact that  $x_1^{\bln}$ is a function of $(w_c,  w_1)$;
\eqref{eq: mawcGaup113} follows from the identity that conditioning reduces differential entropy. 
Combining  the results of \eqref{eq: mawcGaup1111D} and \eqref{eq: mawcGaup113},  it gives the following inequality 
\begin{align}
\Hen(w_2 )  - \frac{n}{2} \log  (1+P^{\alpha-1}\cdot \frac{|h_{12}|^2}{|h_{22}|^2})    -  \bln\epsilon'_n  \leq  \Hen(w_c)  .  \label{eq: mawcGaup174897}
\end{align}
  Finally, given that  $\Hen(w_2 ) =  nR_2$ and  $\Hen(w_c ) =  nR_c$,   \eqref{eq: mawcGaup174897}  implies the following inequality 
\begin{align}
  nR_2    -  \frac{n}{2}  \log  (1+P^{\alpha-1}\cdot \frac{|h_{12}|^2}{|h_{22}|^2})  - \bln {\epsilon'_n}   \leq  nR_c  .  \label{eq: mawcGaup114B}
\end{align}
On the other hand,  by interchanging the roles of transmitter~1 and transmitter~2, we also have 
\begin{align}
  nR_1    -  \frac{n}{2} \log  (1+P^{1-\alpha}\cdot \frac{|h_{11}|^2}{|h_{21}|^2})  - \bln {\epsilon'_n}   \leq  nR_c  . \label{eq: mawcGaup114C}
\end{align}
Based on the definition of  $ \dco (\alpha, d_1, d_2) $ in Section~\ref{sec:sysmawc},  the results of  \eqref{eq: mawcGaup114B} and  \eqref{eq: mawcGaup114C}  give the following bound  
 \begin{align}
\max\{ d_1 - (1-\alpha)^+,  d_2- (\alpha-1)^+ \}    \leq \dco (\alpha, d_1, d_2),  \quad  \text{for} \quad   (d_1, d_2) \in  \mathcal{D}^* ( \alpha )
\end{align}
which completes the  proof of Lemma~\ref{lm:GDoFmawccr}.

\section{Conclusion}   \label{sec:concl}

In this work we showed that adding common randomness at the transmitters  \emph{totally} removes the penalty in sum GDoF or GDoF region of  three basic channels. 
The results reveal that adding common randomness at the transmitters is a constructive way to remove the secrecy constraints in communication networks in terms of GDoF performance. 
Another contribution of this work is the characterization on the  minimal  GDoF of the common randomness to remove the  secrecy constraints.
In the future work, we will focus on how to remove the secrecy constraints in the other communication networks.

\appendices

\section{Proof of Lemma~\ref{lm:ICrateerror537}  \label{sec:ICrateerror537} }

We now prove Lemma~\ref{lm:ICrateerror537}.  
 Let us first provide  \cite[Lemma~1]{ChenIC:18}  that will be used in the proof.
\begin{lemma}  \label{lm:AWGNic} \cite[Lemma~1]{ChenIC:18}
Consider a specific channel model $ y'=  \sqrt{P^{\alpha_1}}  h' x' + \sqrt{P^{\alpha_2}} g' + z'$, 
where $x' \in \Omega (\xi,  Q)$, and $z' \sim \mathcal{N}(0, \sigma^2)$.  $g'  \in  \Sc_{g'}$ is a discrete random variable with a condition \[|g' | \leq  g_{\max}, \quad \forall g' \in  \Sc_{g'}\] for $\Sc_{g'} \subset \Rc$.   $h'$, $g_{\max}$, $\sigma$,  $\alpha_2$ and $\alpha_1$ are finite and positive constants  that are independent of $P$. Also consider the condition $\alpha_1 - \alpha_2 >0$. Let $\gamma' > 0$ be a finite parameter.
If we set  $Q$ and $\xi$ by
\begin{align}
Q =  \frac{P^{\frac{\alpha'}{2}}  h' \gamma' }{2 g_{\max} },    \quad  \quad    \xi =  \frac{\gamma'}{Q},   \quad   \quad   \forall  \alpha' \in (0, \alpha_1 -\alpha_2)  \label{eq:settingAWGNIC}                           
\end{align}
then the  probability of error for decoding $x'$ from $y'$ satisfies \[ \text{Pr} (e) \to 0  \quad  \text{as} \quad   P\to \infty.\] 
\end{lemma}
\vspace{10pt}

In this proof we will first estimate $v_{1,c}$ and  $v_{2,c} + u$   from the observation $y_1$  expressed in \eqref{eq:IC321y1}  with noise removal and signal separation methods, and then estimate $v_{1,p}$. Note that  $v_{1,c} \in \Omega ( \xi =  \frac{ \gamma}{Q} ,   \   Q =  P^{ \frac{ \alpha/2 - \epsilon }{2}} )  $, $v_{2,c}  +  u  \in 2 \Omega ( \xi = \frac{ \gamma}{Q} ,   \   Q =  P^{ \frac{ \alpha/2 - \epsilon }{2}} )$, and $v_{1,p}, v_{2,p} \in    \Omega (\xi   =\frac{ \gamma}{2Q},   \   Q = P^{ \frac{ 1 - \alpha - \epsilon}{2}} )$, where  $2\cdot \Omega (\xi,  Q)  \defeq   \{ \xi \cdot a :   \    a \in  \Zc  \cap [-2Q,   2Q]   \}$.   
For the observation   $y_1$ in \eqref{eq:IC321y1}, it can be expressed as 
\begin{align}
y_{1} & =   \sqrt{P}  h_{11} v_{1,c}   +     \sqrt{P^{ \alpha }} h_{12}  (   v_{2,c}  +  u)     + \tilde{z}_{1}   \non\\
& =    P^{{ \frac{ \alpha/2 + \epsilon }{2}}} \cdot  \gamma \cdot  (  \sqrt{P^{ 1- \alpha }}  g_0 q_0 + g_1 q_1 )  + \tilde{z}_{1} \label{eq:rewrittenIC321y1}
\end{align}
where   $\tilde{z}_{1}   \defeq      \sqrt{P^{ 1 - \alpha}}  h_{11} v_{1,p} +   \sqrt{P^{ (\tau+1)\alpha-\tau}}  \delta_{2, \tau}   h_{12}  u +    h_{12} v_{2,p} +  z_{1} $ and 
 \[ g_0\defeq  h_{11},     \quad  g_1\defeq   h_{12}, \quad  q_0  \defeq  \frac{Q_{\max}}{\gamma} \cdot   v_{1,c}   ,    \quad  q_1  \defeq  \frac{Q_{\max}}{\gamma} \cdot   (   v_{2,c}  +  u),  \quad Q_{\max} \defeq P^{ \frac{ \alpha/2 - \epsilon }{2}} \]
   for  $\gamma  \in \bigl(0, \frac{1}{ \tau \cdot 2^\tau}\bigr]$.
It is true that $q_0, q_1 \in \Zc$, $|q_1| \leq 2 Q_{\max}$ and $|q_0| \leq Q_{\max}$.

Let us consider $ \hat{q}_1$ and $\hat{q}_0$  as  the corresponding  estimates of  $q_1$ and  $ q_0$ from $y_1$   (see \eqref{eq:rewrittenIC321y1}). For this estimation we specifically consider  an estimator that  seeks to  minimize
\[ |y_1 -   P^{{ \frac{ \alpha/2 + \epsilon }{2}}} \cdot  \gamma \cdot  (  \sqrt{P^{ 1- \alpha }}  g_0 \hat{q}_0 + g_1 \hat{q}_1 ) |.   \]
The minimum distance  defined below will be used in the error probability analysis for this estimation
  \begin{align}
 d_{\min}  (g_0, g_1)    \defeq    \min_{\substack{ q_0, q_0' \in \Zc  \cap [- Q_{\max},    Q_{\max}]  \\  q_1, q_1' \in \Zc  \cap [- 2Q_{\max},   2 Q_{\max}]   \\  (q_0, q_1) \neq  (q_0', q_1')  }}  |  \sqrt{P^{ 1- \alpha }}  g_0  (q_0 - q_0') +  g_1 (q_1 - q_1')  |  .   \label{eq:minidis737539}
 \end{align}
 Lemma~\ref{lm:distance5723946} (see below) will reveal that, for almost all channel  realizations, this minimum distance is sufficiently large  when $P$ is large. The result of  \cite[Lemma~1]{ChenLiArxiv:18}  shown below will be used  in the proof.
\begin{lemma}   \label{lm:NMb2} \cite[Lemma~1]{ChenLiArxiv:18}
Let us consider a parameter $\eta$ such that $ \eta >1$ and $\eta  \in \Zc^+$, and consider  $\beta \in (0,1]$ and $Q_0, A_0, Q_1,  A_1 \in \Zc^+$.  Also define  two events as
\begin{align}
 \tilde B(q_0, q_1)  \defeq \{ (g_0, g_1)  \in (1, \eta]^2 :  |  A_1 g_1 q_1 + A_0 g_0 q_0   | < \beta \}       \label{eq:lemmabound0098}  
\end{align}
and  
\begin{align}
\tilde B  \defeq   \bigcup_{\substack{ q_0, q_1 \in \Zc:  \\  (q_0, q_1) \neq  0,  \\  |q_k| \leq Q_k  \ \forall k }}  \tilde B(q_0, q_1) .       \label{eq:lemmabound125}
\end{align}
For $\Lc (\tilde B)$ denoting the Lebesgue measure of $\tilde B $, then this measure  is bounded as
\begin{align*}
\Lc (\tilde B )  \leq   8 \beta   (\eta -1) \min\Bigl\{ \frac{ Q_0 Q_1}{A_1},  \frac{Q_1 Q_0}{A_0},   \frac{Q_0 \eta}{A_1},   \frac{Q_1 \eta}{A_0}\Bigr\}.
\end{align*}
\end{lemma}

\begin{lemma}  \label{lm:distance5723946}
For $\epsilon >0$  and $\kappa \in (0, 1]$ and consider  the  design in  \eqref{eq:xvkkk1}-\eqref{eq:constellationGsym2}
and \eqref{eq:IC321x1}-\eqref{eq:IC321x2}  when $\alpha \in [2/3, 1)$.   Then the following inequality holds true for the minimum distance $d_{\min}$ defined in \eqref{eq:minidis737539}
 \begin{align}
d_{\min}    \geq   \kappa P^{- \frac{3\alpha/2 -1 }{2}}    \label{eq:distancegeq}
 \end{align}
for all   channel  realizations $\{h_{k\ell}\} \in (1, 2]^{2\times 2} \setminus \Ho$, where  the  Lebesgue measure of $\Ho \subseteq (1,2]^{2\times 2}$ satisfies  
 \begin{align}
\mathcal{L}(\Ho) \leq 64 \kappa   \cdot     P^{ - \frac{ \epsilon  }{2}} .    \label{eq:Lebb4893}
 \end{align}
 \end{lemma}
 \begin{proof}
For this case  we set $\eta \defeq 2$ and define 
\[ \beta \defeq    \kappa P^{- \frac{3\alpha/2 -1 }{2}} , \quad A_0 \defeq   \sqrt{P^{ 1- \alpha }} ,  \quad  A_1 \defeq  1,  \quad Q_0 \defeq 2 Q_{\max},    \quad Q_1\defeq 4 Q_{\max},  \quad  Q_{\max} \defeq P^{ \frac{ \alpha/2 - \epsilon }{2}}  .     \]
From the previous definitions, $g_0=  h_{11} $  and  $g_1=   h_{12}$. Without loss of generality (WLOG) we will consider the  case  that\footnote{Our result also holds true for the scenario when any of the four parameters $\{A_{0},  Q_{0},  A_{1}, Q_{1}\}$ isn't  an integer. It just requires  some minor modifications in the proof. Let us consider one example when $A_0$ isn't an integer. For this example,  the parameters $A_0$ and $ g_0$ can be replaced with $A'_0 \defeq  \omega_0  A_0$ and  $g'_0 \defeq \frac{1}{\omega_0}  g_0$, respectively, where $\omega_0 \defeq  {\lceil A_0\rceil}/ A_0$. From the definition,   $\omega_0$ is bounded, i.e.,  $1/2< 1/\omega_0 < 1$, and $A'_0=\omega_0  A_0$ is an integer. 
Let us consider another example when $Q_{0} = 2\sqrt{P^{\alpha/2 - \epsilon}}$  isn't an integer. For this example,  the parameter $\epsilon$ can be slightly modified  such that $Q_{0} = 2\sqrt{P^{\alpha/2 - \epsilon}}$ is an integer and $\epsilon$ can still be very small when  $P$ is large. Therefore, throughout this work, WLOG we will consider   those parameters  to be  integers, i.e.,  $ A_{0},  Q_{0},  A_{1}, Q_{1} \in \Zc^+$. 
}  $ A_{0},  Q_{0},  A_{1}, Q_{1}\in \Zc^+$.  
 Let us define
\begin{align}
  \tilde B(q_0, q_1)  \defeq \{ (g_0, g_1)  \in (1, \eta]^2 :  | A_1 g_1 q_1 +  A_0 g_0 q_0  | < \beta \}     \label{eq:BBB776}
\end{align}
and  
\begin{align}
\tilde B  \defeq   \bigcup_{\substack{ q_0, q_1 \in \Zc:  \\  (q_0, q_1) \neq  0,  \\  |q_k| \leq Q_k  \ \forall k }}  \tilde B(q_0, q_1) .    \label{eq:BBBtilde552}
\end{align}
From Lemma~\ref{lm:NMb2}, the Lebesgue measure of $\tilde B $ can be bounded by 
\begin{align}
\Lc (\tilde B )  & \leq   8 \beta (\eta -1)  \min\{ \frac{ 8 {Q^2}_{\max}}{ 1},  \frac{ 8 {Q^2}_{\max}}{  \sqrt{P^{ 1- \alpha }}},   \frac{2 Q_{\max} \eta}{1},   \frac{4 Q_{\max} \eta}{ \sqrt{P^{ 1- \alpha }}}\}  \label{eq:LeBmeasurest}\\
& =   8 \beta (\eta -1) \cdot Q_{\max} \cdot \min\{8Q_{\max},  \frac{ 8Q_{\max}}{\sqrt{P^{ 1- \alpha }}},   2 \eta,   \frac{ 4\eta}{ \sqrt{P^{ 1- \alpha }}}\} \non\\ 
& \leq   8\beta  (\eta -1) \cdot Q_{\max} \cdot P^{\frac{ \alpha -1 }{2}}  \cdot \min\{8Q_{\max},   4\eta \} \non\\ 
& \leq   8\beta (\eta -1) \cdot Q_{\max} \cdot P^{\frac{ \alpha -1 }{2}}  \cdot    4\eta  \non\\ 
& = 32 \beta \eta (\eta -1) \cdot  P^{\frac{3 \alpha/2 -1 -  \epsilon }{2}} \non \\
& = 32 \eta (\eta -1) \kappa        P^{ - \frac{ \epsilon  }{2}}  \non\\
& = 64 \kappa        P^{ - \frac{ \epsilon  }{2}}.    \label{eq:LeBmeasure}
\end{align}
Based on the  definition in \eqref{eq:BBBtilde552},  $\tilde B$ is a set of  $(g_0, g_1)$, where $(g_0, g_1)  \in (1,\eta]^2$.   For any $(g_0, g_1) \in \tilde B$,  there exists at least one pair  $(q_0, q_1)$ such that $|A_1 g_1 q_1 + A_0  g_0q_0  | <   \kappa P^{- \frac{3\alpha/2 -1 }{2}}$. Thus, $\tilde B$ can be treated as the  outage set and we have the following conclusion: 
  \[d_{\min} (g_0, g_1)   \geq    \kappa P^{- \frac{3\alpha/2 -1 }{2}} , \quad \text{for}  \quad  (g_0, g_1)\notin \tilde B .\]

Let us now define  $\Ho$  as a set of  $(h_{22}, h_{21},  h_{12}, h_{11}) \in (1, 2]^{2\times 2}$   such that the corresponding pairs $(g_0, g_1)$ appear  in  $\tilde B$ (outage set), that is, 
\begin{align}
 \Ho \defeq \{  (h_{22}, h_{21}, h_{12}, h_{11} ) \in (1, 2]^{2\times 2} :      (g_0, g_1) \in \tilde B  \} .   \label{eq:HO4422}
 \end{align}
From the relationship between  $\tilde B$ and  $\Ho$, the Lebesgue measure of $\Ho$  can be bounded by 
\begin{align}
 \Lc (\Ho ) & =   \int_{h_{22}=1}^2      \int_{h_{21}=1}^2     \int_{h_{12}=1}^2      \int_{h_{11}=1}^2    \mathbbm{1}_{\Ho}    (h_{22},   h_{21}, h_{12},   h_{11} )   d h_{11}   d h_{12}  d h_{21} d h_{22}                 \label{eq:LeB38985}\\  
& =  \int_{h_{22}=1}^2      \int_{h_{21}=1}^2     \int_{h_{12}=1}^2      \int_{h_{11}=1}^2  \mathbbm{1}_{\tilde B}  (h_{11}, h_{12})   d h_{11}   d h_{12}  d h_{21} d h_{22}         \non\\ 
& \leq  \int_{h_{22}=1}^2      \int_{h_{21}=1}^2   \int_{g_{1}=1}^{\eta}   \int_{g_{0}=1}^{\eta}      \mathbbm{1}_{\tilde B}  (g_0, g_{1}) d g_{0} d g_{1} d h_{21} d h_{22}             \non \\
& =    \int_{h_{22}=1}^2      \int_{h_{21}=1}^2 \mathcal{L}(\tilde B)  d h_{21} d h_{22}                 \non\\ 
& \leq    \int_{h_{22}=1}^2      \int_{h_{21}=1}^2  64 \kappa   \cdot     P^{ - \frac{ \epsilon  }{2}}  d h_{21} d h_{22}              \label{eq:LeB4859028}  \\     
& =    64 \kappa   \cdot     P^{ - \frac{ \epsilon  }{2}}               \label{eq:LeB590349}  
\end{align}
where  
\begin{subnumcases}
{ \mathbbm{1}_{\Ho}  (h_{22}, h_{21}, h_{12}, h_{11}) 
 =} 
    1  &     if   \ $(h_{22}, h_{21}, h_{12}, h_{11} ) \in  \Ho$            			\non \\
     0  &  if    \  $(h_{22}, h_{21}, h_{12}, h_{11}) \notin  \Ho$ \non
\end{subnumcases}
and 
\begin{subnumcases}
{ \mathbbm{1}_{\tilde B}  (g_0, g_1) 
 =} 
    1  &     if   \ $(g_0, g_1) \in  \tilde B$            			\non \\
     0  &  if    \  $(g_0, g_1) \notin  \tilde B $ \non
\end{subnumcases} 
and \eqref{eq:LeB4859028} is  from  \eqref{eq:LeBmeasure}.  
\end{proof}

Lemma~\ref{lm:distance5723946} suggests that, the minimum distance $d_{\min}$  defined in \eqref{eq:minidis737539} is sufficiently large for almost all  the channel  coefficients when $P$ is large.
Let us focus on the channel  coefficients  not in the outage set   $\Ho$ and  rewrite the observation $y_1$ in  \eqref{eq:rewrittenIC321y1} as
\begin{align}
 y_1    =   &  P^{{ \frac{ \alpha/2 + \epsilon }{2}}}  \cdot \gamma \cdot x_{s}   +   \sqrt{P^{ 1 - \alpha}} \tilde{g}  +  z_{1}     \label{eq:yk182375}  
\end{align}
where $x_{s}  \defeq   \sqrt{P^{ 1- \alpha }}  g_0 q_0 + g_1 q_1$
and $\tilde{g} \defeq h_{11} v_{1,p}   +  \sqrt{P^{ (\tau+2)\alpha-\tau -1}}  \delta_{2, \tau}   h_{12}  u + \sqrt{P^{  \alpha -1}} h_{12} v_{2,p}$. It is true that   
\[ |\tilde{g} | \leq  \tilde{g}_{\max} \defeq \frac{1}{\tau \cdot 2^{\tau-1}} +\frac{2}{\tau},  \quad  \forall   \tilde{g}. \] 
For the observation in \eqref{eq:yk182375}, we will  decode $x_{s}$ by considering other signals as noise (called as noise removal) and then recover $q_0$  and $q_1$ from $x_{s} $  by using the  rational independence between $g_0$ and $g_1$ (called as signal separation, see~\cite{MGMK:14}). 
Given the channel coefficients  outside the outage set   $\Ho$,  Lemma~\ref{lm:distance5723946} suggests that, the minimum distance for $x_{s}$ satisfies $d_{\min}    \geq \kappa P^{- \frac{3\alpha/2 -1 }{2}}$. 
With this result,  the  probability of error for the estimation of  $x_{s}$ from $y_1$  is bounded by
 \begin{align}
    \text{Pr} [ x_s \neq \hat{x}_s ]  
   \leq  &  \text{Pr} \Bigl[   | z_1  +  P^{ \frac{1 - \alpha}{2}} \tilde{g}|  >   P^{{ \frac{ \alpha/2 + \epsilon }{2}}}  \cdot \gamma   \cdot \frac{d_{\text{min}} }{2}  \Bigr]   \non \\
    \leq    & 2  \cdot     {\bf{Q}} \bigl(   P^{{ \frac{ \alpha/2 + \epsilon }{2}}}  \cdot  \gamma   \cdot \frac{d_{\text{min}} }{2}   -  P^{ \frac{1 - \alpha}{2}} \tilde{g}_{\max}  \bigr)     \label{eq:error2256cc}   \\ 
     \leq  &  2  \cdot     {\bf{Q}} \bigl(   P^{ \frac{1 - \alpha}{2}} ( \frac{\gamma \kappa  P^{ \frac{ \epsilon}{2}}}{2}  -  \frac{1}{\tau \cdot 2^{\tau-1}} -\frac{2}{\tau})  \bigr)     \label{eq:error9982cc} 
  \end{align}
  where  ${\bf{Q}}(c )  \defeq  \frac{1}{\sqrt{2\pi}} \int_{c}^{\infty}  \exp( -\frac{ u^2}{2} ) d u$ and 
\eqref{eq:error2256cc} use the fact that   $ |\tilde{g} | \leq  \tilde{g}_{\max} \defeq  \frac{1}{\tau \cdot 2^{\tau-1}} +\frac{2}{\tau},    \forall   \tilde{g}$; and the last step uses the result of $d_{\min}    \geq    \kappa  P^{- \frac{3\alpha/2 -1 }{2}}  $.
From the step in \eqref{eq:error9982cc}, it implies that
  \begin{align}
 \text{Pr} [ x_s \neq \hat{x}_s ] \to 0  \quad  \text{as} \quad   P\to \infty.    \label{eq:error885256}                           
 \end{align}
 Note that $  q_0   $ and  $q_1  $  (and consequently $v_{1,c}$ and $v_{2,c}  +  u$) can be recovered from $x_s$ due to rational independence.

After decoding $x_{s}$ we can estimate  $v_{1,p}$ from the following observation 
\begin{align}
  y_1  -  P^{{ \frac{ \alpha/2 + \epsilon }{2}}}  \cdot \gamma \cdot x_{s} =     \sqrt{P^{ 1 - \alpha}}  h_{11} v_{1,p} +   \sqrt{P^{ (\tau+1)\alpha-\tau}}  \delta_{2, \tau}   h_{12}  u +    h_{12} v_{2,p}  +  z_{1} .    \non
\end{align}
Given that   $v_{1,p}, v_{2,p} \in    \Omega (\xi   =\frac{ \gamma}{2Q},   \   Q = P^{ \frac{ 1 - \alpha - \epsilon}{2}} )$  and $  \sqrt{P^{ (\tau+1)\alpha-\tau}}  \delta_{2, \tau}   h_{12}  u +    h_{12} v_{2,p}  \leq  \frac{2}{\tau} + \frac{1}{\tau \cdot 2^\tau}$,  from Lemma~\ref{lm:AWGNic}  we can conclude that  the  probability of error   for the estimation of $v_{1,p}$  satisfies
 \begin{align}
  \text{Pr} [\hat{v}_{1,p}   \neq  v_{1,p} ]     \to 0  \quad  \text{as} \quad   P\to \infty . \label{eq:error155525}                           
 \end{align}
With \eqref{eq:error885256}  and \eqref{eq:error155525}  we can conclude that 
 \begin{align}
 \text{Pr} [  \{ v_{1,c} \neq \hat{v}_{1,c} \} \cup  \{ v_{1,p} \neq \hat{v}_{1,p} \}  ]  \to 0         \quad \text {as}\quad  P\to \infty    \label{eq:error1c1p0595}
 \end{align}
 for  almost all  the channel realizations.  
For the case with  $k=2$, it is proved with the same way using the symmetry property.

\section{Proof of Lemma~\ref{lm:rateerror48912}  \label{sec:rateerror48912} }

We here prove Lemma~\ref{lm:rateerror48912}.
Due to the symmetry, we only focus on the case  of $k=1$. 
In this setting with $1/2 < \alpha \leq 2/3$,  successive decoding method will be used.
For the observation   $y_1$ described in \eqref{eq:IC1223y1}, it can be expressed in the following form
\begin{align}
 y_1  =      \sqrt{P} h_{11} v_{1,c}       +    \sqrt{P^{ \alpha}} g   +  z_{1}    \label{eq:yvk101}
\end{align}
where 
  \begin{align}
 g    \defeq         h_{12}  (   v_{2,c}  +  u) +    \sqrt{P^{ 1 - 2\alpha}} h_{11} v_{1,p}   +        \sqrt{P^{ - \alpha}}  h_{12} v_{2,p}  -   \sqrt{P^{ 2\alpha-2}}  \cdot   \frac{h^2_{12}h_{21}}{h_{11}h_{22}} u  . \non
 \end{align}
One can check that  $ |g | \leq  \frac{14}{\tau \cdot 2^\tau} $  holds true for any realization of  $g$. 
Then, Lemma~\ref{lm:AWGNic} reveals that   the error probability of the estimation of  $v_{1,c}$  is 
 \begin{align}
\text{Pr} [v_{1,c} \neq \hat{v}_{1,c}] \to 0, \quad \text {as}\quad  P\to \infty.   \label{eq:error776}
\end{align}
After that,   $v_{2,c} + u $  can be estimated from the observation below 
 \begin{align}
 y_1 -     \sqrt{P} h_{11} v_{1,c}     =       \sqrt{P^{ \alpha}}  h_{12}  (   v_{2,c}  +  u)    +    \sqrt{P^{ 1- \alpha }}  g'  +  z_{1}       \label{eq:yvk102}
\end{align}
 where   $ g'  \defeq    h_{11} v_{1,p}   +  \sqrt{P^{   \alpha -1 }}     h_{12} v_{2,p}  -   \sqrt{P^{ 4\alpha-3}}  \cdot   \frac{h^2_{12}h_{21}}{h_{11}h_{22}} u$.   Note that $v_{2,c} + u \in   2 \Omega (\xi   =  \frac{  \gamma}{Q},   Q =  P^{ \frac{ 2\alpha -1 - \epsilon }{2}} )$. One can also check that $ |g' | \leq  \frac{10}{\tau \cdot 2^\tau}$.
Let   $\hat{s}_{vu}$ be an  estimate of $s_{vu} \defeq v_{2,c} + u$. 
Then,  Lemma~\ref{lm:AWGNic} suggests  
\begin{align}
 \text{Pr} [s_{vu}  \neq  \hat{s}_{vu} ]   \to 0, \quad \text {as}\quad  P\to \infty. \label{eq:error887}
\end{align}
Similarly, after decoding $v_{2,c} + u$ we can decode  $v_{1,p} \in    \Omega (\xi   = \frac{ \gamma}{Q},   \   Q = P^{ \frac{ 1 - \alpha - \epsilon}{2}} )$ with    \begin{align}
 \text{Pr} [ v_{1,p} \neq  \hat{v}_{1,p}] &   \to 0    \quad \text {as }\quad  P \to \infty.   \label{eq:error2950}
\end{align}
With results \eqref{eq:error776} and \eqref{eq:error2950}, then we have 
\begin{align}
\text{Pr} [  \{ v_{1,c} \neq \hat{v}_{1,c} \} \cup  \{ v_{1,p} \neq \hat{v}_{1,p} \}  ]    \to   0         \quad \text {as}\quad  P\to \infty  .  \non 
\end{align}
The case with $k=2$ is also proved using the same way due the symmetry.

\section{Proof of Lemma~\ref{lm:errorcase1}  \label{sec:errorcase1} }

Here we  prove Lemma~\ref{lm:errorcase1}, considering the case of  $ 1/2 < \alpha < 1$. We will show that, given $y_1$ expressed in \eqref{eq:wiretap1y1}, we can estimate   $v_{c}$ and $v_{p}$  with vanishing error probability. 
For the expression of   $y_1$, it  can be written in the following form 
 \begin{align}
y_{1} =     \sqrt{P} h_{11}    v_{c}   +     \sqrt{P^{ 1 - \alpha }}  g'     +  z_{1}   \non
\end{align}
where $  g' \defeq     h_{11}    v_{p}    +        \sqrt{P^{ -(2\tau - 2 \tau \alpha  - \alpha)}}   \delta_{2, \tau}   h_{12}  u  $.
It is true that  \[ |g' | \leq  \frac{1 + 4^{\tau}}{(\tau+2)4^{\tau} }\] for any realizations of $g'$.   Recall that
$v_{p}  \in    \Omega (\xi   = \frac{ \gamma}{2Q},   \   Q = P^{ \frac{ 1 - \alpha - \epsilon}{2}} ) $,  $v_{c}, u \in    \Omega (\xi   =   \frac{ \gamma}{Q},   \   Q =  P^{ \frac{ \alpha  - \epsilon}{2}} ) $,  and  $\gamma \in (0,  \frac{1}{(\tau+2)4^\tau}]$.  With the result of Lemma~\ref{lm:AWGNic},  the error probability of estimating  $v_{c}$ from $y_1$  is 
 \begin{align}
\text{Pr} [v_{c} \neq \hat{v}_{c}] \to 0, \quad \text {as}\quad  P\to \infty.   \label{eq:errorvc14262}
\end{align}
Then, we remove $v_{c}$ from $y_1$ and  estimate  $v_{p}$ from the following observation
 \begin{align}
y_{1}  -    \sqrt{P} h_{11}    v_{c}  =           \sqrt{P^{ 1 - \alpha}}  h_{11} v_{p}  +      \sqrt{P^{ 2 \tau \alpha +1 -2\tau}}   \delta_{2, \tau}   h_{12}  u + z_{1}  .  \label{eq:y11492626}
\end{align}
Since 
\[ |   \delta_{2, \tau}   h_{12} u | \leq  \frac{1}{\tau+2},\]
 then by  Lemma~\ref{lm:AWGNic} the error probability of  estimating $v_{p}$ from the observation in \eqref{eq:y11492626}   is 
\begin{align}
\text{Pr} [v_{p}  \neq  \hat{v}_{p} |  v_{c} = \hat{v}_{c}]  \to 0, \quad \text {as}\quad  P\to \infty. \label{eq:errorvp02526}
\end{align}
With  \eqref{eq:errorvc14262} and \eqref{eq:errorvp02526}, it gives 
\begin{align}
\text{Pr} [  \{ v_{c} \neq \hat{v}_{c} \} \cup  \{ v_{p} \neq \hat{v}_{p} \}  ]    \to   0         \quad \text {as}\quad  P\to \infty .  \non 
\end{align}

 \section{Proof of Lemma~\ref{lm:mac1switchalpha}  \label{sec:mac1switchalphaproof} }

 Lemma~\ref{lm:mac1switchalpha} is proved here. For a MAC-WT channel,  the  channel input-output  relationship can be described as (see \eqref{eq:ICchannelGen1} and \eqref{eq:ICchannelGen2})
\begin{gather}
\begin{aligned}
&y_{1} (t) =  \sqrt{P} h_{11} x_{1} (t) +   \sqrt{P^{\alpha}} h_{12} x_{2}(t)  +z_{1} (t)  \\
&y_{2}(t)  = \sqrt{P^{\alpha}} h_{21} x_{1}(t)  +   \sqrt{P} h_{22} x_{2}(t)  +z_{2}(t).    
\end{aligned} \label{eq:MACchannelGena}
\end{gather} 
By interchanging the role of transmitter~1 and transmitter~2 in the MAC-WT channel, the channel input-output  relationship can be alternately represented as
\begin{gather}
\begin{aligned}
&y_{1} (t) =  \sqrt{{P^{'}}^{\alpha^{'}}} h^{'}_{12} x^{'}_{2} (t) +   \sqrt{P^{'}} h^{'}_{11} x^{'}_{1}(t)  +z_{1} (t)       \\
&y_{2}(t)  = \sqrt{P^{'}} h^{'}_{22} x^{'}_{2}(t)  +   \sqrt{{P^{'}}^{\alpha^{'}}} h^{'}_{21} x^{'}_{1}(t)  +z_{2}(t)    
\end{aligned} \label{eq:MACchannelGenb}
\end{gather} 
where
\begin{gather}
\begin{aligned}
h^{'}_{11} =h_{12},\quad  h^{'}_{12} = h_{11}, \quad  h^{'}_{21} =h_{22}, \quad h^{'}_{22}=h_{21} \\
x^{'}_{1}(t)  =x_{2}(t), \quad  x^{'}_{2} (t) = x_{1} (t), \quad P^{'} = P^{\alpha}, \quad \alpha^{'}= \frac{1}{\alpha} .
\end{aligned} \label{eq:eq:parachange}
\end{gather} 
Note that  the secure capacity region and  the secure GDoF region of the MAC-WT channel expressed in  \eqref{eq:MACchannelGenb} are $\bar{C} (P^{'}, \alpha^{'})$ and  $\bar{\mathcal{D}} ( \alpha^{'})$, respectively. 
Assume a scheme $\Gamma$ achieves a rate tuple $(R^{'}_1,R^{'}_2, R^{'}_c)$  in the channel expressed in \eqref{eq:MACchannelGena}, i.e.,   transmitter~1 achieves a rate $R_1= R^{'}_1$ and transmitter~2 achieves a rate $R_2= R^{'}_2$  by using common randomness rate $R_c= R^{'}_c$. Then  the same scheme $\Gamma$  achieves rates $R_1= R^{'}_2, R_2= R^{'}_1$, by  using common randomness rate $R_c= R^{'}_c$ in the channel expressed in  \eqref{eq:MACchannelGenb}, because  the channel expressed in  \eqref{eq:MACchannelGenb} can be reverted back to the channel expressed in   \eqref{eq:MACchannelGena} by interchanging the role of transmitters.  

For any tuple $(d^{'}_1, d^{'}_2, d^{'}_c)$ such that  $(d^{'}_1, d^{'}_2, d^{'}_c) \in \bar{\mathcal{D}} ( \alpha)$ in the channel expressed in \eqref{eq:MACchannelGena},  there exists  a scheme $\Gamma$ that achieves a rate tuple in the form of 
\begin{align}
\big(R_1= \frac{d^{'}_1}{2} \log P + o(\log P),  \ R_2= \frac{d^{'}_2}{2} \log P + o(\log P), \ R_c= \frac{d^{'}_c}{2} \log P + o(\log P)\big).
\end{align}
Based on the above argument, by interchanging the role of transmitter~1 and transmitter~2 in the channel expressed in \eqref{eq:MACchannelGena}, the same scheme $\Gamma$ 
  achieves a rate tuple in the form of
\begin{align}
\big(R_1= \frac{d^{'}_2}{2} \log P + o(\log P),  \ R_2= \frac{d^{'}_1}{2} \log P + o(\log P), \ R_c= \frac{d^{'}_c}{2} \log P + o(\log P)\big)
\end{align}
in  the channel expressed in  \eqref{eq:MACchannelGenb}. Then the following GDoF tuple 
\begin{align} 
\begin{pmatrix}  d_1 \\   \\ d_2  \\  \\ d_c\end{pmatrix} 
= \begin{pmatrix} 
\lim\limits_{P^{'} \to \infty}  \! \frac{ \frac{d^{'}_2}{2} \log P + o(\log P)}{ \frac{1}{2} \log P^{'}}   \\ 
\lim\limits_{P^{'} \to \infty}  \! \frac{ \frac{d^{'}_1}{2} \log P + o(\log P)}{ \frac{1}{2} \log P^{'}} \\ \lim\limits_{P^{'} \to \infty}  \! \frac{ \frac{d^{'}_c}{2} \log P + o(\log P)}{ \frac{1}{2} \log P^{'}}\end{pmatrix}
= \begin{pmatrix}   
\lim\limits_{P^{'} \to \infty}  \! \frac{\frac{1}{\alpha}\big( \frac{d^{'}_2}{2} \log P^{'}  +  \alpha o(\frac{\log P^{'}}{\alpha})\big)}{ \frac{1}{2} \log P^{'}}  \\ 
\lim\limits_{P^{'} \to \infty}  \! \frac{ \frac{1}{\alpha}  \big(  \frac{d^{'}_1}{2} \log P^{'}+  \alpha o(\frac{\log P^{'}}{\alpha})\big) }{ \frac{1}{2} \log P^{'}} \\ 
\lim\limits_{P^{'} \to \infty}  \! \frac{\frac{1}{\alpha} \big(  \frac{d^{'}_c}{2} \log P^{'}  + \alpha o(\frac{\log P^{'}}{\alpha})\big)}{ \frac{1}{2} \log P^{'}} \non
\end{pmatrix} 
= \begin{pmatrix} \frac{1}{\alpha}d^{'}_2 \\   \\ \frac{1}{\alpha}d^{'}_1  \\   \\  \frac{1}{\alpha}d^{'}_c \end{pmatrix} 
\end{align}
is achievable  in the channel expressed in  \eqref{eq:MACchannelGenb}, which implies  $\big(\frac{1}{\alpha}d^{'}_2, \frac{1}{\alpha}d^{'}_1, \frac{1}{\alpha}d^{'}_c \big) \in \bar{\mathcal{D}} ( \alpha^{'})$. Since $\alpha^{'}= \frac{1}{\alpha}$, then we get 
\[\big(\frac{1}{\alpha}d^{'}_2, \frac{1}{\alpha}d^{'}_1, \frac{1}{\alpha}d^{'}_c \big) \in \bar{\mathcal{D}} (\frac{1}{\alpha})\]
which completes the proof.

\section{Proof of Lemma~\ref{lm:mac1rateerror49293}  \label{sec:errorcasemac323} }
 Lemma~\ref{lm:mac1rateerror49293} is proved in this section, for the case with $ 0 \leq \alpha \leq \frac{2}{3}$   and $ 0 \leq B \leq (2\alpha-1)^{+}$. 
In this case, we will first estimate $v_{1, c}$ from $y_1$ expressed in \eqref{eq:macwiretap2y1casea11} based on  successive decoding method, and  we will then estimate $v_{2,c}$ and $v_{1,p}$ simultaneously  based on  noise removal and signal separation methods.

In the first step,   we rewrite $y_1$ from \eqref{eq:macwiretap2y1casea11} to the following form 
\begin{align}
 y_1  =      \sqrt{P} h_{11} v_{1,c}       +    \sqrt{P^{ 1- \alpha}} g   +  z_{1}    \label{eq:yvk101bb}
\end{align}
where  $g    \defeq        h_{11}  v_{1, p}  +  \sqrt{P^{3\alpha-2}   } \frac{ h_{12}h_{21}}{h_{22}}  v_{2, c}$. 
Since  $g$ is bounded, i.e., $ |g | \leq  \frac{9}{\tau \cdot  2^\tau}$,  from Lemma~\ref{lm:AWGNic} we can conclude that $v_{1, c}$ can be estimated from $y_1$  with vanishing error probability: 
 \begin{align}
\text{Pr} [v_{1, c} \neq \hat{v}_{1, c}] \to 0 \quad \text {as}\quad  P\to \infty.   \label{eq:errormac222}
\end{align}
In the second step    $v_{2,c}$ and $v_{1,p}$ will be estimated simultaneously  from the following observation 
 \begin{align}
y_{1}  -   \sqrt{P}  h_{11} v_{1, c}  &=       \sqrt{P^{2\alpha-1 }   } \frac{ h_{12}h_{21}}{h_{22}}  v_{2, c}    +   \sqrt{P^{ 1-   \alpha }} h_{11}  v_{1, p}  +  z_{1}   \non \\
&=  \gamma  (A_0 g_0 q_0 +A_1 g_1 q_1)  +  z_{1}    \label{eq:y1rew3935} 
\end{align}
where   
 \[ g_0\defeq   \frac{\eta_{2, c}  h_{12}h_{21}}{h_{22}},     \    g_1\defeq   \frac{1}{2}h_{11}, \  A_0 \defeq   \sqrt{P^{2\alpha-1-B +\epsilon }},  \   A_1 \defeq  \sqrt{P^{ B+\epsilon }},  \  q_0  \defeq  \frac{  \sqrt{P^{B- \epsilon}} }{\eta_{2, c} \gamma} \cdot   v_{2,c}   ,    \   q_1  \defeq  \frac{2 \sqrt{ P^{ 1- \alpha- B- \epsilon}}}{\gamma} \cdot  v_{1,p} \]
   for  $\gamma  \in \bigl(0,  \frac{1}{\tau \cdot 2^\tau}\bigr]$.
In this setting, $q_0, q_1 \in \Zc$, $|q_0| \leq   \sqrt{P^{B- \epsilon}}$,  and $|q_1| \leq    \sqrt{ P^{ 1- \alpha- B- \epsilon}}$. 
Let us define the following minimum distance 
  \begin{align}
 d_{\min}  (g_0, g_1)    \defeq    \min_{\substack{ q_0, q_0' \in \Zc  \cap \big [-  \sqrt{P^{B- \epsilon}},  \    \sqrt{P^{B- \epsilon}} \big]  \\  q_1, q_1' \in \Zc  \cap \big[-  \sqrt{ P^{ 1- \alpha- B- \epsilon}},  \    \sqrt{ P^{ 1- \alpha- B- \epsilon}}  \big]   \\  (q_0, q_1) \neq  (q_0', q_1')  }}  |   A_0   g_0  (q_0 - q_0') + A_1 g_1 (q_1 - q_1')  |     \label{eq:minidis37837}
 \end{align}
 which will be used for the analysis of  the estimation  of  $q_0$ and  $ q_1$ from the observation in \eqref{eq:y1rew3935}.
The following lemma provides a result on bounding this  minimum distance.
\begin{lemma}  \label{lm:distance573823}
For the case of $ 0 \leq \alpha \leq \frac{2}{3}$   and $ 0\leq B  \leq (2\alpha-1)^{+}$,   consider the signal design in Table~\ref{tab:wiremacpara}, \eqref{eq:constellationGsymmac393}-\eqref{eq:constellationGsymac320}  and \eqref{eq:macwt2x1casea11}-\eqref{eq:macwt2x2casea12}. Let $\epsilon >0$ and $\kappa \in (0, 1]$.  Then the minimum distance $d_{\min}$ defined in \eqref{eq:minidis37837} satisfies the following inequality 
 \begin{align}
d_{\min}    \geq   \kappa P^{ \frac{\epsilon }{2}}    \label{eq:distancegeq5235}
 \end{align}
for all  the channel  coefficients $\{h_{k\ell}\} \in (1, 2]^{2\times 2} \setminus \Ho$, where the Lebesgue measure of the  set $\Ho \subseteq (1,2]^{2\times 2}$  satisfies  
 \begin{align}
\mathcal{L}(\Ho) \leq 3584 \kappa   \cdot     P^{ - \frac{ \epsilon  }{2}} .    \label{eq:Lebb6223}
 \end{align}
 \end{lemma}
 \begin{proof}
  This proof is similar to that of Lemma~\ref{lm:distance5723946}. 
 In this case we set  $\eta \defeq 8$ and     
\[ \beta \defeq    \kappa P^{ \frac{\epsilon }{2}} ,  \quad Q_0 \defeq  2\sqrt{P^{B- \epsilon}},    \quad Q_1\defeq 2  \sqrt{ P^{ 1- \alpha- B- \epsilon}}. \]
  Recall that $g_0=   \frac{  \eta_{2, c}  h_{12}h_{21}}{h_{22}} $,  $g_1=  \frac{1}{2}h_{11}$, $A_0 =  \sqrt{P^{2\alpha-1-B +\epsilon }}$, and $  A_1 =  \sqrt{P^{ B+\epsilon }} $. 
  We also define two events $  \tilde B(q_0, q_1) $  and $\tilde B$ as in \eqref{eq:BBB776} and \eqref{eq:BBBtilde552}, respectively. 
By following the steps in \eqref{eq:LeBmeasurest}-\eqref{eq:LeBmeasure}, with  Lemma~\ref{lm:NMb2} involved,  then the Lebesgue measure of $\tilde B $  satisfies the following inequality
\begin{align}
\Lc (\tilde B )   \leq   896 \kappa   \cdot     P^{ - \frac{ \epsilon  }{2}}  . \label{eq:LeBmeasure893238}
\end{align}
With our definition,   $\tilde B$ is a collection of $(g_0, g_1)$ and can be treated as an  outage set.  
Let us   define  $\Ho$  as a set of   $(h_{11}, h_{21},  h_{12}, h_{22}) \in (1, 2]^{2\times 2}$   such that the corresponding pairs $(g_0, g_1)$ are in the outage set $\tilde B$, using the same definition as in \eqref{eq:HO4422}.
Then by following the steps in \eqref{eq:LeB38985}-\eqref{eq:LeB590349}, the Lebesgue measure of $\Ho$ can be bounded as
 \begin{align}
 \Lc (\Ho ) \leq       3584 \kappa   \cdot     P^{ - \frac{ \epsilon  }{2}}   .            \label{eq:LeB48395}  
\end{align}
\end{proof}

Lemma~\ref{lm:distance573823}  suggests that, the minimum distance $d_{\min}$  defined in \eqref{eq:minidis37837} is sufficiently large, i.e., $d_{\min}    \geq  \kappa P^{\frac{ \epsilon}{2}}$, for almost all  the channel  coefficients when $P$ is large.
Let us focus on the channel  coefficients  not in the outage set   $\Ho$. 
Let $x_{s} \defeq  A_0g_0 q_0 +A_1 g_1 q_1 $. Then it is easy to show that the error  probability for the estimation of  $x_{s}$ from the observation  in \eqref{eq:y1rew3935} is 
 \begin{align}
 \text{Pr} [ x_s \neq \hat{x}_s ] \to 0  \quad  \text{as} \quad   P\to \infty.     \label{eq:error222}                                                     
 \end{align}
Note that $v_{2,c}$ and  $v_{1,p}$ can be recovered from $x_s$ due to rational independence.  At this point,   with \eqref{eq:errormac222} and \eqref{eq:error222}, we finally conclude that 
  \begin{align}
 \text{Pr} [  \{ v_{1,c} \neq \hat{v}_{1,c} \} \cup  \{ v_{1,p} \neq \hat{v}_{1,p} \}\cup  \{ v_{2,c} \neq \hat{v}_{2,c} \}]  \to 0         \quad \text {as}\quad  P\to \infty  \non
 \end{align}
for almost all  the channel  coefficients.

\section{Proof of Lemma~\ref{lm:mac1rateerror8391028}  \label{sec:errorcasemacr2} }
Lemma~\ref{lm:mac1rateerror8391028} is proved in this section, for the case of $ 0 \leq \alpha \leq \frac{2}{3}$   and $  (2\alpha-1)^{+} <B \leq \alpha$. 
In this case,  $v_{1,c}$, $v_{1,p}$, $v_{2,m}$, and $v_{2,c}$ will be estimated from $y_{1} $ expressed in \eqref{eq:macwiretap2y1casea21}  based on  successive decoding.
From the expression  in \eqref{eq:macwiretap2y1casea21}, $y_{1} $ can be rewritten  as
\begin{align}
y_{1} 
=   \sqrt{P}  h_{11} v_{1, c}  +     \sqrt{P^{ 1-   \alpha }} g +  z_{1} \label{eq:macwiretap2y1rew383}
\end{align}
where $g$ is defined as
  \begin{align}
 g    \defeq        h_{11}  v_{1, p} +   \sqrt{P^{3\alpha-2 }   } \frac{ h_{12}h_{21}}{h_{22}}  v_{2, c}    +     \sqrt{P^{B+\alpha-1}} \frac{ h_{12}h_{21}}{h_{22}}  v_{2, m} . \non
 \end{align}
In this setting $g$ is bounded, i.e.,  $ |g | \leq  \frac{11}{\tau \cdot  2^\tau} $ for any realizations of $g$.  
From Lemma~\ref{lm:AWGNic}, it then reveals that   the error probability of decoding   $v_{1,c}$ based on $y_1$  is vanishing, i.e., 
 \begin{align}
\text{Pr} [v_{1,c} \neq \hat{v}_{1,c}] \to 0, \quad \text {as}\quad  P\to \infty.   \label{eq:error8923}
\end{align}
After decoding $v_{1,c}$,   we can estimate  $v_{1,p}$   from the following observation 
\begin{align}
y_{1}  -  \sqrt{P}  h_{11} v_{1, c} =    \sqrt{P^{ 1-   \alpha }}   h_{11}  v_{1, p} +     \sqrt{P^{ B }} g' +  z_{1} \label{eq:macwiretap2y1rew3232}
\end{align}
where $ g'    \defeq      \frac{ h_{12}h_{21}}{h_{22}}  v_{2, m}  +   \sqrt{P^{2\alpha-B-1 }   } \frac{ h_{12}h_{21}}{h_{22}}  v_{2, c}$.  It is easy to show that $g'$ is bounded, i.e.,  $ |g' | \leq \frac{5}{\tau \cdot 2^{\tau-1}}$. Again, Lemma~\ref{lm:AWGNic} suggests that  the error probability  of decoding  $v_{1,p}$  based on the observation in \eqref{eq:macwiretap2y1rew3232}  is 
 \begin{align}
\text{Pr} [v_{1,p} \neq \hat{v}_{1,p}] \to 0, \quad \text {as}\quad  P\to \infty.   \label{eq:error9013}
\end{align}
With the similar method as above (successive decoding), one can also easily show that  $v_{2, m} $ can be decoded with vanishing error probability and then $v_{2,c}$ can be decoded with vanishing error probability as well, i.e., 
\begin{align}
\text{Pr} [v_{2,m} \neq \hat{v}_{2,m}] \to 0, \quad \text {as}\quad  P\to \infty  \label{eq:error00089}
\end{align}
and
\begin{align}
\text{Pr} [v_{2,c} \neq \hat{v}_{2,c}] \to 0, \quad \text {as}\quad  P\to \infty.   \label{eq:error84928}
\end{align}
Finally we can conclude from  \eqref{eq:error8923}, \eqref{eq:error9013},  \eqref{eq:error00089} and   \eqref{eq:error84928} that
 \begin{align}
 \text{Pr} [  \{ v_{1,c} \neq \hat{v}_{1,c} \} \cup  \{ v_{1,p} \neq \hat{v}_{1,p} \}\cup  \{ v_{2,c} \neq \hat{v}_{2,c} \} \cup  \{ v_{2,m} \neq \hat{v}_{2,m} \}  ]  \to 0         \quad \text {as}\quad  P\to \infty .   
 \end{align}

\section{Proof of Lemma~\ref{lm:mac1rateerror29054}  \label{sec:errorcasemacr8325} }

We provide the proof of Lemma~\ref{lm:mac1rateerror29054} in this section for the case of $\frac{2}{3} \leq \alpha \leq 1$    and $ 0\leq B  \leq 2\alpha-1$.
The proof is divided into two sub-cases, i.e.,  $ 0\leq B  \leq 3\alpha-2$ and $ 3\alpha-2 \leq B  \leq 2\alpha-1$.

\subsection{$ 0\leq B  \leq 3\alpha-2$}
For the case of  $\frac{2}{3} \leq \alpha \leq 1$  and $ 0\leq B  \leq 3\alpha-2$, the observation  $y_1$ expressed in \eqref{eq:macwiretap2y1casea11} can be rewritten in the following form 
\begin{align}
y_{1} &=   \sqrt{P^{ 1-   \alpha }} h_{11}  v_{1, p}  +     \sqrt{P}  h_{11} v_{1, c}  +     \sqrt{P^{2\alpha-1 }   } \frac{ h_{12}h_{21}}{h_{22}}  v_{2, c}    +  z_{1}   \non\\
&=  2\gamma \sqrt{P^{\epsilon}} (g_0q_0 + \sqrt{P^{1-\alpha+B }}   g_1 q_1 + \sqrt{P^{2 \alpha -1-B}} g_2 q_2  ) + z_1
\label{eq:macwiretap2y1casea11re}  
\end{align}
where \[g_0 \defeq \frac{h_{11}}{4}, \  g_1 \defeq \frac{\eta_{1, c} h_{11}}{2}, \   g_2 \defeq \frac{ \eta_{2, c} h_{12}h_{21}}{2h_{22}} ,  \ q_0  \defeq  \frac{2\sqrt{P^{1-\alpha-\epsilon}}}{\gamma} \cdot   v_{1,p}  ,    \   q_1  \defeq  \frac{\sqrt{P^{\alpha-B-\epsilon}}}{\eta_{1, c} \gamma} \cdot  v_{1,c} , \    q_2  \defeq  \frac{ \sqrt{P^{B-\epsilon}}}{\eta_{2, c} \gamma} \cdot  v_{2,c}\]
for $\gamma \in (0,  \frac{1}{\tau \cdot 2^\tau}]$,   $1 \leq \eta_{1, c} <2$,  $1 \leq \eta_{2, c} <2$, $v_{1, p}     \in    \Omega ( \xi =  \frac{   \gamma}{2Q} ,   \   Q =  P^{ \frac{ 1-\alpha-\epsilon }{2}} )$, $ v_{1,c}     \in    \Omega ( \xi =  \frac{  \eta_{1, c} \gamma}{Q} ,   \   Q= P^{ \frac{ \alpha-B-\epsilon }{2}} )  $   and  $v_{2, c}        \in    \Omega ( \xi =  \frac{ \eta_{2, c} \gamma}{Q} ,   \   Q = P^{ \frac{ B-\epsilon}{2}} ) $.  Based on our definition, it implies that $q_0, q_1, q_2 \in \Zc$,  $|q_0| \leq \sqrt{P^{1-\alpha-\epsilon}}$, $|q_1| \leq\sqrt{P^{\alpha-B-\epsilon}}$,  and $|q_2| \leq \sqrt{P^{B-\epsilon}} $. 
Let us define the following minimum distance 
  \begin{align}
 d_{\min}  (g_0, g_1,  g_2)    \defeq    \min_{\substack{ q_0, q_0',  \in \Zc  \cap [- \sqrt{P^{1 \!-\! \alpha \!-\! \epsilon}},    \sqrt{P^{1 \!-\! \alpha \!-\! \epsilon}}]  \\  q_1, q_1',  \in \Zc  \cap [- \sqrt{P^{\alpha \! -\! B \!-\! \epsilon}},    \sqrt{P^{\alpha \! - \! B\! -\! \epsilon}}] \\  q_2, q_2' \in \Zc  \cap [- \sqrt{P^{B\! -\! \epsilon}}, \sqrt{P^{B \! -\! \epsilon}}] \\  (q_0, q_1, q_2) \neq  (q_0', q_1', q'_2)  }}  |    g_0  (q_0 \!-\! q_0') + \sqrt{P^{1\!-\! \alpha \!+\! B }}   g_1 (q_1\! -\! q_1')  + \sqrt{P^{2 \alpha \!-\! 1\!-\! B}}  g_2  (q_2 \!-\! q_2')|     \label{eq:minidis3940143}
 \end{align}
 which will be used for the analysis of  the estimation  of  $q_0$,  $ q_1$ and $q_2$ from the observation in \eqref{eq:macwiretap2y1casea11re}.
The  lemma below states a result on  bounding this  minimum distance.

\begin{lemma}  \label{lm:distance784974}
Consider the parameters $\kappa \in (0, 1]$ and  $\epsilon >0$,  and consider the signal design in Table~\ref{tab:wiremacpara}, \eqref{eq:constellationGsymmac393}-\eqref{eq:constellationGsymac320}  and  \eqref{eq:macwt2x1casea11}-\eqref{eq:macwt2x2casea12} for  the case of  $\frac{2}{3} \leq \alpha<1$    and $ 0\leq B  \leq 3\alpha-2$.  
 Then the minimum distance $d_{\min}$ defined in \eqref{eq:minidis3940143} satisfies the following inequality
\begin{align}
d_{\min}   \geq   \kappa P^{\frac{\epsilon }{2}} 
\end{align}
for all  the channel  coefficients $\{h_{k\ell}\} \in (1, 2]^{2\times 2} \setminus \Ho$, where  the Lebesgue measure of the outage set  $\Ho \subseteq (1,2]^{2\times 2}$  satisfies the following inequality
 \begin{align}
\mathcal{L}(\Ho) \leq   193536 \kappa  P^{-  \frac{\epsilon}{2}}.    \label{eq:Lebb348}
 \end{align}
 \end{lemma}

\begin{proof}
In this case we let 
\[\beta \defeq  \kappa P^{\frac{\epsilon }{2}}, \  A_1 \defeq \sqrt{P^{1-\alpha+B }} ,  \    A_2 \defeq \sqrt{P^{2 \alpha -1-B}} , \ Q_0 \defeq  2\sqrt{P^{1-\alpha-\epsilon}},    \  Q_1\defeq  2\sqrt{P^{\alpha-B-\epsilon}},   \     Q_2  \defeq 2\sqrt{P^{B-\epsilon}}, \]
for some $\epsilon >0$, $\kappa \in (0, 1]$, $1 \leq \eta_{1, c} <2$ and $1 \leq \eta_{2, c} <2$. Recall that $g_0 = \frac{h_{11}}{4}$, $ g_1 =  \frac{\eta_{1, c} h_{11}}{2}$, and $ g_2=\frac{ \eta_{2, c} h_{12}h_{21}}{2h_{22}}$. 
 We also define the following two sets
\begin{align}
   B'(q_0, q_1, q_2)  \defeq \{ (g_0, g_1, g_2)  \in (1,4]^3 :  |  g_0q_0 + A_1 g_1 q_1 + A_2 g_2 q_2    | < \beta \}   \non
\end{align}
and  
\begin{align}
B'  \defeq   \bigcup_{\substack{ q_0, q_1, q_2 \in \Zc:  \\  (q_0, q_1, q_2) \neq  0,  \\  |q_k| \leq Q_k  \ \forall k }}  B'(q_0, q_1, q_2)  .  \label{eq:setb399}
\end{align}
With the result of \cite[Lemma~14]{NM:13} we can bound  the Lebesgue measure of $B' $   as 
\begin{align}
\Lc (B' )  
\leq   504 \beta \Bigl(   \frac{2 Q_0}{A_2}  +  \frac{Q_0 \tilde{Q}_2}{A_1} +   \frac{2Q_0}{A_1}+ \frac{Q_0 \tilde{Q}_1}{A_2}   \Bigr)   \label{eq:lebmea311}
\end{align}   
where  $\tilde{Q}_1 = \min\Bigl\{Q_1,  8\cdot \frac{\max\{Q_0, A_2 Q_2\}}{A_1}\Bigr\}  = 16\sqrt{P^{3\alpha -2-B -\epsilon }} $ and  $\tilde{Q}_2 = \min\Bigl\{Q_2,  8\cdot \frac{\max\{Q_0, A_1 Q_1\}}{A_2}\Bigr\}  = 2\sqrt{P^{B-\epsilon}}$. 
 By plugging the  values of the parameters  into \eqref{eq:lebmea311}, we can easily bound $\Lc (B' )$ as
\begin{align}
\Lc (B' )  &\leq  24192 \kappa  P^{-  \frac{\epsilon}{2}} .  \label{eq: lebmea78572}
\end{align}

With our definition,   $B'$ is a collection of $(g_0, g_1, g_2) $ and can be treated as an  outage set.  
Let us   define  $\Ho \defeq \{  (h_{22}, h_{21}, h_{12}, h_{11} ) \in (1, 2]^{2\times 2} :      (g_0, g_1, g_2) \in B'  \}$  as a set of   $(h_{22}, h_{21},  h_{12}, h_{11}) \in (1, 2]^{2\times 2}$   such that the corresponding pairs $(g_0, g_1)$ are in the outage set $B'$.
Let us also define the indicator function $ \mathbbm{1}_{\Ho} (h_{22}, h_{21}, h_{12}, h_{11} )= 1 $ if  $(h_{22}, h_{21}, h_{12}, h_{11}) \in  \Ho$, else   $\mathbbm{1}_{\Ho} (h_{22}, h_{21}, h_{12}, h_{11} )= 0$;   and define another indicator  function $ \mathbbm{1}_{B'}  ( g_1, g_2) =1$ if $(g_0=\frac{g_1}{2 \eta_{1,c}}, g_1, g_2) \in  B'$, else $ \mathbbm{1}_{B'}  (g_1, g_2) =0$. Then by following the steps in \eqref{eq:LeB38985}-\eqref{eq:LeB590349}, the Lebesgue measure of $\Ho$ can be bounded as
\begin{align}
\Lc (\Ho ) & \leq   193536 \kappa  P^{-  \frac{\epsilon}{2}}  .       \label{eq:LeBmeasure323}  
\end{align}
\end{proof}

 Lemma~\ref{lm:distance784974}  suggests that, the minimum distance $d_{\min}$  defined in \eqref{eq:minidis3940143} is sufficiently large, i.e., $d_{\min}    \geq \kappa P^{\frac{\epsilon }{2}}$, for almost all  the channel  coefficients when $P$ is large.
Let us focus on the channel  coefficients  not in the outage set   $\Ho$. 
Let $x_{s} \defeq g_0q_0 + \sqrt{P^{1-\alpha+B }}   g_1 q_1 + \sqrt{P^{2 \alpha -1-B}} g_2 q_2  $. Then it is easy to show that the error  probability for the estimation of  $x_{s}$ from the observation  in \eqref{eq:macwiretap2y1casea11re} is 
 \begin{align}
 \text{Pr} [ x_s \neq \hat{x}_s ] \to 0  \quad  \text{as} \quad   P\to \infty.    \label{eq:error89123}                           
 \end{align}
 Note that $q_0, q_1$, $ q_2$ can be recovered due to the fact that $g_0, g_1, g_2$ are rationally independent. 
At this point we can conclude that 
 \begin{align}
  \text{Pr} [  \{ v_{1,c} \neq \hat{v}_{1,c} \}  \cup  \{ v_{1,p} \neq \hat{v}_{1,p} \}  \cup  \{ v_{2,c} \neq \hat{v}_{2,c} \}  ]  \to 0         \quad \text {as}\quad  P\to \infty  \label{eq:case1984}
 \end{align}
for almost all  the channel  coefficients.

\subsection{$ 3\alpha-2 \leq B  \leq 2\alpha-1$}
For this case with  $\frac{2}{3} \leq \alpha \leq 1$  and $ 3\alpha-2 \leq B  \leq 2\alpha-1$, the proof  is  similar to that of Lemma~\ref{lm:mac1rateerror49293}. We will just provide the outline of the proof in order to avoid the repetition. 
In this case, we will first estimate $v_{1, c}$ from $y_1$ expressed in \eqref{eq:macwiretap2y1casea11} based on  successive decoding method, and  then we will estimate $v_{2,c}$ and $v_{1,p}$ simultaneously  based on  noise removal and signal separation methods.

In the first step,   we rewrite $y_1$ from \eqref{eq:macwiretap2y1casea11} to the following form 
\begin{align}
y_1 =    \sqrt{P}  h_{11} v_{1, c}  +     \sqrt{P^{2\alpha-1 }   } \bar{g}+      z_{1}.
\end{align}
where $\bar{g} \defeq  \frac{ h_{12}h_{21}}{h_{22}}  v_{2, c}    +   \sqrt{P^{ 2-   3\alpha }} h_{11}  v_{1, p} $.  Since  $\bar{g}$ is bounded, i.e., $ |\bar{g} | \leq  \frac{9}{\tau \cdot  2^\tau}$,  from Lemma~\ref{lm:AWGNic} we can conclude that $v_{1, c}$ can be estimated from $y_1$  with vanishing error probability: 
 \begin{align}
\text{Pr} [v_{1,c} \neq \hat{v}_{1,c}] \to 0, \quad \text {as}\quad  P\to \infty.   \label{eq:error8329421}
\end{align}
In the second step    $v_{2,c}$ and $v_{1,p}$ will be estimated simultaneously  from the following observation 
\begin{align}
y_{1}  -  \sqrt{P}  h_{11} v_{1, c} &=     \sqrt{P^{2\alpha-1 }   }  \frac{ h_{12}h_{21}}{h_{22}}  v_{2, c}   +    \sqrt{P^{ 1-   \alpha }}   h_{11}  v_{1, p}+  z_{1} \non\\
&=  \gamma  (A_0 g_0 q_0 +A_1 g_1 q_1)  +  z_{1}     \label{eq:macwiretap2y122322} 
\end{align}
where  $g_0\defeq   \frac{ \eta_{2, c}  h_{12}h_{21}}{h_{22}}$,    $g_1\defeq   \frac{1}{2}h_{11}$,  $A_0 \defeq   \sqrt{P^{2\alpha-1-B +\epsilon }}$,  $A_1 \defeq  \sqrt{P^{2-3\alpha + B+\epsilon }}, \  q_0  \defeq  \frac{  \sqrt{P^{B- \epsilon}} }{ \eta_{2, c} \gamma}    v_{2,c}$ ,    $q_1  \defeq  \frac{2 \sqrt{ P^{ 2 \alpha-1- B- \epsilon}}}{\gamma}   v_{1,p}$.
By following the steps in \eqref{eq:y1rew3935}-\eqref{eq:error222}, one can show that the $v_{2,c}$ and $v_{1,p}$ can be estimated simultaneously from the observation in \eqref{eq:macwiretap2y122322} with vanishing error probability for almost all  the channel  coefficients.
At this point, we can conclude that   
\begin{align}
 \text{Pr} [  \{ v_{1,c} \neq \hat{v}_{1,c} \} \cup  \{ v_{1,p} \neq \hat{v}_{1,p} \}\cup  \{ v_{2,c} \neq \hat{v}_{2,c} \}]  \to 0         \quad \text {as}\quad  P\to \infty  \label{eq:case2ueu}
 \end{align}
for almost all  the channel  coefficients.

\section{Proof of Lemma~\ref{lm:mac1rateerroraleq1}  \label{sec:mac1rateerroraleq1} }
The proof of Lemma~\ref{lm:mac1rateerroraleq1} is provided in this section for the case with $  \alpha=1$. In this case $v_{1,c}$ and $v_{2,c}$ will be estimated simultaneously from $y_1$ expressed in \eqref{eq:MACeqy1} based on  noise removal and signal separation methods.  Recall that  the GDoF pair $(d^{'}_1, d^{'}_2)$ considered here satisfies the conditions of $d^{'}_1\leq 1-d^{'}_2 $ and $0 \leq d^{'}_2\leq 1$.
Let us first rewrite  $y_{1}$ expressed in \eqref{eq:MACeqy1} to the following form
\begin{align} 
y_{1} &=  \varepsilon  \sqrt{P}  h_{11} v_{1, c} +\varepsilon  \sqrt{P}  \frac{h_{12}h_{21}}{h_{22}} \cdot  v_{2, c}  +  z_{1} \non\\
& =\varepsilon  \gamma  (  A_0 g_0 q_0 +A_1 g_1 q_1  ) +  z_{1}   \label{eq:MACeqy1rew3489} 
\end{align}
where   
 \[ g_0\defeq \eta_{1, c}  h_{11},     \   g_1\defeq   \frac{\eta_{2, c}h_{12}h_{21}}{h_{22}}, \  A_0 \defeq \sqrt{P^{{ 1-d^{'}_1 + \epsilon }}}, \  A_1\defeq  \sqrt{P^{{1-d^{'}_2+ \epsilon }}}, \   q_0  \defeq  \frac{\sqrt{P^{{ d^{'}_1- \epsilon}}} }{\eta_{1, c} \gamma} \cdot   v_{1,c}   ,    \   q_1  \defeq  \frac{\sqrt{P^{ d^{'}_2- \epsilon }} }{\eta_{2, c} \gamma} \cdot  v_{2,c} \]
   for  $\gamma  \in \bigl(0,  \frac{1}{\tau \cdot 2^\tau}\bigr]$,  $1\leq \eta_{1, c}<2$,  $1\leq \eta_{2, c}<2$,
  $ v_{1,c}     \in    \Omega ( \xi =  \frac{  \eta_{1, c} \gamma}{Q} ,   \    Q = P^{ \frac{   d^{'}_1-\epsilon }{2}}) $ and 
  $ v_{2, c}       \in    \Omega ( \xi =  \frac{\eta_{2, c} \gamma}{Q} ,   \   Q = P^{ \frac{ d^{'}_2-\epsilon }{2}} )   $.
In this setting, $q_0, q_1 \in \Zc$, $|q_0| \leq \sqrt{P^{{ d^{'}_1- \epsilon}}}$,  $|q_1| \leq   \sqrt{P^{ d^{'}_2- \epsilon }} $. 
Let us define the following minimum distance 
  \begin{align}
 d_{\min}  (g_0, g_1)    \defeq    \min_{\substack{ q_0, q_0' \in \Zc  \cap  [- \sqrt{P^{{ d^{'}_1- \epsilon}}},  \   \sqrt{P^{{ d^{'}_1- \epsilon}}}]  \\  q_1, q_1' \in \Zc  \cap [- \sqrt{P^{ d^{'}_2- \epsilon }} ,  \   \sqrt{P^{ d^{'}_2- \epsilon }} ]  \\  (q_0, q_1) \neq  (q_0', q_1')  }}  |  A_0  g_0  (q_0 - q_0') +  A_1 g_1 (q_1 - q_1')  |     \label{eq:minidis47889}
 \end{align}
 which will be used for the analysis of  the estimation  of  $q_0$ and  $ q_1$ from  $y_1$ expressed in \eqref{eq:MACeqy1rew3489}.
The  lemma below states a result on bounding this  minimum distance.

\begin{lemma}  \label{lm:distanceALEQ1}
For  $  \alpha =1$, consider the signal design in \eqref{eq:MACeqx1}-\eqref{eq:constellationalpha3}. Let $\epsilon >0$ and $\kappa \in (0, 1]$.  Then the minimum distance $d_{\min}$ defined in \eqref{eq:minidis47889} satisfies the following inequality 
 \begin{align}
d_{\min}    \geq   \kappa P^{\frac{ \epsilon}{2}}    \label{eq:distancegeq28932}
 \end{align}
for all  the channel  coefficients $\{h_{k\ell}\} \in (1, 2]^{2\times 2} \setminus \Ho$, where the Lebesgue measure of the  set $\Ho \subseteq (1,2]^{2\times 2}$  satisfies  
 \begin{align}
\mathcal{L}(\Ho) \leq 1792 \kappa   \cdot     P^{ - \frac{ \epsilon  }{2}} .    \label{eq:Lebb2111}
 \end{align}
 \end{lemma}
\begin{proof}
 This proof is similar to that of Lemma~\ref{lm:distance5723946}.
 In this case  we set  $\eta \defeq 8$ and
\[ \beta \defeq    \kappa P^{\frac{ \epsilon}{2}}  , \quad   Q_0 \defeq 2\sqrt{P^{{ d^{'}_1- \epsilon}}},    \quad Q_1\defeq 2  \sqrt{P^{ d^{'}_2- \epsilon }} . \]
Recall that $g_0=   \eta_{1, c} h_{11}$,  $g_1=  \frac{\eta_{2, c} h_{12}h_{21}}{h_{22}} $,  $A_0 =  \sqrt{P^{{1- d^{'}_1 + \epsilon }}}$, and $A_1=  \sqrt{P^{{1-d^{'}_2+ \epsilon }}}$. 
We also define two events $  \tilde B(q_0, q_1) $  and $\tilde B$ as in \eqref{eq:BBB776} and \eqref{eq:BBBtilde552}, respectively. 
By following the steps in \eqref{eq:LeBmeasurest}-\eqref{eq:LeBmeasure}, with  Lemma~\ref{lm:NMb2} involved,  then the Lebesgue measure of $\tilde B $  satisfies the following inequality
\begin{align}
\Lc (\tilde B )   \leq   896 \kappa   \cdot     P^{ - \frac{ \epsilon  }{2}}  . \label{eq:LeBmeasure122u2u}
\end{align}
With our definition,   $\tilde B$ is a collection of $(g_0, g_1)$ and can be treated as an  outage set.  
Let us   define  $\Ho$  as a set of   $(h_{11}, h_{21},  h_{12}, h_{22}) \in (1, 2]^{2\times 2}$   such that the corresponding pairs $(g_0, g_1)$ are in the outage set $\tilde B$, using the same definition as in \eqref{eq:HO4422}.
Then by following the steps in \eqref{eq:LeB38985}-\eqref{eq:LeB590349}, the Lebesgue measure of $\Ho$ can be bounded as
 \begin{align}
 \Lc (\Ho ) \leq       1792 \kappa   \cdot     P^{ - \frac{ \epsilon  }{2}}   .            \label{eq:LeBaleq1}  
\end{align}
\end{proof}

Lemma~\ref{lm:distanceALEQ1} suggests that, the minimum distance $d_{\min}$  is sufficiently large, i.e., $d_{\min}    \geq  \kappa P^{\frac{ \epsilon}{2}}$, for almost all  the channel  coefficients in the regime of large $P$.
Let us focus on the channel  coefficients  not in the outage set   $\Ho$. 
Let $x_{s} \defeq  A_0g_0 q_0 +A_1 g_1 q_1 $. Then it is easy to show that the error  probability for decoding  $x_{s}$ based on $y_1$  (see \eqref{eq:MACeqy1rew3489}) is  vanishing, i.e.,
 \begin{align}
 \text{Pr} [ x_s \neq \hat{x}_s ] \to 0  \quad  \text{as} \quad   P\to \infty.    \label{eq:erroraleq111}                           
 \end{align}
Since $v_{1,c}$ and  $v_{2,c}$ can be recovered from $x_s$, due to rational dependence,  we finally conclude that 
 \begin{align}
 \text{Pr} [  \{ v_{1,c} \neq \hat{v}_{1,c} \} \cup  \{ v_{2,c} \neq \hat{v}_{2,c} \}]  \to 0         \quad \text {as}\quad  P\to \infty  \non
 \end{align}
for almost all  the channel  coefficients.

%\bibliographystyle{IEEEtran}
%\bibliography{IEEEabrv,final_refs}

\begin{thebibliography}{10}
\providecommand{\url}[1]{#1}
\csname url@samestyle\endcsname
\providecommand{\newblock}{\relax}
\providecommand{\bibinfo}[2]{#2}
\providecommand{\BIBentrySTDinterwordspacing}{\spaceskip=0pt\relax}
\providecommand{\BIBentryALTinterwordstretchfactor}{4}
\providecommand{\BIBentryALTinterwordspacing}{\spaceskip=\fontdimen2\font plus
\BIBentryALTinterwordstretchfactor\fontdimen3\font minus
  \fontdimen4\font\relax}
\providecommand{\BIBforeignlanguage}[2]{{%
\expandafter\ifx\csname l@#1\endcsname\relax
\typeout{** WARNING: IEEEtran.bst: No hyphenation pattern has been}%
\typeout{** loaded for the language `#1'. Using the pattern for}%
\typeout{** the default language instead.}%
\else
\language=\csname l@#1\endcsname
\fi
#2}}
\providecommand{\BIBdecl}{\relax}
\BIBdecl

\bibitem{Shannon:49}
C.~E. Shannon, ``Communication theory of secrecy systems,'' \emph{Bell System
  Technical Journal}, vol.~28, no.~4, pp. 656 -- 715, Oct. 1949.

\bibitem{Wyner:75}
A.~D. Wyner, ``The wire-tap channel,'' \emph{Bell System Technical Journal},
  vol.~54, no.~8, pp. 1355 -- 1378, Jan. 1975.

\bibitem{LMSY:08}
R.~Liu, I.~Maric, P.~Spasojevic, and R.~D. Yates, ``Discrete memoryless
  interference and broadcast channel with confidential messages: secrecy rate
  regions,'' \emph{IEEE Trans. Inf. Theory}, vol.~54, no.~6, pp. 2493 -- 2507,
  Jun. 2008.

\bibitem{LBPSV:09}
Y.~Liang, A.~{Somekh-Baruch}, H.~V. Poor, S.~Shamai, and S.~Verdu, ``Capacity
  of cognitive interference channels with and without secrecy,'' \emph{IEEE
  Trans. Inf. Theory}, vol.~55, no.~2, pp. 604 -- 619, Feb. 2009.

\bibitem{KGLP:11}
O.~O. Koyluoglu, H.~{El Gamal}, L.~Lai, and H.~V. Poor, ``Interference
  alignment for secrecy,'' \emph{IEEE Trans. Inf. Theory}, vol.~57, no.~6, pp.
  3323 -- 3332, Jun. 2011.

\bibitem{HKY:13}
X.~He, A.~Khisti, and A.~Yener, ``{MIMO} multiple access channel with an
  arbitrarily varying eavesdropper: {Secrecy} degrees of freedom,'' \emph{IEEE
  Trans. Inf. Theory}, vol.~59, no.~8, pp. 4733 -- 4745, Aug. 2013.

\bibitem{XU:14}
J.~Xie and S.~Ulukus, ``Secure degrees of freedom of one-hop wireless
  networks,'' \emph{IEEE Trans. Inf. Theory}, vol.~60, no.~6, pp. 3359 -- 3378,
  Jun. 2014.

\bibitem{XU:15}
------, ``Secure degrees of freedom of {$K$}-user {Gaussian} interference
  channels: {A} unified view,'' \emph{IEEE Trans. Inf. Theory}, vol.~61, no.~5,
  pp. 2647 -- 2661, May 2015.

\bibitem{GTJ:15}
C.~Geng, R.~Tandon, and S.~A. Jafar, ``On the symmetric 2-user deterministic
  interference channel with confidential messages,'' in \emph{Proc. {IEEE}
  Global Conf. Communications {(GLOBECOM)}}, Dec. 2015.

\bibitem{MM:16}
P.~Mohapatra and C.~R. Murthy, ``On the capacity of the two-user symmetric
  interference channel with transmitter cooperation and secrecy constraints,''
  \emph{IEEE Trans. Inf. Theory}, vol.~62, no.~10, pp. 5664 -- 5689, Oct. 2016.

\bibitem{MU:16}
P.~Mukherjee and S.~Ulukus, ``{MIMO} one hop networks with no eve {CSIT},'' in
  \emph{Proc. Allerton Conf. Communication, Control and Computing}, Sep. 2016.

\bibitem{ChenIC:18}
J.~Chen, ``Secure communication over interference channel: {To} jam or not to
  jam?'' Oct. 2018, available on ArXiv: https://arxiv.org/pdf/1810.13256.pdf.

\bibitem{YTL:08}
R.~D. Yates, D.~Tse, and Z.~Li, ``Secret communication on interference
  channels,'' in \emph{Proc. {IEEE} Int. Symp. Inf. Theory {(ISIT)}}, Jul.
  2008.

\bibitem{LLPS:10}
R.~Liu, T.~Liu, H.~V. Poor, and S.~Shamai, ``Multiple-input multiple-output
  {Gaussian} broadcast channels with confidential messages,'' \emph{IEEE Trans.
  Inf. Theory}, vol.~56, no.~9, pp. 4215 -- 4227, Sep. 2010.

\bibitem{LLP:11}
R.~Liu, Y.~Liang, and H.~V. Poor, ``Fading cognitive multiple-access channels
  with confidential messages,'' \emph{IEEE Trans. Inf. Theory}, vol.~57, no.~8,
  pp. 4992 -- 5005, Aug. 2011.

\bibitem{GJ:15}
C.~Geng and S.~A. Jafar, ``Secure {GDoF} of {$K$}-user {Gaussian} interference
  channels: {When} secrecy incurs no penalty,'' \emph{IEEE Communications
  Letters}, vol.~19, no.~8, pp. 1287 -- 1290, Aug. 2015.

\bibitem{TekinYener:08d}
E.~Tekin and A.~Yener, ``The {Gaussian} multiple access wire-tap channel,''
  \emph{IEEE Trans. Inf. Theory}, vol.~54, no.~12, pp. 5747 -- 5755, Dec. 2008.

\bibitem{KG:15}
S.~Karmakar and A.~Ghosh, ``Approximate secrecy capacity region of an
  asymmetric {MAC} wiretap channel within 1/2 bits,'' in \emph{IEEE 14th
  Canadian Workshop on Information Theory}, Jul. 2015.

\bibitem{CG:18arxiv}
J.~Chen and C.~Geng, ``Optimal secure {GDoF} of symmetric {Gaussian} wiretap
  channel with a helper,'' Dec. 2018, e-print Arxiv:
  http://arxiv.org/pdf/1812.10457.pdf.

\bibitem{BSP15}
P.~Babaheidarian, S.~Salimi, and P.~Papadimitratos, ``Finite-{SNR} regime
  analysis of the {Gaussian} wiretap multiple-access channel,'' in \emph{Proc.
  Allerton Conf. Communication, Control and Computing}, Sep. 2015.

\bibitem{LYT:08}
Z.~Li, R.~D. Yates, and W.~Trappe, ``Secrecy capacity region of a class of
  one-sided interference channel,'' in \emph{Proc. {IEEE} Int. Symp. Inf.
  Theory {(ISIT)}}, Jul. 2008.

\bibitem{LP:08}
Y.~Liang and H.~V. Poor, ``Multiple access channels with confidential
  messages,'' \emph{IEEE Trans. Inf. Theory}, vol.~54, no.~3, pp. 976 -- 1002,
  Mar. 2008.

\bibitem{FW:16}
R.~Fritschek and G.~Wunder, ``Towards a constant-gap sum-capacity result for
  the {Gaussian} wiretap channel with a helper,'' in \emph{Proc. {IEEE} Int.
  Symp. Inf. Theory {(ISIT)}}, Jul. 2016, pp. 2978 -- 2982.

\bibitem{MXU:17}
P.~Mukherjee, J.~Xie, and S.~Ulukus, ``Secure degrees of freedom of one-hop
  wireless networks with no eavesdropper {CSIT},'' \emph{IEEE Trans. Inf.
  Theory}, vol.~63, no.~3, pp. 1898 -- 1922, Mar. 2017.

\bibitem{ETW:08}
R.~H. Etkin, D.~N.~C. Tse, and H.~Wang, ``Gaussian interference channel
  capacity to within one bit,'' \emph{IEEE Trans. Inf. Theory}, vol.~54,
  no.~12, pp. 5534 -- 5562, Dec. 2008.

\bibitem{CT:06}
T.~Cover and J.~Thomas, \emph{Elements of Information Theory}, 2nd~ed.\hskip
  1em plus 0.5em minus 0.4em\relax New York: Wiley-Interscience, 2006.

\bibitem{BMKarxiv:10}
G.~Bagherikaram, A.~S. Motahari, and A.~K. Khandani, ``On the secure
  degrees-of-freedom of the multiple-access-channel,'' 2010, available on
  ArXiv: https://arxiv.org/pdf/1003.0729.pdf.

\bibitem{ChenLiArxiv:18}
J.~Chen and F.~Li, ``Adding a helper can totally remove the secrecy constraints
  in interference channel,'' Dec. 2018, available on ArXiv:
  http://arxiv.org/pdf/1812.00550.pdf.

\bibitem{MGMK:14}
A.~S. Motahari, S.~O. Gharan, M.~A. {Maddah-Ali}, and A.~K. Khandani, ``Real
  interference alignment: Exploiting the potential of single antenna systems,''
  \emph{IEEE Trans. Inf. Theory}, vol.~60, no.~8, pp. 4799 -- 4810, Aug. 2014.

\bibitem{NM:13}
U.~Niesen and M.~A. {Maddah-Ali}, ``Interference alignment: {From} degrees of
  freedom to constant-gap capacity approximations,'' \emph{IEEE Trans. Inf.
  Theory}, vol.~59, no.~8, pp. 4855 -- 4888, Aug. 2013.

\end{thebibliography}


\end{document}